\documentclass[11pt,a4paper]{article}
\pdfoutput=1
\usepackage{jheppub}
\usepackage{amsfonts, amssymb, amsmath, amsthm}
\usepackage{mathrsfs}
\usepackage{graphicx}
\usepackage[utf8]{inputenc}
\usepackage[english]{babel}
\usepackage{slashed}
\usepackage{tikz}
\usepackage[indent]{parskip}

\usepackage{color}
\definecolor{dark-gray}{gray}{0.20}
\definecolor{gray}{gray}{0.30}
\definecolor{light-gray}{gray}{0.80}
\definecolor{dark-red}{rgb}{0.7,0,0}
\definecolor{dark-green}{rgb}{0.1,0.4,0}
\definecolor{dark-blue}{rgb}{0.3,0.3,0.7}
\definecolor{light-blue}{rgb}{0.8,0.8,1}
\definecolor{blue}{rgb}{0,0,1}
\definecolor{red}{rgb}{1,0,0}
\definecolor{green}{rgb}{0,1,0}

\usepackage{hyperref}
\hypersetup{
	colorlinks=true,
	linkcolor=dark-blue,
	citecolor=dark-red,
	urlcolor=dark-blue,
	linktoc=page
}

\def\cD{{\cal D}}
\def\cF{{\cal F}}

\def\cM{{\cal M}}
\def\cN{{\cal N}}
\def\cO{{\cal O}}

\def\cV{{\cal V}}

\newtheorem{theorem}{Theorem}[section]

\newtheorem{lemma}[theorem]{Lemma}
\theoremstyle{definition}
\newtheorem{definition}{Definition}[section]
\theoremstyle{remark}
\newtheorem*{remark}{Remark}

\newcommand{\be}{\begin{equation}}
\newcommand{\ee}{\end{equation}}
\newcommand{\ba}{\begin{aligned}}
\newcommand{\ea}{\end{aligned}}
\newcommand{\bea}{\begin{eqnarray}}
\newcommand{\eea}{\end{eqnarray}}

\newcommand{\mathd}{\mathrm{d}}
\newcommand{\mathe}{\mathrm{e}}
\newcommand{\mathi}{\mathrm{i}}
\newcommand{\vol}{\mathrm{vol}}

\newcommand{\e}{\epsilon}
\newcommand{\lam}{\lambda}
\newcommand{\BV}{\mathbb{V}}
\newcommand{\BP}{\mathbb{P}}
\newcommand{\BC}{\mathbb{C}}
\newcommand{\BR}{\mathbb{R}}

\newcommand{\BZ}{\mathbb{Z}}
\newcommand{\WPL}{\mathbb{WP}^1_{a,b}}
\newcommand{\VVV}{\mathrm{vol}}
\newcommand{\NN}{\bar{\nu}}
\newcommand{\AdS}{\mathrm{AdS}}

\title{Constant maps in equivariant topological strings \\ and geometric modeling of fluxes}

\author[a]{Luca Cassia}
\author[b,c]{and Kiril Hristov\note[$\dagger$]{Corresponding author.}}

\affiliation[a]{School of Mathematics and Statistics, The University of Melbourne Parkville,\\ Melbourne, VIC 3010, Australia}
\affiliation[b]{Faculty of Physics, Sofia University ``St.\ Kliment Ohridski'',\\ J. Bourchier Blvd. 5, 1164 Sofia, Bulgaria}
\affiliation[c]{INRNE, Bulgarian Academy of Sciences, Tsarigradsko Chaussee 72, 1784 Sofia, Bulgaria}

\emailAdd{khristov@phys.uni-sofia.bg}

\abstract{
\noindent We study the equivariant generalization of topological strings on toric manifolds, focusing in particular on defining the contributions of constant maps in the genus expansion of the partition function. This approach regularizes the integration over non-compact Calabi-Yau spaces, producing finite results at each order in the expansion, as illustrated by a broad set of explicit examples. Our investigation highlights the geometric modeling of flux compactifications and clarifies the link between the effective supergravity framework and the equivariant topological string formalism, building on recent developments by Martelli and Zaffaroni. We conclude that the connection between topological string theory and supergravity/field theory involves switching between geometric moduli and fluxes, shedding light on the role of ensemble averages in string theory. We propose an \emph{exact non-perturbative} holographic match with the corresponding M2-brane partition functions, which we test perturbatively at all orders in the gauge group rank $N$ in a companion paper. A special case of our proposal for vanishing flux reformulates the Ooguri--Strominger--Vafa conjecture within the equivariant topological string framework.
}
\date{\today}
\begin{document}
\maketitle

\section{Introduction}
Since its emergence as a theory of quantum gravity, string theory has reached significant milestones, particularly in the precise microscopic counting of black hole entropy \cite{Strominger:1996sh, Maldacena:1997de} and the development of the holographic correspondence \cite{Maldacena:1997re}. The field of holography has expanded dramatically, producing (among much else) numerous formal results that focus on the exact computation of field theory partition functions via supersymmetric localization, \cite{Nekrasov:2002qd, Pestun:2007rz, Pestun:2016zxk}, often linked to branes wrapping cycles on curved manifolds. In the bulk, effective supergravity analysis frequently takes center stage \cite{Gauntlett:2003di, Grana:2005jc}, but finite $N$ calculations in field theory are driving us beyond the leading two-derivative approximation. These advancements owe much to the mathematical developments of Duistermaat--Heckman \cite{Duistermaat:1982vw} and
Atiyah--Bott--Berline--Vergne (ABBV) \cite{berline1982classes, Atiyah:1984px} in equivariant localization. Parallel to holography, but evolving somewhat independently due to the direct relation to formal mathematics, the pursuit of exact quantum calculations has fueled major progress in topological string theory, \cite{Witten:1988xj, Dijkgraaf:1990dj, Witten:1992fb, Bershadsky:1993cx}. The development of the topological vertex, \cite{Aganagic:2003db,Aganagic:2003qj}, and its refinement on toric Calabi--Yau manifolds, \cite{Iqbal:2007ii,Huang:2010kf,Krefl:2010fm}, have yielded a solution for all string amplitudes. This was eventually applied holographically, resulting in impressive non-perturbative matches with Chern--Simons-matter theories, \cite{Aharony:2008ug,Kapustin:2009kz,Drukker:2011zy,Hatsuda:2013oxa}. Most recently, lower-dimensional results have sparked interest in the non-perturbative nature of holography and quantum gravity, \cite{Saad:2019lba,Penington:2019npb,Almheiri:2019psf,Penington:2019kki,Almheiri:2019qdq}, and have put forward the question of ensemble averages. The present work, along with a companion paper, \cite{Cassia:2025jkr}, addresses all the aforementioned topics in a novel manner, aiming to open a new avenue for quantum gravity calculations within string and M-theory.

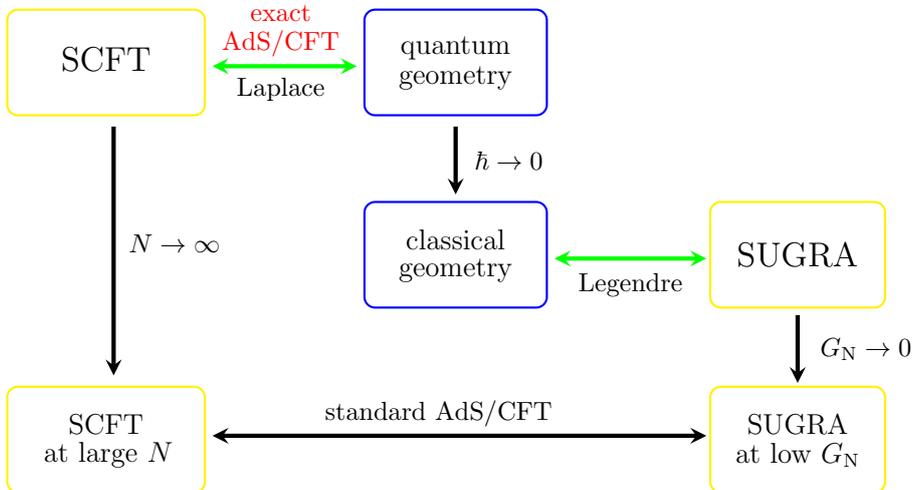
\begin{figure}[ht!]
\hspace{4em}
\begin{tikzpicture}[scale=1, every node/.style={scale=.9}]
	
	\draw[fill = white,thick, rounded corners, draw = yellow] (.1,-.75) rectangle (2.7,.65);
	\node at (1.4,-0){\Large SCFT};
	
	\draw[fill = white,thick, rounded corners, draw = blue] (4.8,-0.75) rectangle (7.2,.65);
	\node at (6,0.15){\large quantum};
	\node at (6,-0.25){\large geometry};

	\draw[fill = white,  thick, rounded corners, draw = yellow] (.1,-5.75) rectangle (2.7,-4.35);
	\node at (1.4,-4.85){\large SCFT};
	\node at (1.43,-5.25){\large at large $N$};

	\draw[fill = white,  thick, rounded corners, draw = yellow] (9.35,-5.75) rectangle (11.65,-4.35);
	\node at (10.5,-4.85){\large SUGRA};
	\node at (10.5,-5.25){\large at low $G_\mathrm{N}$};

	\draw[fill = white,  thick, rounded corners, draw = blue] (4.8,-3.3) rectangle (7.2,-1.9);
	\node at (6,-2.4){\large classical};
	\node at (6,-2.8){\large geometry};

	\draw[fill = white, thick, rounded corners, draw = yellow] (9.35,-3.3) rectangle (11.65,-1.9);
	\node at (10.5,-2.6){\Large SUGRA};

	\draw[->,>=stealth, black, ultra thick] (10.5,-3.4) -- (10.5,-4.3);
	\node at (11.4,-3.85){$G_\mathrm{N} \to 0$};	

	\draw[->,>=stealth, black, ultra thick] (6,-0.9) -- (6,-1.8);
	\node at (6.7,-1.35){$\hbar \to 0$};	

	\draw[->,>=stealth, black, ultra thick] (1.5,-0.9) -- (1.5,-4.2);
	\node at (2.3,-2.45){$N \to \infty$};	

	\draw[<->,>=stealth, black, ultra thick] (2.8,-5) -- (9.3,-5);	
	\node at (5.78,-4.7){standard AdS/CFT};

	\draw[<->,>=stealth, green, ultra thick] (7.3,-2.65) -- (9.3,-2.65);
	\node at (8.3,-3){Legendre};				

	\draw[<->,>=stealth, green, ultra thick] (2.8,-.1) -- (4.7,-.1);
	\node at (3.7,-0.4){Laplace};		
	\node at (3.7,0.6){\color{red} exact};	
	\node at (3.7,0.2){\color{red} AdS/CFT};	

\end{tikzpicture}
\caption{Conceptual sketch of {\it exact} holography from string/M-theory point of view, based on results in \cite{Martelli:2023oqk,Colombo:2023fhu,Cassia:2025jkr} and here. Blue contours signify a ``$\lam$-ensemble'' of equivariant geometric moduli, while yellow correspond to $N$ or $G_{\rm N}$-ensembles in field theory and supergravity, respectively. Green arrows show a change of ensemble, i.e.\ the relation holds via a forward or inverse transform.}
\label{fig:new}
\end{figure}

The connections between these subjects might at first seem intricate, based on the concept of \emph{exact} holography, sketched in Figure~\ref{fig:new}, and the related TS/ST correspondence, \cite{Grassi:2014zfa,Marino:2015nla,Codesido:2015dia}. The core idea that we explore here as a concrete realization of this picture is instead remarkably concrete. We propose that the equivariant generalization of topological string theory, defined on a toric manifold $X$, provides the gravitational (bulk) description of M2-branes on the underlying supersymmetric transverse background $L$. The toric manifolds are geometrically modeled after the underlying BPS backgrounds, in the simplest case as smooth resolutions of the cones over them. We thus claim that the equivariant topological string partition function on $X$ provides the \emph{exact quantum gravity dual} to the corresponding field theory partition function arising in the decoupling limit of M2-branes on $L$, upon a change of ensembles between the redundant K\"ahler moduli $\lam$ on $X$ (see below) and the fluxes $N$ that correspond to the brane charge. With certain modifications, which we outline in the concluding Section~\ref{sec:7}, these statements should apply to any brane system in string/M-theory.

Before quantifying this proposal in Subsection~\ref{sec:1.1} and Figure~\ref{fig:1}, we first give a broader introduction to the main subject of the present paper, equivariant topological string theory. More concretely, as a first step in the present work we focus our attention on the constant maps, i.e.\ the perturbative part of the topological string partition function encoding the information of classical geometry, and provide explicit results for a set of interesting examples of toric manifolds. These examples also serve as a foundation for holographically dual matches with field theory localization results at finite $N$, which we explore in \cite{Cassia:2025jkr}. We have thus divided our work into two parts: the purely mathematical aspects along with the basic holographic setup, presented here; and the applied holographically driven aspects (which for the moment remain less mathematically rigorous), detailed in \cite{Cassia:2025jkr} in relation to M2-brane partition functions.

To introduce the basic idea, we recall the general form of the topological string partition function in the A-model version, \cite{Hori:2003ic,Marino:2005sj}, defined on Calabi--Yau three-folds $X$ (CY$_3$'s) as an expansion of contributions of different genera $\mathfrak{g}$:
\be
\label{eq:1.1}
    F_X (t; g_s) = \sum_{\mathfrak{g} = 0}^\infty g_s^{2\mathfrak{g}-2}\, F_{X, \mathfrak{g}} (t)\ ,
\ee
with $t$ the (complexified) K\"ahler moduli of the manifold $X$, parametrizing the volumes of the two-cycles in $H_2(X)$, and $g_s$ the string coupling constant. Each of the partition functions $F_{X, \mathfrak{g}}$ in turn enjoys an expansion in terms of the so called Gromov--Witten invariants $N^X_{\mathfrak{g}, \beta}$ that in a certain sense count the number of holomorphic maps from a genus-$\mathfrak{g}$ Riemann surface into the manifold $X$,~\footnote{Here, we focus directly on the topological string description through Gromov--Witten theory, rather than the worldsheet perspective. See \cite{Givental:1996equ} for previous work on equivariant Gromov--Witten theory, which however does not explicitly consider the constant maps terms.}
\be
\label{eq:1.2}
 F_{X, \mathfrak{g}} (t) = \sum_\beta N^X_{\mathfrak{g}, \beta} (t)\, Q^\beta\ ,
\ee
where the power of $Q := \mathe^{-t}$ labels the corresponding instanton number (or degree of the map). 

The leading perturbative contributions to the partition function, corresponding to $\beta = 0$ above, are called \emph{constant maps terms} and enjoy some special geometric interpretation. These terms correspond to trivial holomorphic maps in Gromov--Witten theory, and we argue they play a key role in linking topological strings to supersymmetric observables. The genus $0$ constant maps term, $N^X_{0, 0}$, is proportional to the triple intersection numbers, $C_{a b c}$, of $X$ and is a homogeneous cubic function of the K\"ahler moduli, while $N^X_{1, 0}$ is a linear function proportional to the second Chern numbers of the four cycles of $X$, denoted as $b_a$. We come back to the definition of $N^X_{\mathfrak{g}, 0}$ in the next section. More formally, the constant maps terms are determined by the (classical) cohomology ring, while the instanton corrections of the genus $0$ Gromov--Witten invariants measure the deformed cohomology ring, or quantum cohomology.

An important observation is that toric Calabi--Yau manifolds, which are necessarily non-compact, allow for a generalization of the topological string partition function via equivariant cohomology.~\footnote{Note that toric CY's are also a natural setting for defining the \emph{refined} topological string partition function, which we discuss more carefully in \cite{Cassia:2025jkr}.} Crucially, in the non-compact case, the equivariant upgrade regularizes the volumes of the target manifolds which would otherwise be divergent and in turn it allows for well-defined constant maps terms.
The main message here, inspired by the work of Martelli--Zaffaroni and its predecessors, \cite{Martelli:2023oqk,Couzens:2018wnk,Gauntlett:2018dpc,Hosseini:2019use,Hosseini:2019ddy,Gauntlett:2019roi}, is that this has a clear physical meaning. In support of this claim, we can actually consider $X$ to be a K\"ahler manifold of arbitrary dimension $d$, admitting a torus action $\mathbb{T}^n$, with $n$ being the number of homogeneous coordinates. We can thus naturally extend the topological string partition function with a set of equivariant parameters $\e_i$, $i = 1, \dots, n$. They   correspond to weights of the torus action and regularize infinities in non-compact CY volumes, see \cite{Nekrasov:2021ked,Cassia:2022lfj}. This is at the heart of the work in \cite{Martelli:2023oqk}, who considered the equivariant volume of toric manifolds and its holographic relation to various wrapped brane models, see also \cite{Colombo:2023fhu}. Importantly, they make use of a set of \emph{redundant} K\"ahler parameters, $\lam^i$, that replace the \emph{effective} K\"ahler moduli $t^a$ (in the language of \cite{Ruan:2020qua} that we find appropriate to use) which allows us to look at the alternative expansion
\be
\label{eq:1.4}
 F_X (\lam, \e; g_s) = \sum_{\mathfrak{g} = 0}^\infty g_s^{2\mathfrak{g}-2}\,
 F_{X, \mathfrak{g}} (\lam, \e) = \sum_{\mathfrak{g}, \beta} g_s^{2\mathfrak{g}-2}\,
 N^X_{\mathfrak{g}, \beta} (\lam, \e)\, \tilde{Q}^\beta\ ,
\ee
where $\tilde{Q}$ are different coordinates parametrizing the instanton contributions corresponding to exponentials of $\lam$-parameters, analogously to how the $Q$ in \eqref{eq:1.2} parametrize intanton contributions in terms of the effective moduli $t$. As we explain below, using the $\lam$-parametrization is more natural from the point of view of the equivariant intersection numbers and the geometric description of flux compactifications.
Roughly speaking, one should consider the $\lam^i$ as conjugate variables to the redundant equivariant parameters $\e_i$ associated to the non-faithful action of $\mathbb{T}^n$.

An important comment is that the equivariant generalization represents a significant functional change. While the standard volume of a manifold $X$ is of a fixed degree in terms of the Kähler parameters (corresponding to its dimension), the equivariant volume is typically an exponential function with generally non-vanishing expansion coefficients for an infinite set of different degrees in the Kähler moduli. This implies that the equivariant intersection numbers are not restricted to a particular dimension. Consequently, we can (and do) use the definitions of the constant maps in terms of the triple intersection numbers for a toric manifold $X$ of \emph{any} dimension. We argue that this is the relevant construction capturing the dynamics of M2-branes in particular, while other degree intersection numbers should be relevant for the description of other string/M-theory objects, see \cite{Colombo:2023fhu}.

\subsection{Main proposal}
\label{sec:1.1}
Having built up some notation, we now state our main conjecture in terms of explicit, albeit still not fully fleshed out, formulae. We consider resolutions, denoted generally by $X$, of the cone over a $(2d -1)$-dimensional toric Sasakian manifold, $L$, with metric
\be
\label{eq:coneoverL}
 \mathd s^2(X_\circ) = \mathd r^2 + r^2\, \mathd s^2(L)\ ,
\ee
where $r$ is a coordinate on $\BR_+$ and $X_\circ = C(L)$.
The toric Sasakian condition requires that $X_\circ$ is a non-compact toric local Calabi--Yau cone of complex dimension $d$.~\footnote{More precisely, the local Calabi--Yau condition on a non-compact Kähler manifold $X_\circ$ requires that the canonical bundle be holomorphically trivial, equivalently that the first Chern class of the tangent bundle vanish in integral cohomology.
For conical complex geometries $X_\circ = C(L)$, this is equivalent to the Gorenstein condition, namely that the canonical divisor of $X_\circ$ is Cartier. 
The existence of a Sasaki--Einstein metric on $L$ constitutes an additional analytic problem.
As discussed in~\cite{Gauntlett:2006vf}, various obstructions to such metrics can occur in general, but for toric Gorenstein cones it is expected that no such obstructions arise, and a Sasaki--Einstein metric can always be realized by an appropriate choice of Reeb vector~\cite{Martelli:2006yb}.} In the simplest case when $L$ is a round sphere, S$^{2d-1}$, $X_\circ = \BC^d$ and it does not need to be resolved as it is already smooth. In the generic case $X_\circ$ exhibits a conical singularity, and $X$ denotes a smooth resolution of $X_\circ$ preserving the CY condition. The toric assumption means that the Reeb vector field $\xi$ on the space $L$ can be written as
\be
\label{eq:Reeb}
 \xi = 2\pi\sum_{i = 1}^n\, \e_i\, \partial_{\varphi_i}\ ,
\ee
where $\partial_{\varphi_i}$ generate the torus action on $X$. At the same time the parameters $\e_i$ can also be viewed as squashing parameters that parametrize the metric of the compact manifold. We give more details of the relation between the equivariant intersection numbers and the Sasakian volume in Section~\ref{sec:5} based on the work of Martelli--Sparks--Yau and its subsequent generalizations, see \cite{Martelli:2005tp,Butti:2005vn,Butti:2005ps,Martelli:2006yb,Amariti:2011uw,Couzens:2018wnk,Gauntlett:2018dpc,Hosseini:2019use,Hosseini:2019ddy,Gauntlett:2019roi,Gauntlett:2019pqg,Boido:2022mbe}.

We conjecture that the partition functions of M2-branes with exact charge $N_{\rm M2}$ (related to the brane theory gauge group rank $N$ up to a constant shift, see App.~\ref{app:A}), sitting on the tip of the cone $X_\circ$ (such that the space $L$ is the near-horizon geometry), is dictated by the full non-perturbative topological string partition function on $X$. Schematically, the $\lam^i$ parameters are dual to a set of fluxes $N_{{\rm M2}, i}$, such that the M2-brane partition function on the (possibly squashed) three-sphere, S$^3$, \cite{Kapustin:2009kz,Hama:2011ea,Imamura:2011wg}, is the Laplace transform of the topological string partition function,~\footnote{See Section~\ref{sec:7} for the description of more general brane systems within string/M-theory.}
\be
\label{eq:mainconjecture}
	Z^{\rm M2}_L (N_{\rm M2}, \e) = \int \mathd \lam\, \exp \left(  F_X (\lam, \e; g_s) - \lam\, N_{\rm M2}  \right)\ .
\ee
The expressions on the right hand side are not fully detailed here, as we omit certain subtleties that relate to the mesonic twist (to be defined in Section~\ref{sec:2.3}) and the topological string refinement. A detailed companion work, \cite{Cassia:2025jkr}, examines both sides of \eqref{eq:mainconjecture} independently and demonstrates perturbative agreement across a range of non-trivial examples, thereby substantiating our proposal. However, in essence this proposal represents a quantum generalization of the extremization proposal of \cite{Colombo:2023fhu} for the $\lam$ parameters. The extremization corresponds to the saddle point approximation of the integral transform above, including only the leading genus-zero constant maps $N_{0,0}^X$ and the leading M2-brane charge, $N$ (see App.~\ref{app:A} for the difference between $N_{\rm M2}$ and $N$). Consequently, it yields correctly the corresponding large $N$ expression for the partition function $Z_{\rm M2}$ (the right-most part of the holographic map sketched below). We should stress here our 11d point of view, where we regard $g_s$ as a formal expansion parameter that has no stringy or geometric origin, but ensures the convergence of the series. In the holographic match it is therefore allowed to be a numerical constant that fixes the normalization.

In this context, our focus on the definition of equivariant constant maps describes the complete perturbative part of the topological string and the M2-brane partition function,
\be
\label{eq:mainconjecture-pert}
 Z^\text{pert}_L (N_{\rm M2}, \e) = \int \mathd \lam\,
 \exp \left( \sum_{\mathfrak{g}\geq0} g_s^{2\mathfrak{g}-2}\,
 N^X_{\mathfrak{g}, 0} (\lam, \e) - \lam\, N_{\rm M2}  \right)\ .
\ee
We defer the explicit details on the integration and the comparison with available M2-brane localization results to \cite{Cassia:2025jkr}.

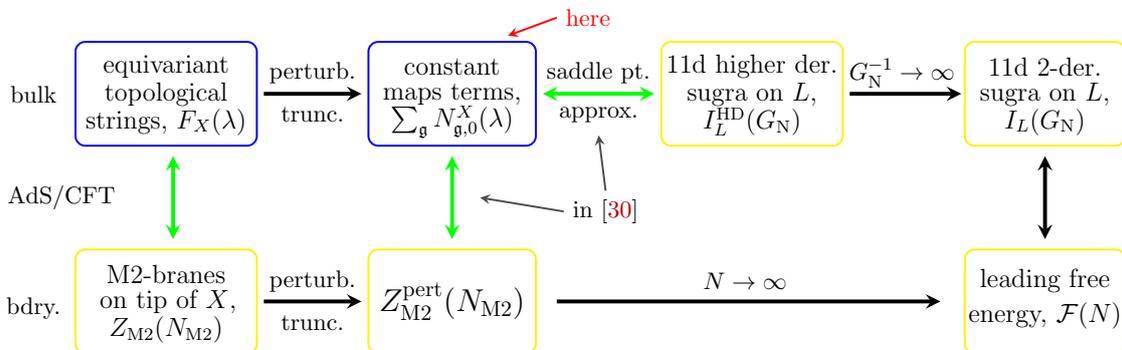
\begin{figure}[ht!]
\centering
\resizebox{\textwidth}{!}{
\begin{tikzpicture}[every node/.style={scale=.9}]
	
	\node at (-2.5,0){bulk};
	
	\node at (-2.5,-3.1){bdry.};

	\draw[fill = white,thick, rounded corners, draw = blue] (-1.9,-.75) rectangle (0.7,.75);
	\node at (-0.6,0.4){\large equivariant};
    \node at (-0.6,0){\large topological};
    \node at (-0.6,-0.4){\large strings, $ F_X (\lam)$};
	
	\draw[fill = white,thick, rounded corners, draw = blue] (2.3,-0.75) rectangle (4.7,.75);
	\node at (3.5,0.4){\large constant};
	\node at (3.5,-0){\large maps terms,};
	\node at (3.5,-0.4){\large $\sum_{\mathfrak{g}}  N^X_{\mathfrak{g}, 0} (\lam) $};

	\draw[fill = white,thick, rounded corners, draw = yellow] (6.5,-.75) rectangle (9.1,.75);
	\node at (7.8,0.4){\large higher der.};
	\node at (7.8,-0){\large sugra on $L$, };
	\node at (7.8,-0.4){\large $I_L^\text{HD} (G_\mathrm{N})$};

	\draw[fill = white,thick, rounded corners, draw = yellow] (10.85,-.75) rectangle (13.15,.75);
	\node at (12,0.4){\large 2-der.};
	\node at (12,-0){\large sugra on $L$, };
	\node at (12,-0.4){\large $I_L (G_\mathrm{N})$};

	\draw[->,>=stealth, black, ultra thick] (.8,0) -- (2.2,0);
	\node at (1.5,0.3){perturb.};
		\node at (1.5,-0.3){trunc.};

	\draw[<->,>=stealth, green, ultra thick] (4.8,0) -- (6.4,0);
	\node at (5.6,0.3){saddle pt.};
	\node at (5.6,-0.3){approx.};

	\draw[->,>=stealth, black, ultra thick] (9.2,0) -- (10.8,0);
	\node at (9.95,0.3){$G_\mathrm{N}^{-1} \to \infty$};

	\draw[fill = white,  thick, rounded corners, draw = yellow] (-1.9,-3.75) rectangle (.7,-2.25);
	\node at (-0.6,-2.6){\large M2-branes};
	\node at (-0.57,-3.0){\large on tip of $X$,};
    \node at (-0.6,-3.4){\large $Z_{\rm M2} (N_{\rm M2})$};	

	\draw[fill = white,  thick, rounded corners, draw = yellow] (2.3,-3.75) rectangle (4.7,-2.25);
	\node at (3.5,-3){\Large $Z^\text{pert}_{\rm M2} (N_{\rm M2})$};

	\draw[fill = white, thick, rounded corners, draw = yellow] (10.85,-3.75) rectangle (13.15,-2.25);
	\node at (12,-2.7){\large leading free};
	\node at (12,-3.2){\large energy, $\cF (N)$};	

	\draw[<->,>=stealth, black, ultra thick] (12,-0.9) -- (12,-2.1);

	\draw[<->,>=stealth, green, ultra thick] (3.5,-0.9) -- (3.5,-2.1);

	\draw[<->,>=stealth, green, ultra thick] (-0.5,-0.9) -- (-0.5,-2.1);
	\node at (-2.1,-1.5){AdS/CFT};

	\draw[->,>=stealth, black, ultra thick] (.8,-3) -- (2.2,-3);	
    \node at (1.5,-2.7){perturb.};
    \node at (1.5,-3.3){trunc.};

	\draw[->,>=stealth, black, ultra thick] (5,-3) -- (10.5,-3);
	\node at (7.7,-2.7){$N \to \infty$};

	\draw[->,>=stealth, red, thick] (5,1.1) -- (4.3,.84);
	\node at (5.45,1.1){{\color{red} here}};	

	\draw[->,>=stealth, gray, thick] (5.7,-1.4) -- (5.6,-0.6);
	\draw[->,>=stealth, gray, thick] (5.1,-1.7) -- (3.9,-1.5);
	\node at (5.7,-1.7){in \cite{Cassia:2025jkr}};

\end{tikzpicture}
}
\caption{Schematic diagram of the relations we propose. Blue contour signifies quantities in the $\lam$-ensemble, while yellow corresponds to the $N$ or $G_\mathrm{N}$-ensemble in field theory and supergravity, respectively ($N$ and $G_\mathrm{N}$ are related directly via AdS/CFT). Green arrows signify a change of ensemble, and we have suppressed additional indices and (equivariant) parameters.}
\label{fig:1}
\end{figure}

As a continuation of the previous relations, again initiated in \cite{Martelli:2023oqk}, we can also describe fibrations over a toric base $\Sigma$, known as a spindle, which can be interpreted as wrapped M2-brane configurations. The branes wrap cycles in the resulting space $M$, and we geometrically model this system via the CY manifold $Y$, 
\be
\label{eq:Sigmafibrations}
\begin{array}{ccc}
 L & \to & M \\
 && \downarrow \\
 && \Sigma
\end{array}
 \qquad \qquad \qquad 
\begin{array}{ccc}
 X & \to & Y \\
 && \downarrow \\
 && \Sigma
\end{array}
\ee
The resulting M2-brane partition function on S$^1 \times \Sigma$ has the interpretation of an index, known as either the twisted or the anti-twisted spindle index depending on the way supersymmetry is preserved, \cite{Benini:2015noa,Inglese:2023wky,Colombo:2024mts}. In both cases we again schematically propose the holographic relation
\be
\label{eq:mainconjecture-pert-sigma}
 Z^\text{pert}_{M} (N_{\rm M2}, \e) = \int \mathd \lam\,
 \exp\left( \sum_{\mathfrak{g}\geq0} g_s^{2\mathfrak{g}-2}\,
 N^{Y}_{\mathfrak{g}, 0} (\lam, \e) - \lam\, N_{\rm M2} \right)\ ,
\ee
which we also discuss further and verify at leading order in \cite{Cassia:2025jkr}.~\footnote{At this point, it might not be obvious that \eqref{eq:mainconjecture-pert-sigma} is on the same footing as \eqref{eq:mainconjecture-pert} from the geometric modeling point of view. We started by defining $X$ as a resolution of $C(L)$, while the fact that $Y$ is a resolution of $C(M)$ follows more subtly from the definition \eqref{eq:Sigmafibrations}. We explore the precise relation between $Y$ and $C(M)$ in Sections~\ref{sec:4} and \ref{sec:5} after defining the mesonic twist in Section~\ref{sec:2.3}.}

A specific corollary of the broader conjecture relates to effective supergravity, which serves as an alternative gravitational description of the same physical system. In fact, one can either consider M2-brane backgrounds within 11d supergravity, as done here in Sec.~\ref{sec:5} and App.~\ref{app:A}, or 4d gauged supergravity arising as a consistent truncation on the manifold $L$, as done in our companion paper, \cite{Cassia:2025jkr}.~\footnote{After the first version of this work, reference \cite{Gautason:2025plx} proposed a further refinement on the supergravity side of the diagram below. The authors of \cite{Gautason:2025plx} argue that in fact 11d supergravity observables naturally come in the $\lam$ (blue) ensemble, while 4d supergravity observables are in the $G_\mathrm{N}$ (yellow) ensemble.}
 For readers primarily focused on the topological string aspects, the emphasis on the constant maps $N^X_{\mathfrak{g}, 0} (\lam, \e)$ from the infinite topological string expansion might seem limited. However, our main objective here and in \cite{Cassia:2025jkr} is to verify as much as possible the perturbative validity of the general proposal, \eqref{eq:mainconjecture-pert} and \eqref{eq:mainconjecture-pert-sigma}. Within this scope, the constant maps provide all the essential information. Specifically, $N^X_{0, 0} (\lam, \e)$ fully captures the two derivative supergravity data (either in 11d or in 4d), as discussed at the level of Sasakian geometry in \cite{Martelli:2005tp,Butti:2005vn,Butti:2005ps,Martelli:2006yb,Amariti:2011uw,Couzens:2018wnk,Gauntlett:2018dpc,Hosseini:2019use,Hosseini:2019ddy,Gauntlett:2019roi,Gauntlett:2019pqg,Boido:2022mbe} and reviewed in Sec.~\ref{sec:5}; while we assert that $N^X_{1, 0} (\lam, \e)$ (and its refined generalization) encompasses the entire spectrum of higher-derivative (HD) corrections within effective 4d supergravity, see \cite{Bobev:2020egg,Bobev:2021oku,Hristov:2021qsw, Hristov:2022lcw}. Alternatively, from 11d perspective, our conjecture is that all HD corrections come into the (refined) genus-one constant maps, except for the topological CS-term explicitly considered in App.~\ref{app:A} responsible for the shift in the exact M2-brane charge, $N_{\rm M2}$.~\footnote{The difference between the 11d and 4d higher-derivative corrections is simply due to the fact that $N_{\rm M2}$, subject to 11d HD corrections, enters in the value of the effective 4d Newton constant, $G_{\rm 4d}$, see Sec.~\ref{sec:5}. Thus only the effective 4d supergravity fits clearly in the yellow ensemble, see the previous footnote.} 
 This approach allows us to determine the complete perturbative behavior of any background that effective supergravity can describe.~\footnote{Relating the higher-genus constant maps terms $N_{\mathfrak{g} > 1,0}$ to supergravity is more subtle and involves a Schiwnger-like calculation, see \cite{Dedushenko:2014nya}.} Note that due to the difference between full integration and saddle point approximation, the HD supergravity terms do not coincide holographically with the perturbative part of the M2-brane partition function, but only to a subset of it. A notable example are the logarithmic corrections, $\log N$, which do follow from the integration of the constant maps but do not correspond to HD corrections in supergravity, see \cite{Bhattacharyya:2012ye,Liu:2017vbl,Hristov:2021zai,Bobev:2023dwx}.

Interestingly, the above conjecture also has implications for string compactifications with vanishing flux, where $N_{\rm M2} = 0$. In such scenarios, from a 11d perspective, $L$ takes the product form S$^1 \times \mathrm{CY}_3$, with a compact CY 3-fold that is non-toric (we are excluding more general $G_2$-holonomy manifolds that preserve less supersymmetry and cannot be described by standard topological strings). However, the cone over S$^1$ corresponds to the complex plane $\BC$, which still allows us to work equivariantly. Given the structure of the full manifold $X$, we find that the constant maps $N^X_{0, 0} (\lam, \e)$ and $N^X_{1, 0} (\lam, \e)$ in this case are actually independent of $\lam$, and the proposed general relations \eqref{eq:mainconjecture-pert} and \eqref{eq:mainconjecture-pert-sigma} no longer involve an integration. However, these relations are still non-trivial, as we find that the content of \eqref{eq:mainconjecture-pert-sigma} is a reformulation of the Ooguri--Strominger--Vafa (OSV) conjecture, \cite{Ooguri:2004zv}, concerning the microscopic entropy of Maldacena--Strominger--Witten black holes \cite{Maldacena:1997de}. Consequently, we propose a derivation of the OSV conjecture in the toric direction and rephrase the remaining non-toric part of the conjecture in a more natural form, as discussed in Section~\ref{sec:6}.

Finally, we note in passing that our proposal has some interesting conceptual consequences for the discussion on ensemble averages in quantum gravity, see e.g.\ \cite{Belin:2023efa} for an introduction. Although at low energies the standard holographic dictionary relates field theory and supergravity in the same ensembles, the topological string partition function is naturally in a different ensemble as it depends on the $\lam$-parameters. Therefore the full non-perturbative holographic dictionary, which kicks-in at higher order of precision, does require a change of ensemble. From the field theory point of view, we then need to sum over different sectors at fixed value of $N_{\rm M2}$,
\be
 F_X (\lam, \e; g_s) = \log \left( \sum_{N_{\rm M2}}\,
 \mathe^{\lam\, N_{\rm M2}}\,  Z^{\rm M2}_L (N_{\rm M2}, \e)  \right)\ ,
\ee
which formally represents the inverse Laplace transform of \eqref{eq:mainconjecture}. 

\subsection*{Outline of the paper}
The rest of this paper is organized as follows. In Section~\ref{sec:2}, we give a brief overview of the topological string partition function and focus on defining the equivariant upgrade of the constant maps terms. In Sections~\ref{sec:3}, \ref{sec:3.5}, and \ref{sec:4}, we provide a comprehensive set of toric examples to build a more hands-on intuition on the subject. While many of these examples have appeared in various forms in other references, our aim is to consolidate them into a uniform notation and focus on some key points. In Section~\ref{sec:5}, we elaborate further on the relation between the toric manifolds and flux compactifications in string/M-theory.  We consider the specific case of vanishing flux backgrounds in Section~\ref{sec:6}, which provide a special example that completes the results in Sections~\ref{sec:3}--\ref{sec:4} and relates to the OSV conjecture. We conclude with some general remarks and open directions in Section~\ref{sec:7}. In Appendix~\ref{app:A} (to be read after Section~\ref{sec:5}) we calculate the exact M2-brane charge, denoted above by $N_{\rm M2}$. Finally, some additional toric fibration examples are relegated to Appendix~\ref{app:B} as they have no direct flux compactification relevance.

\subsection*{Note for mathematicians}
If one is interested in the purely mathematical side of this work, we have structured it such that all formal definitions can be found in Section~\ref{sec:2}, while explicit examples can be found in Sections~\ref{sec:3}--\ref{sec:4}. The rest of the paper contains applied string theoretic discussions that provide physical motivation.

\section{Equivariant constant maps}
\label{sec:2}

\subsection{Topological strings and intersection numbers}
Let $\phi_a\in H^2(X)$ be a basis of the degree-2 cohomology of $X$, for $X$ a complex (Kähler) $d$-fold which we assume to be compact for the moment.
We will also assume that the cohomology ring of $X$ is a polynomial ring in the classes $\phi_a$ subject to some relations (i.e.\ there is no torsion and odd-cohomology is trivial). In particular, we can decompose the Kähler form $\omega_X \equiv \omega$ over this basis of degree-2 classes as
\be
\label{eq:2.1}
 \omega = \sum_a \phi_a t^a
\ee
where $t^a$ are regarded as (real) numbers parametrizing the choice of Kähler structure in $H^2(X,\BR)$. In this sense, we can think of $t^a$ as Kähler moduli (coordinates on the moduli space of Kähler structures).

Following \cite{Iqbal:2007ii}, the general form of the topological string amplitudes $F_\mathfrak{g}(t)$, cf.\ \eqref{eq:1.1}, is given in the A-twisted model by the following integrals,~\footnote{Here and in the remainder of the paper we suppress the implicit label $X$ in the topological string amplitudes, which served in the introductory section to emphasize the dependence on the manifold, cf.\ \eqref{eq:1.1}-\eqref{eq:1.4}. We however keep the label $X$ for the volumes and equivariant volumes in the main parts of this section.}
\be
\label{eq:2.2}
\begin{aligned}
 F_0(t) &= \frac{1}{3!}\int_X\omega^3
 +\sum_{\beta\in H_2(X,\BZ)}\,N_{0,\beta}\,\mathe^{-\int_\beta\omega}\,,\\
 F_1(t) &=-\frac{1}{24}\int_X\omega\cup c_2(TX)
 +\sum_{\beta\in H_2(X,\BZ)}\,N_{1,\beta}\,\mathe^{-\int_\beta\omega},\\
 F_{\mathfrak{g}\geq2}(t) &= \left( \frac{(-1)^{\mathfrak{g}}}2 \int_{\overline{\cM}_\mathfrak{g}}
 c_{\mathfrak{g}-1}^3\right) \int_X c_3(TX)+\sum_{\beta\in H_2(X,\BZ)}\,N_{\mathfrak{g},\beta}\,
 \mathe^{-\int_\beta\omega}\,,
\end{aligned}
\ee
where $N_{\mathfrak{g},\beta}$ are the genuine Gromov--Witten invariants in genus $\mathfrak{g}$ and degree $\beta$, cf.\ \eqref{eq:1.2}, $\overline{\cM}_\mathfrak{g}$ is the moduli space of genus $\mathfrak{g}$ Riemann surfaces and $c_{\mathfrak{g}-1}$ is the Chern class of the Hodge bundle over $\overline{\cM}_\mathfrak{g}$. From here on, we focus the discussion on the constant maps, which are the first terms in the expressions above, also denoted $N_{\mathfrak{g}, 0}$.

Since any cohomology class in $H^\ast(X)$ can be written as a polynomial in the classes $\phi_a$, we can define topological invariants of $X$ by pairing all possible degree-$n$ polynomials with the fundamental class of $X$. We can then define numbers
\be
 C_{a_1,\dots,a_n} := \int_X \phi_{a_1}\cup\dots\cup\phi_{a_n}\ ,
\ee
which we can also interpret as intersection numbers of divisors in $X$ as follows. First, observe that to a degree-2 cohomology class $\phi_a$ we can associate a line bundle $L_a$ such that $c_1(L_a) = \phi_a$. Then one can associate a divisor $D_a$ in $X$ to each such line bundle, as the (homology class of the) zero locus of a generic section of $L_a$.
Then we have the identification
\be
 C_{a_1,\dots,a_n} = D_{a_1} \cap \dots\cap D_{a_n}\ ,
\ee
where on the r.h.s.\ we have the intersection number of the divisors $D_{a_i}$. Roughly speaking, one should regard $\phi_a$ as the Poincar\'e dual class to the divisor $D_a$.

Now we can package all of this information about intersection numbers into a generating function. To do so, we define the function
\be
\label{eq:intnum}
\begin{aligned}
 \VVV_X(t) :=& \sum_{a_1, \dots, a_d} \frac{1}{d!} C_{a_1,\dots,a_d}\, t^{a_1}\dots t^{a_d} \\
 =& \sum_{a_1, \dots, a_d} \frac{1}{d!} t^{a_1}\dots t^{a_d}
 \int_X \phi_{a_1}\cup\dots\cup\phi_{a_d} \\
 =& \int_X\frac{\omega^d}{d!}
 = \int_X \mathrm{e}^{\omega}\ ,
\end{aligned}
\ee
where in the last step we used that, from dimensional reasons, the integral over $X$ selects the order $d$ term in the expansion of the exponential of the Kähler class. It follows that this function corresponds to the symplectic volume of the manifold, which we also denote by $\VVV_X (t)$.
By construction then, the function $\VVV_X(t)$ satisfies the identity
\be
\label{eq:importantidentity}
 C_{a_1,\dots,a_d} = \frac{\partial^d \VVV_X(t)}{\partial t^{a_1}\dots \partial t^{a_d}}
 = \frac{\partial^d \VVV_X(t)}{\partial t^{a_1}\dots \partial t^{a_d}} \Big|_{t = 0}\ ,
\ee
with the latter equality, here trivial, written for a simpler comparison with the forthcoming equivariant generalization.

In the case when $X$ is a 3-fold, we obtain that $\VVV_X(t)$ is a polynomial of degree 3 in the Kähler moduli, which correspond to the genus-zero constant-maps term in the topological string free energy, cf.\ \eqref{eq:2.2},~\footnote{The constant maps contributions to the topological string free energy are sometimes also denoted as $F_{\mathfrak{g}}^{\rm const}(t)$.}
\be
\label{eq:g0constantmap}
	N_{0, 0} (t) = \VVV_X(t) = \frac{1}{3!}\, \sum_{a, b, c} C_{a,b,c}\, t^a t^b t^c\ .
\ee

The constant-maps term in the genus-one free energy of the topological string on a three-fold $X$ is instead given by, cf.\ \eqref{eq:2.2},
\be
\label{eq:F1int}
 N_{1, 0} (t) = \frac1{24}\int_X c_2(TX) \cup \omega
 = \frac1{24} \sum_{a} t^a \int_X c_2(TX) \cup \phi_a\ .
\ee
If we similarly expand the second Chern class $c_2(TX)$ as a degree-2 homogeneous polynomial in $\phi_a$, we can write
\be
\label{eq:chern2}
 c_2(TX)
 = \frac12 \sum_{a, b} c_2^{(a,b)}\, \phi_a\cup\phi_b\ ,
\ee
where $ c_2^{(a,b)}$ are the numerical coefficients of the expansion.
One can then define the invariants
\be
\label{eq:2ndchernnumbers}
 b_a := \int_X c_2(TX) \cup \phi_a = \frac12 \sum_{b, c} C_{a,b,c}\, c_2^{(b,c)}\ ,
\ee 
which imply that
\be
 N_{1, 0} (t)  = \frac1{24}\sum_a b_a\, t^a\ .
\ee
The integral in \eqref{eq:F1int} can therefore be regarded as the generating function of the $b_a$'s.

The higher genus constant maps, defined in \eqref{eq:2.2}, can be also rewritten in the following form, \cite{Marino:1998pg,Gopakumar:1998ii,Faber:1998gsw},
\be
\label{eq:hodge-integrals}
 N_{\mathfrak{g}\geq 2,0} (t) = N_{\mathfrak{g}\geq 2, 0}
 = (-1)^\mathfrak{g}\frac{\chi(X)}{2}\,
 \frac{|B_{2\mathfrak{g}}|}{2\mathfrak{g}} \frac{|B_{2\mathfrak{g}-2}|}{(2\mathfrak{g}-2)}
 \frac{1}{(2\mathfrak{g}-2)!}\ ,
\ee
using the Bernoulli numbers $B_n$. Importantly, the higher genus constant maps are independent of the K\"ahler moduli $t^a$, which will be important to keep in mind for the equivariant upgrade.

\subsection{Equivariant volumes, toric geometry, and generating functions}

The formula for the symplectic volume in \eqref{eq:intnum} admits a natural equivariant upgrade if we substitute the Kähler class by its equivariant completion. Namely, let $T$ be a torus of rank $n$ acting on $X$ with moment map $H:X\to\mathrm{Lie}(T)^\ast$, and let $\rho_i$ be a basis for $\mathrm{Lie}(T)^\ast$ so that $H=\rho_iH^i$.
The equivariant completion of the symplectic form $\omega$ is the combination $\omega-\e_i H^i$, with $\e_i\in H^2_T(\mathrm{pt})$ being the generators of the $T$-equivariant cohomology of the point. The equivariant upgrade of the symplectic form is closed w.r.t.\ the equivariant differential $\mathd+\e_i\iota_{v^i}$ for vector fields $v^i$ generators of the flow along the 1-parameter subgroups of $T$ associated to the generators $\rho_i$.
Then we can define the equivariant (symplectic) volume of $X$ as the integral
\be
\label{eq:equivol}
 \VVV_X (t,\e) := \int_X \mathe^{\omega-\e_iH^i} \ .
\ee
We observe that the equivariant volume is a function of the equivariant parameters $\e_i$ and it is no longer polynomial in the variables $t^a$. If the manifold $X$ is compact, then the equivariant volume is a regular function around $\e_i = 0$, such that $\vol_X (t, 0)$ is simply the symplectic volume $\VVV_X(t)$ that generates the intersection numbers, as in \eqref{eq:intnum}. If $X$ is non-compact, then the $\e_i\to 0$ limit is not smooth as the manifold does not possess a finite volume. In this sense, one can think of the equivariant generalization as a way of regularizing the volume of non-compact manifolds, as heavily explored in the physics literature starting with \cite{Nekrasov:2002qd}.

In the toric case, $X$ can be described as a K\"ahler quotient of a linear space $\BC^n$ by a torus $U(1)^r$, such that the complex dimension of $X$ is given by $d = n - r$. The action of this torus is described in terms of a matrix of integer charges $Q_i^a$ with $i=1,\dots,n$ and $a=1,\dots,r$, corresponding to an embedding of $U(1)^r$ into the larger torus $U(1)^n$ acting diagonally on the homogeneous coordinates $z_i$ on $\BC^n$. The associated moment map $\mu(z)=\rho_a Q_i^a |z_i|^2$ is the composition of the map $Q$ and the moment map $H=\rho_i H^i=\rho_i|z_i|^2$ for the $U(1)^n$-action, where $\rho_a$ and $\rho_i$ are a basis of $\mathrm{Lie}(U(1)^r)^\ast\cong\BR^r$ and $\mathrm{Lie}(U(1)^n)^\ast\cong\BR^n$, respectively.

The (ordinary) cohomology of the quotient space $X:=\BC^n//U(1)^r=\mu^{-1}(t)/U(1)^r$ is generated by tautological classes $\phi_a\in H^\ast(X)$, one for each of the $U(1)$ factors in the torus $U(1)^r$.
Moreover, since the original $U(1)^n$-action on $\BC^n$ descends to a (not necessarily faithful) action on the quotient space $X$, we can also consider the $U(1)^n$-equivariant cohomology $H^\ast_{U(1)^n}(X)$ as a module over the $U(1)^n$-equivariant cohomology of the point $H^\ast_{U(1)^n}(\mathrm{pt})\cong\BC[\e_1,\dots,\e_n]$, which is isomorphic to a polynomial ring in the equivariant parameters $\e_i$.

The Chern classes dual to the toric divisors $D_i=(\{z_i=0\}\cap\mu^{-1}(t))/U(1)^r\subset X$ admit an equivariant upgrade as equivariant Chern classes
\be
\label{eq:equiv-chern-roots}
 x_i := \e_i+\phi_aQ^a_i \in H^\ast_{U(1)^n}(X)\ .
\ee
Equivariant localization à la Jeffrey--Kirwan (JK) then allows to compute the equivariant volume in \eqref{eq:equivol} via a residue integral of the form
\be
\label{eq:altequivol}
 \VVV_X(t,\e) = \oint_\text{JK} \prod_{a=1}^r\frac{\mathd\phi_a}{2\pi\mathi}
 \frac{\mathe^{\phi_a t^a}}{\prod_{i=1}^n x_i}\ ,
\ee
where the choice of contour depends on the \emph{chamber} for the K\"ahler parameters $t^a$, also known as \emph{phase} of the symplectic quotient \cite{Goldin:2003dis}.
Evaluation of the residues at each of the JK-poles allows to recover the contributions of each fixed-point as in the ABBV localization formula.

Again, we can regard $\VVV_X(t,\e)$ as a generating function of certain geometric quantities, namely, the equivariant integrals over $X$ of all polynomial functions in the tautological classes $\phi_a$.
An alternative way to construct a generating function of equivariant integrals is as follows.
\begin{definition} Let $X$ be a toric quotient as before, with equivariant Chern roots $x_i$ as in \eqref{eq:equiv-chern-roots}. Let $\lam^i$ be a set of $n$ formal variables conjugate to the classes $x_i$, then we define the generating function of \emph{equivariant intersection numbers} as
\be
\label{eq:identity-tlam}
 \BV_X(\lam,\e) := \oint_\text{JK} \prod_{a=1}^r\frac{\mathd\phi_a}{2\pi\mathi}
 \frac{\mathe^{x_i\lam^i}}{\prod_{i=1}^n x_i}
\ee
which we regard as a formal power series in the $\lam^i$'s with coefficients in the field of fractions of the $U(1)^n$-equivariant cohomology ring of the point.
\end{definition}
\begin{remark}
The formal parameters $\lam^i$'s are not directly related to the K\"ahler moduli of $X$, in fact one should think of them as conjugate variables to all of the isometries of the prequotient $\BC^n$, including those that are not gauged when reducing to $X$. In this sense, the $\lam^i$'s can be regarded as \emph{redundant K\"ahler parameters}, the same way that the $\e_i$'s can be regarded as redundant equivariant parameters.%
\footnote{Observe that the same kind of redundant K\"ahler parameters appeared before in the literature in \cite{Ruan:2020qua}, for instance, although in the context of (quantum) K-theory. A similar redundant parametrization was also employed in the definition of the equivariant volume in \cite{Martelli:2023oqk}.}
\end{remark}
In particular, the coefficients in the $\lam^i$ expansion encode all equivariant integrals of polynomials in the equivariant Chern classes $x_i$. For this reason, we define \emph{divisors operators}
\be
\label{eq:divisor-operator}
 \cD_i:=\e_i+Q_i^a\frac{\partial}{\partial t^a}\,.
\ee
such that they produce inserions of classes $x_i$ when acting on the equivariant volume $\VVV_X(t,\e)$. Similarly, we can produce the same insertions of classes $x_i$'s in the function $\BV_X(\lam,\e)$ if we act with the operators $\frac{\partial}{\partial\lam^i}$.
We then have the identity
\be
\label{eq:divisors-on-C}
 \left.\cD_i \VVV_X(t,\e)\right|_{t=0}
 = \left.\frac{\partial}{\partial\lam^i} \BV_X(\lam,\e)\right|_{\lam=0}
\ee
and more generally we have the following.
\begin{lemma}
The equivariant volume $\VVV_X(t,\e)$ and the generating function $\BV_X(\lam,\e)$ are related by the identity
\be
\label{eq:BVtoVol}
 \left.\mathe^{\e_i\lam^i}\VVV_X(t,\e)\right|_{t^a=Q_i^a\lam^i} = \BV_X(\lam,\e)
\ee
\end{lemma}
\begin{proof}
This follows from definitions \eqref{eq:altequivol}, \eqref{eq:identity-tlam} and \eqref{eq:equiv-chern-roots}.
\end{proof}
Compared to $\VVV(t,\e)$, the function $\BV(\lam,\e)$ does not depend on all of the parameters $\e_i$ but rather on $d = n-r$ independent linear combinations of them. This follows by observing the following.
\begin{lemma}
\label{lem:shift-eps}
For $\gamma\in\BR^n$ a vector in the image of the map $Q:\BR^r\to\BR^n$,
i.e.\ $\gamma_i=\varphi_a Q_i^a$ for some arbitrary parameters $\varphi_a$, we have
\be
\label{eq:e+gamma}
 \BV_X(\lam,\e+\gamma) = \BV_X(\lam,\e)\,.
\ee
Equivalently, if we take the limit where all $\varphi_a$'s go to zero, we can rephrase the identity \eqref{eq:e+gamma} in differential form as
\be
\label{eq:Qde}
 Q_i^a\frac{\partial}{\partial\e_i} \BV_X(\lam,\e) = 0\,.
\ee
\end{lemma}
\begin{proof}
Suppose $\gamma_i=\varphi_a Q_i^a$ for some $\varphi_a\in\BR$, then we have
\be
 (\e_i+\gamma_i)+\phi_a Q_i^a=\e_i+(\phi_a+\varphi_a)Q_i^a\,.
\ee
Then we can redefine the integration variables in \eqref{eq:identity-tlam} as
\be
 \tilde\phi_a = \phi_a+\varphi_a\ ,
\ee
which leads to the claim of the lemma, namely
\be
\ba
 \BV_X(\lam,\e+\gamma)
 &= \oint\prod_{a=1}^r\frac{\mathd\phi_a}{2\pi\mathi}
 \frac{\mathe^{(\e_i+(\phi_a+\varphi_a)Q_i^a)\lam^i}}
 {\prod_{i=1}^n(\e_i+(\phi_a+\varphi_a)Q_i^a)} \\
 &= \oint\prod_{a=1}^r\frac{\mathd\tilde\phi_a}{2\pi\mathi}
 \frac{\mathe^{(\e_i+\tilde\phi_a Q_i^a)\lam^i}}
 {\prod_{i=1}^n(\e_i+\tilde\phi_a Q_i^a)}
 = \BV_X(\lam,\e)\,.
\ea
\ee
Similarly, \eqref{eq:Qde} follows by observing that
\be
 Q_i^a\frac{\partial}{\partial\e_i}\prod_j\frac{\mathe^{x_j\lam^j}}{x_j}
 = \frac{\partial}{\partial\phi_a}\prod_j\frac{\mathe^{x_j\lam^j}}{x_j}\ ,
\ee
is a total derivative in $\phi_a$.
\end{proof}
Let us then define a map $v:\BR^n\to\BR^{n-r}$ as the cokernel of $Q$, i.e.\
\be
\label{eq:map-v}
 v^i_\alpha Q_i^a = 0\,,
\ee
then we have that the function $\BV_X(\lam,\e)$ depends only on the combinations
\be
\label{eq:nuinsteadofeps}
 \nu_\alpha(\e) := v^i_\alpha\e_i
\ee
which one can regard as equivariant parameters for the quotient torus $U(1)^n/U(1)^r\cong U(1)^{n-r}$, i.e.\ the faithfully acting torus symmetry of $X$. In particular, this implies that there exists a lift of $\BV_X(\lam,\e)$ to a function $\tilde{\BV}_X(\lam,\nu)$ such that
\be
 \left.\tilde{\BV}_X(\lam,\nu)\right|_{\nu_\alpha=v_\alpha^i\e_i} = \BV_X(\lam,\e)\,.
\ee
We observe that the function $\tilde{\BV}_X(\lam,\nu)$ coincides with the function defined in \cite[(2.29)]{Martelli:2023oqk} up to a change of sign in front of the formal variables $\lam^i$.%
\footnote{In order to compare with the formulas in \cite{Martelli:2023oqk}, we first notice that the equivariant parameters $\e_i$ in \cite{Martelli:2023oqk} are to be identified with the effective equivariant parameters $\nu_\alpha$ here, while the redundant equivariant parameters $\bar{\e}_a$ in \cite{Martelli:2023oqk} are to be identified with the $\e_i$'s here.
With this dictionary between notations, we can identify the equivariant volume $\VVV_X(t,\e)$ with the Molien--Weyl integral $\BV_{\rm MW}(t,\bar{\e})$ in \cite[(2.55)]{Martelli:2023oqk}.}

In a similar way to how Lemma~\ref{lem:shift-eps} expresses the fact that the $\e_i$'s are redundant equivariant parameters, we have a similar redundancy in the K\"ahler parameters $\lam^i$'s. As observed in \cite[(5.7)]{Martelli:2023oqk}, this leads to an identity for the function $\tilde{\BV}_X(\lam,\nu)$ with respect to (infinitesimal) shifts in the $\lam$-variables.
More generally, we have the following.
\begin{lemma}
The following identities hold true:
\be
 v^i_\alpha\cD_i\,\VVV_X(t,\e) = v^i_\alpha\e_i\,\VVV_X(t,\e)\,,
\ee
\be
\label{eq:vdlamBV}
 v^i_\alpha\frac{\partial}{\partial\lam^i}\BV_X(\lam,\e) = v^i_\alpha\e_i\,\BV_X(\lam,\e)\,,
\ee
\be
 v^i_\alpha\frac{\partial}{\partial\lam^i}\tilde{\BV}_X(\lam,\nu)
 = \nu_\alpha\,\tilde{\BV}_X(\lam,\nu)\,.
\ee
\end{lemma}
\begin{proof}
The claim follows from the constraint \eqref{eq:map-v} together with the integral representations of $\VVV_X(t,\e)$, $\BV_X(\lam,\e)$ and $\tilde{\BV}_X(\lam,\nu)$, respectively.
\end{proof}

\subsubsection{\texorpdfstring{$\mathfrak{g}=0$}{g=0}}
The next logical step after generalizing the symplectic volume is to consider the equivariant version of the standard intersection numbers, which is really the reason we introduced the function $\BV(\lam,\e)$ above.%
\footnote{From this point on, we supress the implicit label $X$ on the equivariant volumes $\VVV, \BV, \tilde \BV$ unless the underlying manifold $X$ is not clear in the context.} As discussed in \cite{Cassia:2022lfj} and \cite{Martelli:2023oqk}, we can formally identify the equivariant Chern roots $x_i$ with the differential operators \eqref{eq:divisor-operator} such that acting with $\cD_i$ on the equivariant volume $\VVV(t,\e)$ is the same as inserting $x_i$ inside the integral, 
\be
 \cD_{i_1}\cdots\cD_{i_m} \VVV(t,\e)
 = \int_X \mathe^{\omega-\e_i H^i}\, x_{i_1} \cdots x_{i_m}
 = \int_{D_{i_1}\cap\dots\cap D_{i_m}} \mathe^{\omega-\e_i H^i}\,,
\ee
which corresponds to the equivariant generalization of the intersection number of divisors $D_{i_1}\cap\dots\cap D_{i_m}$. Notice that now the dimension of $X$ does not determine the degree of the non-vanishing intersection numbers, and the integral above will generically not be a topological invariant as it depends on the K\"ahler moduli themselves, on top of the equivariant parameters $\e_i$. We can nevertheless define a generalization of the intersection numbers $C_{a_1,\dots,a_n}$, see \eqref{eq:importantidentity}, that depends only on the equivariant parameters as follows.
\begin{definition}
We define $m$-tuple equivariant intersection numbers as
\be
\label{eq:n-tuple-intersection}
 C_{i_1,\dots,i_m} (\e) := \cD_{i_1}\cdots\cD_{i_m} \VVV(t,\e) \Big|_{t=0}\ ,
\ee
or, equivalently, as
\be
 C_{i_1,\dots,i_m} (\e) = \frac{\partial^m \BV(\lam,\e)}
 {\partial\lam^{i_1}\cdots\partial\lam^{i_m}} \Big|_{\lam=0}\ ,
\ee
where we used \eqref{eq:divisors-on-C}.
\end{definition}

Now we are ready to address the question of the equivariant generalization of the genus-zero constant maps, defined in \eqref{eq:g0constantmap}.
First, we notice that the generating function of $m$-tuple equivariant intersection numbers on $X$ coincides precisely with the part of $\BV(\lam,\e)$ homogeneous of degree $m$ in $\lam^i$'s, as noticed in \cite{Martelli:2023oqk},
\be
\label{eq:g0equimapn}
 \BV^{(m)}(\lam,\e)
 := \frac{1}{m!}\,\sum_{i_1,\dots,i_m} C_{i_1,\dots,i_m}(\e)\,\lam^{i_1}\cdots\lam^{i_m}\,,
 \hspace{30pt}
 \BV(\lam,\e) = \sum_{m=0}^\infty \BV^{(m)}(\lam,\e)
\ee
Moreover, this can be regarded as a generalization of the genus-zero constant-maps contribution to the GW free energy for an arbitrary toric CY $d$-fold.

Given the replacement of the K\"ahler parameters $t^a$ with the parameters $\lam^i$ in the discussion of the intersection numbers, we can define the equivariant generalization of the genus-zero constant maps $N_{0,0}(t)$, cf.\ \eqref{eq:g0constantmap}, as
\be
\label{eq:g0equimap3}
 N_{0,0}(\lam,\e) := \frac{1}{3!}\sum_{i,j,k} C_{i,j,k}(\e)\,\lam^i\,\lam^j\,\lam^k
 = \BV^{(3)}(\lam,\e)\,.
\ee

\begin{remark}
It could be counterintuitive at first, but we keep this definition, using the triple intersection numbers, for a toric manifold of \emph{arbitrary} complex dimension. We argue that this is the correct choice based on explicit comparison with M2-brane partition functions, \cite{Cassia:2025jkr}, see also \cite{Martelli:2023oqk}.
\end{remark}

\subsubsection{\texorpdfstring{$\mathfrak{g}=1$}{g=1}}
For the equivariant upgrade of the genus-one function $N_{1,0}(t)$ we also have several options: the most obvious equivariant generalization (that we choose not to follow) is
\be
\label{eq:F1equiv}
 N^\text{alt.}_{1, 0} (t,\e)
 := \frac1{24}\int_X c_{d-1}(TX)_{\mathrm{equiv}}
 \cup\mathe^{\omega-\e_iH^i}\ ,
\ee
where $d=\dim_\BC(X)$ and $c_k(TX)_{\mathrm{equiv}}$ is the $k$-th Chern class of the equivariant tangent to $X$. More specifically, one can write the Chern class $c_k$ as the $k$-th elementary symmetric polynomial $e_k(x_1,\dots,x_n)=\sum_{i_1<\dots<i_k}x_{i_1}\cdots x_{i_k}$ in terms of the equivariant Chern roots $x_i$, which are now linear combinations of both the $\phi_a$ and the $\e_i$. Just like above, observe that the formula in \eqref{eq:F1equiv} is well-defined also in any dimension $d$ other than 3, even though the connection to topological string free energy might become more subtle.

For a toric quotient, $N_{1,0}(t,\e)$ can be expressed using equivariant localization as the contour integral
\be
 N^\text{alt.}_{1, 0}(t,\e)
 = \frac1{24}\prod_{a=1}^r\oint\frac{\mathd\phi_a}{2\pi\mathi} \mathe^{\phi_a t^a}
 \frac{\sum_{1\leq i_1<\dots<i_{d-1}\leq n}x_{i_1}\cdots x_{i_{d-1}}}{\prod_{i=1}^n x_i}\ ,
\ee
where $x_i=\e_i+Q_i^a\phi_a$ and $d=n-r$ is the complex dimension of the quotient. Because of equivariance again, this is not a polynomial of fixed degree in $t^a$'s, but a power series. 

Following the same logic as in the previous section, we would like to define instead a function homogeneous of degree-one (in the redundant K\"ahler moduli $\lam^i$'s) which we can regard as the generating function of equivariant intersections between one divisor $D_i$ and the (equivariant) Poincar\'e dual to the second Chern class.
We are then led to define the function
\be
\label{eq:g1equimap}
 N_{1, 0} (\lam,\e)
 := \frac1{24} \sum_{i<j}\sum_k C_{i,j,k}(\e) \lam^k
 = \frac1{24}\oint_\text{JK} \prod_{a=1}^r\frac{\mathd\phi_a}{2\pi\mathi}
 \frac{(\sum_{i<j}x_ix_j)(x_k\lam^k)}{\prod_{i=1}^n x_i}
\ee
It follows that we can extract both the genus-zero and the genus-one equivariant constant maps contributions from a single quantity: the equivariant volume $\VVV(t,\e)$, see \eqref{eq:n-tuple-intersection}.

\subsubsection{\texorpdfstring{$\mathfrak{g}\geq 2$}{g>=2}}

In higher genus, the constant-maps contribution to the topological string free energy can be computed by observing that the GW moduli space of stable maps factorizes \cite{Faber:1998gsw} as
\be
 \overline{\cM}_{\mathfrak{g}}(X,0) \cong X \times \overline{\cM}_{\mathfrak{g}}
\ee
so that the generalization of \eqref{eq:F1equiv} is
\be
 N_{\mathfrak{g},0}(t,\e)
 := \frac{(-1)^{\mathfrak{g}}}{2} \int_X c_d(TX)_{\mathrm{equiv}}
 \cup\mathe^{\omega-\e_iH^i}
 \times \int_{\overline{\cM}_{\mathfrak{g}}}c^3_{\mathfrak{g}-1}\,,
\ee
with the Hodge integrals evaluated explicitly in \eqref{eq:hodge-integrals}.
In analogy with the CY threefold case, we define
\be
\label{eq:ggeq2equimap}
 N_{\mathfrak{g},0}(\lam,\e)
 := \frac{(-1)^{\mathfrak{g}}}{2} \sum_{i<j<k} C_{i,j,k}(\e)
 \times \int_{\overline{\cM}_{\mathfrak{g}}}c^3_{\mathfrak{g}-1}\,.
\ee
First of all, we notice that this is actually independent of the parameters $\lam^i$'s.
Furthermore, we observe that when $X$ is a threefold,
\be
 \sum_{i<j<k} C_{i,j,k}(\e)
 = \chi(X)
\ee
is a constant also independent of $\e_i$'s because of cohomological degree reasons.
In general however, this function is homogeneous of degree $3-d$ w.r.t.\ a simultaneous rescaling of all the $\e_i$'s.

We again stress that the above definition, \eqref{eq:ggeq2equimap}, might be counterintuitive at first when used on a toric manifold of arbitrary dimension, but it follows the logic of using triple equivariant intersection numbers for all constant maps terms. We comment more on the possibility for considering intersections of different order at the end of this paper.

\subsubsection*{Summary of main formulae}
Here we gather the formulae that allow us to efficiently calculate the equivariant volumes and corresponding constant maps terms in the explicit examples in the following sections.

The equivariant volume, as a function of the K\"ahler moduli $t$, for a given toric manifold with a charge matrix $Q$ can be computed via the JK-residue expression
\be
 \VVV(t,\e) = \oint_\text{JK} \prod_{a=1}^r\frac{\mathd\phi_a}{2\pi\mathi}
 \frac{\mathe^{\phi_a t^a}}{\prod_{i=1}^n x_i}\ ,
\ee
and the $\lam$-parametrization can be obtained via
\be
\label{eq:maineqforcalculationglamparametrization}
 \BV(\lam,\e) = \mathe^{\lam^i\cD_i}
 \,\VVV(t,\e)\, \Big|_{t = 0}\,,
\ee
with summations over repeated indices implied.

The equivariant constant maps terms can be computed by ordinary partial derivatives w.r.t.\ the $\lam$ parameters:
\be
\label{eq:cconstmapssummary}
\begin{aligned}
N_{0,0} (\lam, \e) &= \frac16\, \sum_{i, j, k} \frac{\partial^3 \BV (\lam,  \e)}{\partial \lam^i \partial \lam^j \partial \lam^k} \Big|_{\lam = 0}\, \lam^i \lam^j \lam^k\ , \\
N_{1,0} (\lam, \e) &= \frac1{24} \sum_{i < j} \sum_k \frac{\partial^3 \BV (\lam,  \e)}{\partial \lam^i \partial \lam^j \partial \lam^k} \Big|_{\lam = 0}\, \lam^k\ , \\
N_{\mathfrak{g}>1,0} (\e) &=  \frac{|B_{2\mathfrak{g}}|}{4 \mathfrak{g}} \frac{|B_{2\mathfrak{g}-2}|}{(2\mathfrak{g}-2)}
 \frac{(-1)^\mathfrak{g}}{(2\mathfrak{g}-2)!} \sum_{i < j < k} \frac{\partial^3 \BV (\lam,  \e)}{\partial \lam^i \partial \lam^j \partial \lam^k} \Big|_{\lam = 0}\ .
\end{aligned}
\ee

\subsection{Calabi--Yau cones, K\"ahler bases, and mesonic twist}
\label{sec:2.3}

So far, we have kept the nature of the toric action general, but in fact, there are two distinct types of symmetries to consider. The geometry typically includes non-contractible two-cycles, corresponding to the number of Kähler parameters $t^a$, which can be understood as internal symmetries—referred to as baryonic symmetries—since the toric action associated to these cycles is unrelated to the global isometries of the manifold. The remaining symmetries arise from the maximal torus of the manifold's isometry group and are called mesonic symmetries, following \cite{Hosseini:2019use,Hosseini:2019ddy}. Notably, we are always effectively breaking baryonic symmetries by sitting at a specific point in the Kähler moduli space.

Based on the above discussion, we note an interesting holographic observation that, while not fully proven, has been empirically verified in several examples (see \cite{Hosseini:2019use,Hosseini:2019ddy}). The dual field theory seems to be unaware of the baryonic symmetries of the geometry. Instead, the global symmetries of the field theory on M2-branes appear to correspond solely to the mesonic symmetries, which arise from the isometries of the manifold transverse to the brane.

Assuming this holographic observation holds, achieving a correct holographic match requires implementing the so-called \emph{mesonic truncation} or \emph{mesonic twist} in the geometry, which involves removing the baryonic symmetries, introduced by the symplectic quotient. Geometrically, this corresponds precisely to going to the origin of the K\"ahler moduli space of $X$ thus contracting all non-trivial two-cycles. This operation can be regarded as the opposite of a blow-up process and as such it introduces a conical singularity, rather than being a resolution of a singular space. The geometrical setup is the following.

Let us consider a toric CY space $X$ obtained as symplectic reduction of $\BC^n$ under the action of the torus $U(1)^r$ as described in the previous sections. The precise geometry and topology of $X$ depends on both the choice of charges $Q_i^a$ as well as the choice of chamber in the extended K\"ahler cone. Different choices of chamber lead to genuinely different quotient spaces $X$. Furthermore, from the geometry of the symplectic reduction, we have that each real K\"ahler parameter $t^a$ parametrizes the volume of one of the cycles $\beta_a$ in $H_2(X,\BZ)$,
\be
 \int_{\beta_a}\omega = t^a
\ee
In turn, one can regard such parameters $t^a$ as measuring the lengths of the compact 1-dimensional legs in the toric diagram of $X$. By definition, the toric diagram is the 1-skeleton of the intersection locus of the preimage of the moment maps with the positive orthant in $\BR^n$. The CY condition on $Q_i^a$ implies that this intersection in non-compact, and very far from the origin it looks like a real $2d$-dimensional cone, $d = (n-r)$.
When we set all the K\"ahler moduli to zero, we then expect that $X$ becomes a cone where all the non-trivial homology 2-cycles have been shrunk to a single point coinciding with the tip of the cone. The base of this cone is a Sasakian manifold $L$, of real dimension $2 d-1$. Conversely, we can regard $X$ as a toric resolution of the cone over $L$. Because $t=0$ is at the origin of the extended K\"ahler moduli space, after shrinking all homology 2-cycles in $X$, we end up with a conical singular space which no longer depends on the choice of chamber (all chambers contain the origin as a boundary point).
From now on we will refer to this conical space as $X_\circ$ and we will use that
\be
 X_\circ \cong C(L)\ .
\ee
When this is the case, we can write its metric as
\be
\label{eq:conicalmetricX}
 \mathd s^2 (X_\circ) := \mathd r^2 + r^2\, \mathd s^2 (L) \ ,
\ee
where $r$ is the radial coordinate and $\mathd s^2(L)$ is a metric on the base.
Furthermore, we will assume that $L$ can be described as a S$^1$-bundle over a conformally K\"ahler manifold $B$ such that
\be
 \mathd s^2 (L) = \eta^2 + \mathd s^2 (B)\ ,
\ee
where $\eta$ is the contact 1-form on $L$ dual to the Reeb vector $\xi$ (see \cite{Martelli:2005tp,Martelli:2006yb} and references therein for more details),
\be
 \iota_\xi \eta = 1\,.
\ee
In the following, we will also assume that the base $B$ can be described as follows. Within the base of the toric fibration of the cone $X_\circ$, we define a compact sublocus corresponding to the intersection with an hypersurface associated to an equation of the form
\be
\label{eq:moment-map-B}
 \sum_i b_i |z_i|^2 = c\ ,
\ee
where $b_i$ are some real parameters and $z_i$ are the homogeneous coordinates on $X_\circ$.
If the $b_i$ are integers, then we can regard \eqref{eq:moment-map-B} as a moment map for a Hamiltonian $U(1)$-action on $X_\circ$ and $c$ as the corresponding K\"ahler parameter. 
Then $B$ is obtained as the symplectic reduction
\be
 B = X_\circ//U(1)
\ee
and $c$ is its only non-vanishing modulus, which is to say that $H_2(B,\BZ)\cong \BZ$.

Concretely, the mesonic twist acts as a specialization of parameters,
\be
\label{eq:mesonic-twist-on-t}
 t^a_\text{mes.} = 0\ , \qquad \forall a\ ,
\ee
which in terms of the ineffective K\"ahler parameters $\lam^i$, corresponds to imposing the linear constraints
\be
\label{eq:mesonictwist}
 \sum_i Q_i^a \lam_\text{mes.}^i = 0\ , \quad \forall a\ .
\ee
Using the above condition together with \eqref{eq:BVtoVol}, it is straightforward to derive the general behavior,
\be
 \BV_X (\lam_\text{mes.}, \e)
 = \mathe^{\e_i \lam_\text{mes.}^i}\, \VVV_X(t=0, \e)
 =: \mathe^{\e_i \lam_\text{mes.}^i}\, \VVV_X^\text{mes.} (\e)\ .
\ee
Note that the expression for the equivariant volume $\VVV_X(t,\e)$ at general non-zero values of $t$, depends not only on the charge matrix $Q_i^a$, but also on the chamber in the K\"ahler moduli space where $t$ lies, cf.\ \eqref{eq:altequivol}. On the other hand, in the limit $t\to 0$ the final expression for the equivariant volume is independent of the chamber, which means that under the mesonic twist there is a universal formula for $\VVV_X^\text{mes.} (\e)$ which only depends on the charge matrix $Q_i^a$. 

We can further understand the mesonic condition through the use of the toric fan vectors $v_\alpha^i$, which by definition satisfy the relation \eqref{eq:map-v}
where $\alpha$ labels precisely the mesonic symmetries (i.e.\ those symmetries which act effectively on the quotient $X$), while $a$ labels the baryonic ones. It follows that the parameters $\lam^i$, in order to satisfy the constraint \eqref{eq:mesonictwist}, must be expressed as linear combinations of the vectors of the toric fan, namely
\be
\label{eq:lam-mesonic}
 \lam_\text{mes.}^i(\mu) := \sum_\alpha v^i_\alpha \mu^{\alpha}\ ,
\ee
for some coefficients $\mu^\alpha$ that play the role of ``reduced'' $\lam$ parameters.
In other words, the parameters $\mu^\alpha$ are the totally redundant K\"ahler moduli, i.e.\ those orthogonal to the effective K\"ahler moduli $t^a$.

Notice that \eqref{eq:lam-mesonic} together with \eqref{eq:nuinsteadofeps} imply $\e_i\lam_\text{mes.}^i(\mu) = \e_i v_\alpha^i \mu^\alpha = \nu_\alpha(\e)\mu^\alpha$, so that the equivariant volume at the mesonic twist point can be written as an explicit function of the reduced parameters $\mu^\alpha$ and $\nu_\alpha$:
\be
 \BV^\text{mes.}(\mu,\nu) := \tilde{\BV}(\lam_\text{mes.}(\mu),\nu)\ .
\ee

Importantly, the above discussion allows us to simplify the genus-zero and genus-one constant maps terms in \eqref{eq:cconstmapssummary} at the mesonic twist point. Namely, if we assume that $\lam^i$ are of the form \eqref{eq:lam-mesonic} then, for genus $\mathfrak{g}=0,1$ we can use \eqref{eq:vdlamBV} to write
\be
\label{eq:mesonicconstmaps}
\begin{aligned}
 N_{0,0}^\text{mes.} (\lam_\text{mes.}, \e) &=
 \frac16\, \left(\e_i \lam_\text{mes.}^i\right)^3\, \VVV^\text{mes.}(\e)\ , \\
 N_{1,0}^\text{mes.} (\lam_\text{mes.}, \e) &=
 \frac1{24}\,
 \left(\e_i \lam_\text{mes.}^i\right)\, \sum_{i < j} \frac{\partial^2 \BV (\lam,  \e)}{\partial \lam^i \partial \lam^j} \Big|_{\lam = 0}\,\ , 
\end{aligned}
\ee
while the higher-genus constant maps cannot be simplified from their expression in \eqref{eq:cconstmapssummary}. It is important to stress here that the specialization of $\lam$-parameters as in \eqref{eq:lam-mesonic} should be done after taking derivatives w.r.t.\ the unconstrained variables $\lam^i$. Moreover, since $\sum_{i<j}\partial_i\partial_j\BV|_{\lam=0}$ corresponds to the integral of the second Chern class, in the case that $X$ is a complex two-fold it becomes equal to the Euler number which is independent of the equivariant parameters, while in all other cases it is a non-trivial rational function of the~$\e_i$. A similar consideration applies to the function $\sum_{i<j<k}\partial_i\partial_j\partial_k\BV|_{\lam=0}$ (which appears at genus $\mathfrak{g}>1$) that becomes equal to the Euler number if $X$ is a complex three-fold.

These expressions can in turn be converted as functions of $\mu$ and $\nu$ in a straightforward way given the previous discussion and identities.

\subsubsection*{Mesonic constraint}

An important observation one needs to make regarding the relation between the redundant and effective equivariant parameters is that, while it is straightforward to define the $\e_i$'s as functions of the $\nu_\alpha$'s, the converse is not true. In fact, \eqref{eq:nuinsteadofeps} generically does not admit a unique solution as a system of equations for the $\e_i$'s. While it is true that by making some non-canonical choices one can always find a way to invert those relations, one would like to be able to have a systematic procedure to define $\e_i$'s as functions of $\nu_\alpha$'s.
In order to do so, we will make use of the mesonic twist \eqref{eq:mesonic-twist-on-t} together with the observation that
\be
\label{eq:mesonic-constr-eps}
 \left.v_\alpha^i\frac{\partial}{\partial\lam^i}\log\BV(\lam,\e)\right|_{\lam=\lam_{\rm mes.}(\mu)}
 = v^i_\alpha\e_i\ ,
\ee
which follows directly from the identity \eqref{eq:vdlamBV}.
The matrix $v_i^\alpha$ appears on both sides, however we cannot get rid of it since it is not an isomorphism. Nevertheless, one could decide to impose the condition that the $\e_i$'s are such that they satisfy
\be
\label{eq:mesonicconstraint}
 \e_i \stackrel{!}{=}
 \left.\frac{\partial}{\partial\lam^i}\log\BV(\lam,\e)\right|_{\lam=\lam_{\rm mes.}(\mu)}\ ,
\ee
which now form a system of non-linear equations due to the fact that the r.h.s.\ is a complicated rational function of the $\e_i$'s, namely we can rewrite \eqref{eq:mesonicconstraint} explicitly as
\be
 \e_i \stackrel{!}{=}
 \frac{\prod_a\oint\frac{\mathd\phi_a}{2\pi\mathi}\prod_{j\neq i}(\e_j+\phi_aQ_j^a)^{-1}}
 {\prod_a\oint\frac{\mathd\phi_a}{2\pi\mathi}\prod_j (\e_j+\phi_aQ_j^a)^{-1}}\ .
\ee
We refer to this as the \emph{mesonic constraint} on the $\e_i$'s.
As it turns out, not all of these equations are linearly independent precisely because of the identity \eqref{eq:mesonic-constr-eps}, which means that one cannot solve for the $\e_i$'s completely, in fact only $r$ of them can be fixed by the equations. The remaining $d$ independent degrees of freedom can be though of as parametrizing a specific choice of the effective parameters $\nu_\alpha$'s. In practice, getting rid of the redundancies in the $\e_i$'s by imposing the constraints \eqref{eq:mesonicconstraint} is an effective but not unique way to invert the relations in \eqref{eq:nuinsteadofeps}.%
\footnote{Alternatively, we can define functions
\be
 \e_i^{\rm mes.}(\nu) := \left.\frac{\partial}{\partial\lam^i}
 \log\tilde{\BV}(\lam,\nu)\right|_{\lam=\lam_{\rm mes.}(\mu)}
\ee
where we remark that the r.h.s.\ does not depend on $\mu^\alpha$'s due to some cancellations that only happen after we impose the mesonic twist on the K\"ahler parameters $\lam^i$'s.
With this definition, we readily have that $v^i_\alpha\e_i^{\rm mes.}(\nu)=\nu_\alpha$,
which means that we have found a way to invert the relations in \eqref{eq:nuinsteadofeps},
provided we know the function $\tilde{\BV}(\lam,\nu)$ beforehand.} It is interesting to observe that the mesonic constraint on the $\e$ parameters also has a holographic counterpart which has been previously explored in \cite{Hosseini:2019ddy}.

\subsubsection*{Universal parametrization}
As we discussed, the mesonic twist acts as a blow-down of the non-trivial two cycles in $X$ corresponding to the vanishing of all of its effective K\"ahler parameters. What remains, are the ineffective parameters $\mu^\alpha$ which are formally conjugate variables to the equivariant parameters $\nu_\alpha$ associated to the effective torus isometries of $X_\circ$ and therefore those of the base $L$.
In order to pass to the quotient $B$, we have to perform an additional symplectic reduction w.r.t.\ a $U(1)$ subgroup of the isometry group, which introduces a unique non-trivial two-cycle in $H_2(B,\BZ)$. The associated K\"ahler parameter is the variable conjugate to the equivariant parameter of the symmetry that has been quotiented out.
For the sake of concreteness, we can assume that the $U(1)$ quotient is defined by the moment map
\be
\label{eq:relationtoB}
 \sum_i v^i_\alpha |z_i|^2 = \mu^\alpha
\ee
where $z_i$ are homogeneous coordinates on $X_\circ$ and the label $\alpha$ identifies the 1-parameter subgroup.
Then $\mu^\alpha$ becomes the unique (effective) K\"ahler parameter of $B$ and the corresponding equivariant parameter $\nu_\alpha$ becomes an integration variable in the contour integral representation of the equivariant volume of $B$. In other words, we have removed the $\alpha$-th column vector from $v$ and added it as a row vector to $Q$.
In particular, this means that the direct relation between supergravity and the dual field theory can be performed by picking a single non-vanishing $\mu$-parameter. 

Among all possible choices of label $\alpha$ for the reduction, there is one which is somewhat special. This corresponds to the $\alpha$ such that $v^i_\alpha=1$ for all $i$, which always exists due to the CY condition on the matrix $Q_i^a$. It is indeed a standard convention that the first toric fan vector of a Calabi--Yau manifold is simply $v_1 = (1, 1,\dots, 1)$. We will then refer to this choice as the \emph{universal reduction} $B$, and 
we will refer to the parametrization
\be 
 \lam^i = v_1^i \mu^1 = \mu^1\,, \quad \forall i\,,
\ee 
as the \emph{universal parametrization}.

\section{Basic toric examples}
\label{sec:3}

In this section, we begin with some of the simplest toric manifolds, which serve as the fundamental building blocks for the more complex cases discussed in subsequent sections. Taking advantage of the computational simplicity of these models, we introduce key objects of interest in a more explicit and pedagogical manner. Readers already familiar with equivariant intersection theory on toric manifolds may wish to skip this section, as it primarily covers well-established results.

\subsection{Complex plane and discrete quotients, \texorpdfstring{$\BC^n/\Gamma$}{Cn/Gamma}}
The simplest (non-compact) toric CY manifold consists of $n$ copies of the complex plane, $\BC^n$. It has no non-trivial cycles, so it has no effective K\"ahler parameters $t^a$ to start with. Correspondingly, there is no charge vector $Q_i^a$ that can be defined. For illustrational purposes, we can be fully explicit here and introduce local coordinates,
\be
 \mathd s^2_{\BC^n} = \sum_{i=1}^n\,  (\mathd r_i^2 + r_i^2 \mathd \varphi_i^2)\ ,
\ee
with $2\pi$-periodic angle coordinates $\varphi_i$.
The K\"ahler form is given by
\be
 \omega
 = \sum_{i=1}^n \frac{\mathd z_i\wedge\mathd\bar{z}_i}{2\pi\mathi}
 = \frac1{\pi} \sum_{i=1}^n \mathd\varphi_i\wedge r_i \mathd r_i \ .
\ee
 There is a natural $U(1)^n$ action of rotations around the origin, parametrized by
\be
 v := \sum_{i = 1}^n \e_i\, v^i = 2 \pi\, \sum_{i = 1}^n \e_i\, \partial_{\varphi_i}\ ,
\ee
such that the equivariant indices in this case are identified with the label on the respective copy of $\BC$. The equivariant completion of $\omega$ is given by the moment maps
\be
 H^i = (r_i)^2 \quad \Rightarrow \quad (\mathd + \e_i \iota_{v^i}) (\omega - \e_i\, H^i)
 = \e_i\, (\iota_{v^i} \omega - \mathd H^i) = 0\ .
\ee
Clearly the standard symplectic volume on $\BC^n$ is not well-defined, so we can directly consider the equivariant volume. Via a direct integration from \eqref{eq:equivol} we find,
\be
\begin{aligned}
 \VVV(\e) &= \int_{\BC^n}\mathe^{\omega-\e_i H^i}
 = \int_{\BC^n} \frac{\omega^n}{n!} \mathe^{-\e_i H^i} \\ 
 &= \frac1{\pi^n} \int_0^{2\pi} \prod_i \mathd \varphi_i \int_0^\infty
 \prod_i (r_i \mathd r_i)\, \mathe^{- \sum_i \e_i\, r_i^2}
 = \frac1{\prod_i \e_i}\ .
\end{aligned}
\ee
The same result is reproduced also via the ABBV localization theorem, where the fixed point of the Killing vector $\xi$ corresponds to the origin of each copy of $\BC$, such that
\be
\label{eq:Cequivol}
 \VVV(\e)  = \frac{\mathe^{-\e_i H^i} |_{r=0}}{e(T_{\{0\}}\BC^n)}
 = \frac1{\prod_i \e_i}\ .
\ee
where $e(T_{\{0\}}\BC^n)$ is the equivariant Euler class of the normal bundle to the fixed point $\{0\}\hookrightarrow\BC^n$.

The third way of reproducing the calculation of the equivariant volume, i.e.\ the contour integral in \eqref{eq:altequivol}, automatically agrees with the result above, since the equivariant Chern roots are simply the equivariant parameters, $x_i = \e_i$, and there is no integration due to a lack of K\"ahler moduli.

After this straightforward calculation, let us turn to the main object of interest here: the constant maps. We first evaluate the equivariant volume in the $\lam$-parametrization, \eqref{eq:identity-tlam},
\be
\label{eq:CequivvolMZ}
 \BV(\lam,\e) = \exp(\e_i\lam^i)\,\VVV(\e)
 = \frac{\mathe^{\e_i\lam^i}}{\prod_i\e_i}\ ,
\ee
where summation over repeated up-down indices is implied.
Note that in this simplest case, due the lack of effective K\"ahler moduli, the mesonic twist condition is trivial, and the matrix $v_\alpha^i$ can be taken to be the identity matrix, so that $\lam=\mu$ and $\e=\nu$. With these identifications, we can write
\be
 \BV_{\BC^n}(\lam,\e) \equiv \BV^\text{mes.}_{\BC^n}(\mu,\nu)\ .
\ee
This is only true for the present example and does not hold on any of the other examples.

The genus-zero and genus-one constant maps follow directly from \eqref{eq:g0equimap3} and \eqref{eq:g1equimap}:
\be
\label{eq:cplanemaps0and1}
 N_{0,0} = \frac{(\sum_i \e_i\lam^i)^3}{6\,\prod_i\e_i}\ ,
 \hspace{30pt}
 N_{1,0} = \frac{(\sum_{i<j} \e_i \e_j)\, (\sum_k\e_k\lam^k)}{24\, \prod_i\e_i}\ ,
\ee
while the higher genus constant maps follow from \eqref{eq:ggeq2equimap},
\be
\label{eq:cplanemapsg}
 N _{\mathfrak{g}\geq 2,0} = \frac{(-1)^\mathfrak{g}}{(2\mathfrak{g}-2)!}\,
 \frac{|B_{2\mathfrak{g}}|}{2\mathfrak{g}} \frac{|B_{2\mathfrak{g}-2}|}{(2\mathfrak{g}-2)}\,
 \frac{\sum_{i<j<k} \e_i \e_j \e_k}{2\, \prod_i \e_i}\, \ .
\ee
Notice that the genus-zero expression is non-vanishing for any dimension $n$, while the genus-one expression only vanishes for $n=1$ and the higher genus expression vanishes for $n=1,2$. 

Let us also consider a generalization of the above case allowing for finite discrete quotients of the form $\BC^n / \Gamma$. Assuming that the finite group $\Gamma$ preserves the symplectic structure, the full space inherits the toric action and the equivariant structure from the $\BC^n$ case above. Even though in general the resulting space is singular (we give examples of resolutions in the next section), we can still define the equivariant volume as before, \eqref{eq:equivol}. If we have no further knowledge of the group $\Gamma$, the equivariant volume becomes a sum over all group elements, $g \in \Gamma$. From the fixed-point theorem it further follows that each group element contributes through its fixed-point set, $\text{Fix}(g)$. 

Let us then assume the simplest case where the discrete group $\Gamma$ acts on $\BC^n$ with no fixed points other than the origin. Then we have the simple identity
\be
 \VVV_{\BC^n/\Gamma} = \frac{\VVV_{\BC^n}}{|\Gamma|}\ ,
 \quad \Rightarrow \quad
 \BV_{\BC^n/\Gamma} = \frac{\BV_{\BC^n}}{|\Gamma|}\ .
\ee
This includes the special case of the cyclic groups $\BZ_k$, with corresponding action on $(z_1, \dots, z_n) \in \BC^n$ given by
\be
 (z_1,\dots, z_n) \mapsto
 (\mathe^{2\pi\mathi m_1/k} z_1, \dots, \mathe^{2\pi\mathi m_n/k} z_n)\ ,
\ee
which are particularly relevant for holographic applications.
In this case, we have
\be
 \BV_{\BC^n/\BZ_k}(\lam,\e) = \exp(\e_i\lam^i)\, \VVV_{\BC^n/\BZ_k}(\e)
 = \frac{\mathe^{\e_i\lam^i}}{k\prod_i\e_i}\ ,
\ee
The corresponding constant maps terms $N_{\mathfrak{g},0}$ can be computed as in \eqref{eq:mesonicconstmaps} by first resolving the singularity and then imposing the mesonic twist condition after taking $\lam$-derivatives.

\subsection{Spindle, \texorpdfstring{$\WPL$}{WPL}}
\label{sec:spindle}
One-dimensional complex projective space $\BP^1$, also known as the Riemann sphere, can be realized as a toric quotient $\BC^2//U(1)$ by taking the $U(1)$ action corresponding to both charges being 1, so that the moment map is $\mu(z)=|z_1|^2+|z_2|^2$. This simple quotient admits a generalization to a weighted projective space corresponding to the quotient taken w.r.t.\ the action with two charges which are not necessarily equal to 1.

In this case, the moment map is $\mu(z)=a|z_1|^2+b|z_2|^2$ where the charges $a,b\in\BZ_{>0}$ introduce two conical deficit angles at the poles of the Riemann sphere. The resulting quotient is still toric however it becomes a singular space, better described via the language of Deligne--Mumford stacks.\footnote{In the case where $\gcd(a,b)\neq1$ the resulting quotient stack is not reduced and generically takes the form of a $\BZ_{\gcd(a,b)}$-gerbe over the reduced stack.}

Known colloquially as a \emph{spindle}, this manifold has recently been considered as a horizon geometry for asymptotically AdS black holes in gauged supergravity, generalizing the usual spherical topologies, see \cite{Ferrero:2020laf,Ferrero:2020twa} and references therein. One can think of the spindle topologically as a sphere, with a single K\"ahler modulus $t$. It can be represented as a stack quotient,
\be
 \WPL = \BC^2 // U(1)\ ,
\ee
and a corresponding charge matrix (in this case a vector)
\be
\label{eq:spindleQ}
	Q = (a, b)\ .
\ee

This example was considered in some detail in \cite[Section 3.1]{Martelli:2023oqk}, where local coordinates and explicit expressions for the equivariant K\"ahler form can be found. To avoid repetition while still remaining pedagogical, let us focus on evaluating the equivariant volume using \eqref{eq:altequivol}, our main mathematical tool in the forthcoming examples. Applying this formula for the spindle (using the relabeling of parameters $\e_1 = \e_+, \e_2 = \e_-$ and likewise for the superscript of the $\lam$-parameters) leads to
\be
\label{eq:Cequivolspindle}
 \VVV(t, \e)
 =\oint\frac{\mathd\phi}{2\pi\mathi}\frac{\mathe^{\phi t}}{(\e_++a\phi)(\e_-+b\phi)} \\
 = \frac{\mathe^{-\frac{\e_+ t}{a}}}{(a\e_--b\e_+)}
 + \frac{\mathe^{-\frac{\e_- t}{b}}}{(b\e_+-a\e_-)}\ ,
\ee
where we took the poles at $\phi = -\frac{\e_+}{a}$ and $\phi = -\frac{\e_-}{b}$. Note that the denominators of the two terms are equal, up to a sign, such that in the non-equivariant limit we have a cancellation of all the singular terms,
\be
\label{eq:spindle-volume-eq}
 \lim_{\e_\pm \to 0} \VVV(t,\e)
 = \frac{t}{a\, b} - \frac{t^2}{2 a^2 b^2}\, (b\e_+ + a\e_-) + \cO (\e^2)\ .
\ee
As expected, the leading order is proportional to $t$ and it parametrizes the ordinary volume, while the equivariant upgrade adds corrections at higher orders in the K\"ahler modulus $t$ and the equivariant parameters $\e_\pm$'s.\footnote{From a cohomological point of view, one should assign degree $-2$ to the K\"ahler moduli $t^a$ and degree $+2$ to the equivariant parameters $\e_i$, so that the function $\VVV(t,\e)$ becomes homogeneous of degree $-2d$.}
Note that in order for the quotient to be a compact two-fold as described above, we need to impose that both charges are positive integers and that the K\"ahler modulus is also greater than zero. However, once the volume has been computed via equivariant localization as in \eqref{eq:Cequivolspindle}, we can analytically continue the function $\VVV(t,\e)$ to arbitrary values of $a,b$ and $t$. In particular, we will allow $b$ to become negative while keeping $a$ positive (or equivalently the other way around).
The reason for doing so has to do with the fact that, from the physics point of view, this choice still corresponds to a BPS background in string theory, as argued in \cite{Martelli:2023oqk} (see also \cite{Cassia:2025jkr}).

In the $\lam$-parametrization, the equivariant volume takes the following form, cf.\ \eqref{eq:identity-tlam} and \eqref{eq:spindleQ},
\be
\label{eq:equivolwp1}
 \BV(\lam,\e) = 
 \frac{\mathe^{\lam^+\,\e_+ + \lam^-\,\e_-}}{(a \e_- - b \e_+)}\,
 \left(\mathe^{-(a\lam^+ + b \lam^-) \e_+/a}
 - \mathe^{-(a \lam^+ + b \lam^-) \e_-/b} \right)\ ,
\ee
in agreement with the calculations in \cite{Martelli:2023oqk}. It is natural to introduce a ``reduced'' equivariant parameter on the spindle,
\be
\label{eq:definingsimplifiedparameterforthespindle}
 \NN := a \e_- - b \e_+\ ,
\ee
so that we can write the simplified expression
\be
 \tilde \BV (\lam, \nu) = \frac{\mathe^{\frac{\NN}{a} \lam^-}
 - \mathe^{-\frac{\NN}{b} \lam^+}}{\NN}\ .
\ee
where we have lifted the redundancy in the equivariant parameters from $\e_+,\e_-$ to $\NN$.

The latter parametrization in terms of $\NN$, which notationally we reserve throughout the paper for the equivariant parameter of the spindle, is immediate in this simple case. In the more involved examples that follow later, one can instead use the toric fan vectors to define $\nu(\e)$, cf.\ \eqref{eq:nuinsteadofeps}. In this case, the toric fan is given by the kernel of $Q$,
\be
 v =
 \begin{pmatrix}
 -b \\ a 
 \end{pmatrix}\ ,
\ee
so that we can define $\nu_1(\e) = v_1^i \e_i = a \e_- - b \e_+$ which matches the definition of $\NN$ in \eqref{eq:definingsimplifiedparameterforthespindle}. Note that the toric fan vector is not uniquely defined and allows for an arbitrary constant rescaling, which would result in an ambiguous overall numerical factor.

Using the expression in terms of a single equivariant parameter, the genus-zero constant maps term \eqref{eq:g0equimap3} reads
\be
\label{eq:spindleconstantmaps}
	N_{0, 0} = \frac{\NN^2}{6\, a^3 b^3}\, \left( (a \lam^+)^3 + (b \lam^-)^3 \right)\ ,
\ee
while the higher genus constant maps, as defined in \eqref{eq:g1equimap} and \eqref{eq:ggeq2equimap}, vanish identically due to the low dimension of the manifold,
\be
	N_{1, 0} = 0 \ , \qquad N_{\mathfrak{g} \geq 2, 0} = 0\ .
\ee
The spindle is of course \emph{not} a CY manifold, so the meaning of this calculation in itself is not immediately clear from our previous discussion. Our main motivation to look at the spindle lies in the role it has as the base of a fibration, whose total space is indeed a CY manifold, as discussed in Section~\ref{sec:4}.

Finally, note that imposing the mesonic twist condition, \eqref{eq:mesonictwist}, in the case of the spindle leads to the trivial result $\BV_{\WPL}^\text{mes.} = 0$. This is because the mesonic twist erases the information about two cycles, which is in essence what the spindle is, thus erasing the entire topology.

\subsection{Projective space, \texorpdfstring{$\BP^k$}{Pk}}
A higher-dimensional generalization of the previous example is the $k$-dimensional complex projective space. For each $k$, one could consider further orbifold generalizations, $\mathbb{WP}^k_{a_1, a_2, \dots, a_{k+1}}$, e.g.\ $\mathbb{WP}^2_{a, b, c}$ considered in \cite[Section 3.2.1]{Martelli:2023oqk}. In this generic case, the charge vector is given accordingly by
\be
\label{eq:WPk}
 Q = (a_1, a_2, \dots, a_{k+1})\ .
\ee
Since the addition of the conical deficit numbers is straightforward and goes in analogy to the previous example, let us just take the case $( a_1, a_2, \dots, a_{k+1} ) = (1, 1, \dots, 1)$ in order not to clutter the following formulae. 

In the $k$-dimensional case, the manifold exhibits a total of $(k+1)$ fixed points, and a single K\"ahler modulus $t$. The equivariant volume can again be calculated by a fixed point formula, \eqref{eq:equivol}, or via the contour integral, \eqref{eq:altequivol}. In the latter case we find
\be
 \VVV(t, \e)
 =\oint\frac{\mathd\phi}{2\pi\mathi}\frac{\mathe^{\phi t}}{\prod_{i = 1}^{k+1}(\e_i+\phi)}
 = \sum_{i=1}^{k+1} \frac{\mathe^{-\e_i t}}{\prod_{j \neq i} (\e_j - \e_i)}\ ,
\ee
where the JK contour includes all the poles of the integrand. Due to compactness of the quotient space, the non-equivariant limit is well-defined and produces the ordinary volume as the leading term in the $\e$-expansion,
\be
 \lim_{\e_i\to 0} \VVV_{\BP^k} (t, \e)
 = \frac{t^k}{k!} - \frac{t^{k+1}}{(k+1)!}\, \sum_{i=1}^{k+1} \e_k + \cO(\e^2)\ .
\ee

In the $\lam$-parametrization, the equivariant volume takes the form,  cf.\ \eqref{eq:identity-tlam},
\be
 \BV (\lam, \e) = \sum_{i=1}^{k+1} \prod_{j \neq i}
 \frac{\mathe^{\lam^j(\e_j-\e_i)}}{(\e_j-\e_i)}\ ,
\ee
generalizing \eqref{eq:equivolwp1}. We can rewrite the answer in terms of parameters $\nu_\alpha$'s, cf.\ \eqref{eq:nuinsteadofeps}, observing that the toric fan is given by the matrix 
\be
 v =
 \begin{pmatrix}
 1 & 0 & \dots & 0 \\
 0 & 1 & \dots & 0 \\
 \vdots & \vdots & \ddots & \vdots \\
 0 & 0 & \dots & 1 \\
 -1 & -1 & \dots & -1
 \end{pmatrix}\ ,
\ee
such that we find
\be
	\nu_\alpha = \e_\alpha - \e_{k+1}\ , \qquad \forall \alpha = 1, \dots, k\ .
\ee
We can then lift the redundancy in the equivariant parameters and write
\be
 \tilde\BV (\lam, \nu)
 = \left( \prod_{\alpha = 1}^k \frac{\mathe^{\lam^\alpha\nu_\alpha}} {\nu_\alpha} \right)
 - \sum_{\alpha = 1}^k \left( \frac{\mathe^{-\lam^{k+1}\nu_\alpha}}
 {\nu_\alpha} \prod_{\beta \neq \alpha} \frac{\mathe^{\lam^\beta(\nu_\beta-\nu_\alpha)}}
 {(\nu_\beta - \nu_\alpha)} \right)\ .
\ee

The constant maps terms that can be derived from the equivariant volume can now be written either in the $\e$ of in the $\nu$ parametrization, but in this case the expressions are more compact with the former choice. Using it, we evaluate the genus-zero and genus-one constant maps,
\be
\ba
 N_{0, 0} &= \sum_{i = 1}^{k+1} \frac{(\sum_{j\neq i} \lam^j (\e_j - \e_i))^3}{6\, \prod_{j \neq i} (\e_j - \e_i)} \ ,\\
 N_{1, 0} &=  \sum_{i = 1}^{k+1} \frac{(\sum_{m<n \neq i} (\e_m - \e_i) (\e_n-\e_i) ) (\sum_{l \neq i} \lam^l (\e_l-\e_i))}{24\, \prod_{j \neq i} (\e_j - \e_i)} \ ,
\ea
\ee
and the higher-genus contribution,
\be
 N_{\mathfrak{g} \geq 2, 0} = \frac{(-1)^\mathfrak{g}}{(2\mathfrak{g}-2)!}\,
 \frac{|B_{2\mathfrak{g}}|}{2\mathfrak{g}} \frac{|B_{2\mathfrak{g}-2}|}
 {(2\mathfrak{g}-2)}\, \sum_{i = 1}^{k+1}
 \frac{\sum_{l<m<n \neq i} (\e_l-\e_i) (\e_m-\e_i) (\e_n-\e_i)}
 {2\, \prod_{j \neq i} (\e_j-\e_i)}\, \ .
\ee
Again, since $\BP^k$ is not a CY manifold, the interpretation of these results is not immediately clear and it is presented here mostly for pedagogical purposes.

\section{Calabi--Yau cones and resolutions}
\label{sec:3.5}

In this section we consider a set of examples of CY cones and their resolutions, which in general present more involved cases in comparison with the basic examples of the previous section. In all of the examples here we look at symplectic quotients given by the charge matrices obeying
\be
	\sum_i Q^a_i = 0\ , \qquad \forall a\ ,
\ee
which corresponds to the Calabi--Yau condition on toric manifolds. As a main calculational tool we directly apply formula \eqref{eq:altequivol}, and determine the equivariant volume in the various parametrizations we have introduced. In order to keep a relatively compact presentation, in this section we refrain from explicitly evaluating the constant maps terms of the topological string, which follow straightforwardly from the calculations we present. The reader can find other explicit examples of CY resolutions in \cite[Sections 10--11]{Cassia:2022lfj}.

\subsection{Local \texorpdfstring{$\BP^{n-1}$}{Pn-1}}
The space $X = K_{\BP^{n-1}}$, i.e.\ the total space of the canonical bundle over $\BP^{n-1}$, is a local CY manifold obtained as a toric quotient defined by the charge matrix
\be
 Q = \left(1,1,\dots,1,-n\right)\ ,
\ee
and corresponds to the resolution of the cone $\BC^n/\BZ_n$ (blowing up the singular point into a $\BP^{n-1}$). The corresponding Sasakian space at the base of the cone is $L = \mathrm{S}^{2 n-1} / \BZ_n$, which can be described as the total space of a circle bundle over the manifold $B = \BP^{n-1}$.

The equivariant volume in the chamber $t>0$ is given by, \cite{Cassia:2022lfj}, 
\be
 \VVV(t,\e) = \oint_\text{JK} \frac{\mathd\phi}{2\pi}
 \frac{\mathe^{\phi t}}{(\e_{n+1} - n\phi)\prod_{i=1}^n(\e_i+\phi)}
 = \sum_{i=1}^n\frac{\mathe^{-t\e_i}}{(\e_{n+1}+n\e_i)
 \prod_{j \neq i}^n(\e_j-\e_i)}\ ,
\ee
where we took the poles at $\phi=-\e_i,i=1,\dots,n$. It follows that the equivariant volume in the $\lam$-parametrization gives
\be
 \BV(\lam,\e) = \sum_{i=1}^n
 \frac{\mathe^{\lam^{n+1}(\e_{n+1}+n\e_i)}}{(\e_{n+1}+n\e_i)}
 \prod_{j\neq i}^n\frac{\mathe^{\lam^j(\e_j-\e_i)}}{(\e_j-\e_i)}\ .
\ee
If we now impose the mesonic twist condition, \eqref{eq:mesonictwist},
$Q_i \lam^i_\text{mes.} = (\sum_{i=1}^n\lam_\text{mes.}^i)-n\lam_\text{mes.}^{n+1}=0$,
it is straightforward to show that
\be
 \BV(\lam_\text{mes.},\e)=\mathe^{\sum_{i=1}^{n+1} \lam_\text{mes.}^i\e_i}\,
 \sum_{i=1}^n \frac1{(\e_{n+1}+n\e_i)\prod_{j\neq i}^n (\e_j-\e_i)}
 = \mathe^{\sum_{i=1}^{n+1}\lam_\text{mes.}^i\e_i}\,\VVV(0,\e)\ .
\ee

After a simple inspection of these results, it is clear that we are allowed to reparametrize the result in terms of the reduced variables $\mu^\alpha,\nu_\alpha$, $\alpha=1,\dots,n$:
\be
\label{eq:locPnreducedvars}
 \nu_\alpha := \sum_{i=1}^n \delta_\alpha^i \e_i + \frac{\e_{n+1}}{n}\ ,
 \qquad
 \lam_\text{mes.}^{i\leq n} = \sum_{\alpha=1}^n\delta^i_\alpha\mu^\alpha\ ,
 \qquad
 \lam_\text{mes.}^{n+1} = \frac{1}{n} \sum_{\alpha=1}^n\mu^\alpha\ ,
\ee
which corresponds to the choice of matrix $v_\alpha^i$,
\be
 v_\alpha^i = \left\{\begin{array}{ll}
 \delta_\alpha^i\ , & \, \text{if $i\leq n$} \\
 1/n\ , & \, \text{if $i=n+1$}
 \end{array}\right.
\ee
such that we find
\be
\label{eq:locPnmesonicequivolMZ}
 \BV^\text{mes.}(\mu,\nu) = \mathe^{\sum_{\alpha=1}^n \mu^\alpha \nu_\alpha}\,
\sum_{\alpha=1}^n\frac1{n\nu_\alpha\prod_{\beta\neq\alpha}(\nu_\beta-\nu_\alpha)}
 = \frac{\mathe^{\mu^\alpha\nu_\alpha}}{n\prod_{\alpha=1}^n\nu_\alpha}\ .
\ee
Unsurprisingly, after performing the mesonic twist, i.e.\ blowing down the non-trivial two-cycle of volume $t$, we obtain the equivariant volume of the cone $\BC^n/\BZ_n$:~\footnote{Note that the overall numerical prefactor is ambiguous as it can be reabsorbed in the definition of the equivariant parameters, such that the first equality is conventional.}
\be
 \BV^\text{mes.}_{\text{loc.}\,\BP^{n-1}}
 = \frac{1}{n}\BV_{\BC^n}
 = \BV_{\BC^n/\BZ_n}\ .
\ee

For the sake of completeness, we show here (but not in the other examples below) that picking the other chamber, $t < 0$, does not change the answer after imposing the mesonic twist. In this case we find
\be
 \VVV(t < 0,\e) = \oint_\text{JK} \frac{\mathd\phi}{2\pi}
 \frac{\mathe^{\phi t}}{(\e_{n+1}-n\phi)\prod_{i=1}^n(\e_i+\phi)}
 = \frac{n^{n-1}\,\mathe^{t\e_{n+1}/n}}{\prod_{i=1}^n(n\e_i+\e_{n+1})}\ ,
\ee
where we took the single pole $\phi=\e_{n+1}/n$. In the $\lam$-parametrization we have
\be
 \BV (\lam, \e) = \frac{\mathe^{\sum_{i=1}^n \lam^i (\e_i + \e_{n+1} /n)}}
 {n\,\prod_{i=1}^n (\e_i + \e_{n+1}/n)}\ ,
\ee
such that the mesonic twist clearly reproduces \eqref{eq:locPnmesonicequivolMZ} upon the identifications \eqref{eq:locPnreducedvars}. Moreover, explicitly applying the mesonic constraint equation, \eqref{eq:mesonicconstraint}, we find
\be
 \e_i = \sum_{\alpha=1}^n \delta_i^\alpha \nu_\alpha\ ,
 \qquad \e_{n+1} = 0\ .
\ee
The resulting constant maps terms at the mesonic twist point are therefore the ones of the complex plane $\BC^n$, \eqref{eq:cplanemaps0and1} -- \eqref{eq:cplanemapsg}, divided by the factor of $n$, namely
\be
 N^{\text{loc.}\,\BP^{n-1}}_{\mathfrak{g}, 0}
 = \frac{1}{n} N^{\BC^n}_{\mathfrak{g}, 0}\ ,
 \qquad \forall \mathfrak{g}\ ,
\ee
with the exception of $\mathfrak{g}=1$ for $n=2$, and $\mathfrak{g}\geq2$ for $n=3$, due to the issue discussed after \eqref{eq:mesonicconstmaps}.

\subsection{Resolved conifold}
\label{sec:4.2}
We now consider a toric CY three-fold defined by the charge matrix
\be
 Q = \left(1,1,-1,-1\right)\ ,
\ee
and the chamber $t>0$. This corresponds to the resolution of the singular cone described by the quadratic equation in $\BC^4$:
\be
 z_1 z_3 - z_2 z_4 = 0\ .
\ee
The singular geometry is described by a cone over the Sasakian space $L = T^{1,1}$, which by definition can be described as a circle bundle over $B = \BP^1 \times \BP^1$ with first Chern class equal to $(1,1)\in H^2(B,\BZ)\cong\BZ\times\BZ$.

The equivariant volume is given by, \cite{Cassia:2022lfj},
\be
\ba
 \VVV(t, \e) &= \oint_\text{JK} \frac{\mathd \phi}{2 \pi}
 \frac{\mathe^{\phi t}}{(\e_1 + \phi) (\e_2+\phi) (\e_3-\phi) (\e_4-\phi)} \\
 &= \frac1{(\e_2-\e_1)}\, \left( \frac{\mathe^{-t \e_1}}{(\e_3+\e_1) (\e_4+\e_1)}
 - \frac{\mathe^{-t \e_2}}{(\e_3+\e_2) (\e_4+\e_2)} \right)\ ,
\ea
\ee
where we took the poles at $\phi = -\e_1$ and $\phi = -\e_2$.  Similarly, the equivariant volume in the $\lam$-parametrization gives
\be
 \BV (\lam, \e) = \frac1{(\e_2-\e_1)}\, \left(
 \frac{\mathe^{\lam^2 (\e_2-\e_1)+\lam^3 (\e_3+\e_1) +\lam^4 (\e_4+\e_1)}}
 {(\e_3+\e_1) (\e_4+\e_1)}
 - \frac{\mathe^{\lam^1 (\e_1-\e_2)+\lam^3 (\e_3+\e_2) +\lam^4 (\e_4+\e_2)}}
 {(\e_3+\e_2) (\e_4+\e_2)} \right)\ .
\ee
If we now impose the mesonic twist condition, $\lam^1+\lam^2=\lam^3+\lam^4$, it is straightforward to show that
\be
 \BV (\lam_\text{mes.}, \e) = \frac{\mathe^{\e_i\lam_\text{mes.}^i}\, \sum_{i=1}^4\e_i}
 {(\e_1+\e_3) (\e_2+\e_3) (\e_1+\e_4) (\e_2+\e_4)}
 = \mathe^{\e_i \lam_\text{mes.}^i}\, \VVV^\text{mes.} (\e)\ .
\ee
If we look at the mesonic constraint on the $\e$ parameters, \eqref{eq:mesonicconstraint},  it takes the simple form:
\be
\label{eq:conifoldmesonicparametrizationfore}
\begin{array}{ccc}
 \e_1 = \frac{(\e_1 + \e_3) (\e_1 + \e_4)}{\sum_{i = 1}^4 \e_i}\ ,
 &\quad&
 \e_3 = \frac{(\e_1 + \e_3) (\e_2 + \e_3)}{\sum_{i = 1}^4 \e_i}\ , \\
 && \\
 \e_2 = \frac{(\e_2 + \e_3) (\e_2 + \e_4)}{\sum_{i = 1}^4 \e_i}\ ,
 &\quad&
 \e_4 = \frac{(\e_1 + \e_4) (\e_2 + \e_4)}{\sum_{i = 1}^4 \e_i}\ ,
\end{array}
\ee
which can be simplified to give the following quadratic relation,
\be
\label{eq:conifoldmesonicconstr}
 \e_1 \e_2 = \e_3 \e_4\ .
\ee
This is precisely the constraint discussed in \cite[Section 5.1]{Hosseini:2019ddy},~\footnote{The rest of the quantities discussed here also have a direct correspondence to the ones in  \cite[Section 5.1]{Hosseini:2019ddy}, with $\e^\text{here} \leftrightarrow \Delta^\text{there}$ and $\nu^\text{here} \leftrightarrow b^\text{there}$. One should however take into account the additional equivariant parameter in that reference that comes from a direct product with $\BC$, which we have not considered here but is straightforward to include.} which on the field theory side can be understood as lifting up the flat directions in the matrix model, see \cite{Jafferis:2011zi}. The constant maps terms in the topological string amplitudes after the mesonic twist are given by \eqref{eq:mesonicconstmaps}, under the constraint \eqref{eq:conifoldmesonicconstr}. 

To illustrate the alternative parametrization in the variables $\mu^\alpha$ and $\nu_\alpha$, we can use the toric fan matrix
\be
 v =
 \begin{pmatrix}
 1 & 1 & 0 \\
 1 & 0 & 1 \\
 1 & 1 & 1 \\
 1 & 0 & 0
 \end{pmatrix}
\ee
such that
\be
\begin{array}{c}
 \lam_\text{mes.}^1 = \mu^1+\mu^2\ ,
 \quad
 \lam_\text{mes.}^2 = \mu^1+\mu^3\ ,
 \quad
 \lam_\text{mes.}^3 = \mu^1+\mu^2+\mu^3\ ,
 \quad
 \lam_\text{mes.}^4 = \mu^1\ ,
 \\ \\
 \nu_1 = \e_1+\e_2+\e_3+\e_4\ ,
 \quad
 \nu_2 = \e_1+\e_3\ ,
 \quad
 \nu_3 = \e_2+\e_3\ ,
\end{array}
\ee
and we find
\be
\label{eq:conifoldmesonic}
	\BV^\text{mes.} (\mu, \nu)=  \frac{\mathe^{\nu_\alpha \mu^\alpha}\, \nu_1}{\nu_2 \nu_3 (\nu_1 - \nu_2) (\nu_1 - \nu_3)}\ .
\ee

Note that the parametrization in terms of the $\nu$ variables is insensitive of whether we impose the constraint \eqref{eq:conifoldmesonicconstr} on the $\e$ parameters or not. This is because the conditions \eqref{eq:conifoldmesonicparametrizationfore} translate to
\be
 \e_1 = \frac{\nu_2 (\nu_1 - \nu_3)}{\nu_1}\ , \quad
 \e_2 = \frac{\nu_3 (\nu_1 - \nu_2)}{\nu_1}\ , \quad
 \e_3 = \frac{\nu_2 \nu_3}{\nu_1}\ , \quad
 \e_4 = \frac{(\nu_1 -\nu_2) (\nu_1 - \nu_3)}{\nu_1}\ ,
\ee
such that the constraint on $\e$, \eqref{eq:conifoldmesonicconstr} is automatically satisfied and does not lead to any constraint on the $\nu$ parameters. On the other hand the evaluation of the constant maps is more straightforward in the $\e$ parametrization, \eqref{eq:mesonicconstmaps}.

\subsection{Cone over \texorpdfstring{$M^{1,1,1}$}{M111}}
We now consider the CY four-fold $X$ defined as the cone over the Sasakian manifold $L = M^{1,1,1,}$, which is the circle bundle over the base $B = \BP^1 \times \BP^2$ \cite{Fabbri:1999hw}. The charge vector of the K\"ahler quotient is given by
\be
 Q = \left(3,3,-2,-2,-2\right)\ ,
\ee
so that the equivariant volume can be computed as
\be
\ba
 \VVV(t, \e) &= \oint_\text{JK} \frac{\mathd \phi}{2 \pi} \frac{\mathe^{\phi t}}
 {(\e_1 + 3 \phi) (\e_2+3 \phi) (\e_3-2 \phi) (\e_4-2 \phi) (\e_5 - 2 \phi)} \\
 &= \frac1{3(\e_2-\e_1)}\, \left( \frac{\mathe^{-t \e_1/3}}
 {(\e_3+\tfrac23 \e_1) (\e_4+\tfrac23 \e_1) (\e_5 + \tfrac23 \e_1)}
 - \frac{\mathe^{-t \e_2/3}}
 {(\e_3+\tfrac23 \e_2) (\e_4+\tfrac23 \e_2) (\e_5 + \tfrac23 \e_2)} \right)\ ,
\ea
\ee
where we took the poles at $\phi = -\e_1 / 3$ and $\phi = -\e_2 / 3$. In the $\lam$-parametrization we have
\begin{multline}
 \BV (\lam, \e) = \frac1{3(\e_2-\e_1)}\,
 \Big( \frac{\mathe^{\lam^2 (\e_2-\e_1)+\lam^3 (\e_3+\tfrac23\e_1)
 +\lam^4 (\e_4+\tfrac23\e_1) + \lam^5 (\e_5+\tfrac23 \e_1)}}
 {(\e_3+\tfrac23 \e_1) (\e_4+\tfrac23 \e_1) (\e_5 + \tfrac23 \e_1)}\\
 - \frac{\mathe^{\lam^1 (\e_1-\e_2)+\lam^3 (\e_3+\tfrac23 \e_2)
 +\lam^4 (\e_4+\tfrac23 \e_2) + \lam^5 (\e_5+\tfrac23 \e_2)}}
 {(\e_3+\tfrac23 \e_2) (\e_4+\tfrac23 \e_2) (\e_5 + \tfrac23 \e_2)} \Big)\ .
\end{multline}
If we now impose the mesonic twist condition, $\lam_\text{mes.}^1+ \lam_\text{mes.}^2 = \tfrac23 (\lam_\text{mes.}^3+\lam_\text{mes.}^4 +\lam_\text{mes.}^5)$, we find
\be
 \BV (\lam_\text{mes.}, \e)= \mathe^{\lam_\text{mes.}^i \e_i}\, \VVV(0, \e)\ .
\ee
Noting that the toric fan can be parametrized by
\be
 v =
 \begin{pmatrix}
 1 & 2 & 0 & 2 \\
 1 & 0 & 2 & 0 \\
 1 & 3 & 3 & 0 \\
 1 & 0 & 0 & 3 \\
 1 & 0 & 0 & 0
 \end{pmatrix}
\ee
it is straightforward to derive the expression for $\BV^\text{mes.} (\mu, \nu)$,
\begin{multline}
 \BV^\text{mes.} (\mu, \nu) = \frac{36\,\mathe^{\nu_\alpha\mu^\alpha}}{\nu_3-\nu_2}
 \Big(
 \frac1{\nu_2\nu_4(6\nu_1+\nu_2-3\nu_3-2\nu_4)} \\
 +\frac1{\nu_3(\nu_2-\nu_3-\nu_4)(6\nu_1-\nu_2-\nu_3-2\nu_4)}
 \Big)\ .
\end{multline}
We refrain from giving the explicit expression for the resulting constant maps terms, which also follow directly from the results above.

\subsection{\texorpdfstring{$A_2$}{A2} geometry}
The $A_2$ geometry is the minimal resolution of the singular space $\BC^2 / \BZ_3$, which is the cone over the Lens space S$^3/\BZ_3$. As a toric quotient, it is defined by the charge matrix
\be
Q =
\begin{pmatrix}
1 & -2 & 1 & 0 \\
0 & 1 & -2 & 1
\end{pmatrix}
\ee
in the chamber $t^1,t^2>0$. The equivariant volume  is given by, \cite{Cassia:2022lfj},
\be
\ba
 \VVV(t,\e) =& \oint_\text{JK} \frac{\mathd\phi_1\mathd\phi_2}{(2\pi\mathi)^2}
 \frac{\mathe^{\phi_1 t^1+\phi_2 t^2}}
 {(\e_1+\phi_1)(\e_2-2\phi_1+\phi_2)(\e_3+\phi_1-2\phi_2)(\e_4+\phi_2)} \\
 =& -\frac{\mathe^{-t^1\e_1-t^2(2\e_1+\e_2)}}
 {(3\e_1+2\e_2+\e_3)(2\e_1+\e_2-\e_4)}
 -\frac{\mathe^{-t^1\e_1-t^2\e_4}}{(\e_1-\e_3-2\e_4)(2\e_1+\e_2-\e_4)} \\
 &+\frac{\mathe^{-t^2\e_4-t^1(\e_3+2\e_4)}}{(\e_1-\e_3-2\e_4)(\e_2+2\e_3+3\e_4)} \ ,
\ea
\ee
where we took the poles for $(\phi_1,\phi_2)$ at $(-\e_1, -\e_2-2 \e_1)$, $(-\e_1, -\e_4)$ and $(-\e_3-2 \e_4, -\e_4)$. In the $\lam$-parametrization we have
\be
\ba
 \BV(\lam,\e) =& -\frac{\mathe^{\lam^3(3\e_1+2\e_2+\e_3)-\lam^4(2\e_1+\e_2-\e_4)}}
 {(3\e_1+2\e_2+\e_3)(2\e_1+\e_2-\e_4)}
 -\frac{\mathe^{\lam^2(2\e_1+\e_2-\e_4)-\lam^3(\e_1-\e_3-2\e_4)}}
 {(\e_1-\e_3-2\e_4)(2\e_1+\e_2-\e_4)} \\
 &+\frac{\mathe^{\lam^1(\e_1-\e_3-2\e_4)+\lam^2(\e_2+2\e_3+3\e_4)}}
 {(\e_1-\e_3-2\e_4)(\e_2+2\e_3+3\e_4)} \ .
\ea
\ee
If we now impose the mesonic twist conditions, $\lam_\text{mes.}^1+\lam_\text{mes.}^3=2\lam_\text{mes.}^2$ and $\lam_\text{mes.}^2+\lam_\text{mes.}^4=2\lam_\text{mes.}^3$, we find
\be
 \BV(\lam,\e) = \mathe^{\e_i\lam_\text{mes.}^i}\, \VVV(0,\e)
 = \frac{3\,\mathe^{\e_i\lam_\text{mes.}^i}}{(3\e_1+2\e_2+\e_3)(\e_2+2\e_3+3\e_4)}\ .
\ee
The toric fan can be defined by the four vectors\footnote{In order to preserve some symmetry of the choice of parametrization, we do not fix the first column vector of $v$ to be with constant entries equal to 1. However, it is clear that the constant vector is spanned by a linear combination of those in \eqref{eq:v-matrix-A2}.}
\be
\label{eq:v-matrix-A2}
 v = 
 \begin{pmatrix}
 1 & 0 \\
 \frac23 & \frac13 \\
 \frac13 & \frac23 \\
 0 & 1
 \end{pmatrix}
\ee
such that a possible parametrization in terms of variables $\mu$ and $\nu$ can be chosen as
\be
\begin{array}{c}
 \lam_\text{mes.}^1 = \mu^1\ ,
 \quad \lam_\text{mes.}^2 =\frac23 \mu^1+ \frac13 \mu^2\ ,
 \quad \lam_\text{mes.}^3 = \frac13 \mu^1 +\frac23 \mu^2\ ,
 \quad \lam^4 = \mu^2\ , \\ \\
 \nu_1 = \frac13\, (3 \e_1 + 2 \e_2 + \e_3)\ ,
 \qquad \nu_2 = \frac13\, (\e_2 + 2 \e_3 +3 \e_4)\ ,
\end{array}
\ee
such that 
\be
 \BV^\text{mes.} (\mu,\nu) = \frac{\mathe^{\nu_\alpha\mu^\alpha}}{3\,\nu_1\nu_2}\ .
\ee
We thus find the equivariant volume of $\BC^2/\BZ_3$,
as expected after the mesonic twist, forgets about the information of the blown up singularity,  
\be
 \BV^\text{mes.}_{A_2} = \frac13\BV_{\BC^2} = \BV_{\BC^2/\BZ_3}\ .
\ee

\subsection{Local \texorpdfstring{$\BP^1\times\BP^1$}{P1xP1}}
Consider the three-fold $X = K_{F_0}$, the canonical bundle over the Hirzebruch surface $ F_0 = \BP^1\times\BP^1$, which is defined by the charge matrix
\be
Q =
\begin{pmatrix}
1 & 1 & 0 & 0 & -2 \\
0 & 0 & 1 & 1 & -2
\end{pmatrix}
\ee
and corresponds to the resolution of the cone over some Sasaki manifold $L$, which can be written as a circle bundle over the base manifold $B = \BP^1\times\BP^1$.

We choose the chamber $t^1 > 0, t^2 > 0$, such that we can evaluate the equivariant volume to
\be
\ba
 \VVV(t,\e) =& \oint_\text{JK} \frac{\mathd\phi_1\mathd\phi_2}{(2\pi\mathi)^2}
 \frac{\mathe^{\phi_1 t^1+\phi_2 t^2}}
 {(\e_1+\phi_1)(\e_2+\phi_1)(\e_3+\phi_2)(\e_4+\phi_2) (\e_5-2 \phi_1 - 2 \phi_2)} \\
 =& \frac{\mathe^{-t^1\e_1-t^2 \e_3}}{(\e_2-\e_1)(\e_4-\e_3) (\e_5 + 2 \e_1 + 2 \e_3)} +  \frac{\mathe^{-t^1\e_1-t^2 \e_4}}{(\e_2-\e_1)(\e_3-\e_4) (\e_5 + 2 \e_1 + 2 \e_4)}\\
 &+ \frac{\mathe^{-t^1\e_2-t^2 \e_3}}{(\e_1-\e_2)(\e_4-\e_3) (\e_5 + 2 \e_2 + 2 \e_3)} +  \frac{\mathe^{-t^1\e_2-t^2 \e_4}}{(\e_1-\e_2)(\e_3-\e_4) (\e_5 + 2 \e_2 + 2 \e_4)} \ ,
\ea
\ee
where we took the residues at the poles $(\phi_1, \phi_2)$ equal to $(-\e_1, -\e_3)$, $(-\e_1, -\e_4)$, $(-\e_2, -\e_3)$ and $(-\e_2, -\e_4)$. In the $\lam$-parametrization we have
\be
\ba
 \BV(\lam,\e) =&\frac{\mathe^{\lam^2 (\e_2 - \e_1) + \lam^4 (\e_4-\e_3) + \lam^5 (\e_5+2 \e_1+2 \e_3)}}{(\e_2-\e_1)(\e_4-\e_3) (\e_5 + 2 \e_1 + 2 \e_3)} +  \frac{\mathe^{\lam^2 (\e_2 - \e_1) + \lam^3 (\e_3-\e_4) + \lam^5 (\e_5+2 \e_1+2 \e_4)}}{(\e_2-\e_1)(\e_3-\e_4) (\e_5 + 2 \e_1 + 2 \e_4)}\\
 &+ \frac{\mathe^{\lam^1 (\e_1 - \e_2) + \lam^4 (\e_4-\e_3) + \lam^5 (\e_5+2 \e_2+2 \e_3)}}{(\e_1-\e_2)(\e_4-\e_3) (\e_5 + 2 \e_2 + 2 \e_3)} +  \frac{\mathe^{\lam^1 (\e_1 - \e_2) + \lam^3 (\e_3-\e_4) + \lam^5 (\e_5+2 \e_2+2 \e_4)}}{(\e_1-\e_2)(\e_3-\e_4) (\e_5 + 2 \e_2 + 2 \e_4)} \ .
\ea
\ee
After imposing the mesonic twist conditions, $\lam_\text{mes.}^1+\lam_\text{mes.}^2=\lam_\text{mes.}^3+\lam_\text{mes.}^4=2 \lam_\text{mes.}^5$, we find
\be
\ba
 \BV(\lam_\text{mes.},\e)
 &= \mathe^{\e_i\lam_\text{mes.}^i}\, \VVV(0,\e) \\
 &= \frac{8\,\mathe^{\e_i\lam^i} \sum_{i=1}^5 \e_5}{ (2 (\e_1 + \e_3) +\e_5)  (2 (\e_1 + \e_4) +\e_5)  (2 (\e_2 + \e_3) +\e_5)  (2 (\e_2 + \e_4) +\e_5)}\ .
\ea
\ee
The toric fan can be parametrized by the five vectors
\be
 v =
 \begin{pmatrix}
 1 & 1 & 1 \\
 1 & 0 & 0 \\
 1 & 1 & 0 \\
 1 & 0 & 1 \\
 1 & \tfrac12 & \tfrac12
 \end{pmatrix}
\ee
such that a possible parametrization in terms of $\mu$ and $\nu$ is given by
\be
 \lam_\text{mes.}^i = v^i_\alpha \mu^\alpha\ , \quad  \nu_1 = \sum_{i=1}^5 \e_i\ ,
 \quad \nu_2 = \e_1 + \e_3 + \frac12\e_5\, \quad \nu_3 = \e_1 + \e_4 + \frac12 \e_5\ .
\ee
We thus find
\be
 \BV^\text{mes.} (\mu,\nu) = \frac{\mathe^{\nu_\alpha \mu^\alpha} \nu_1}
{2 \nu_2 \nu_3 (\nu_1 - \nu_2)  (\nu_1 - \nu_3)}\ .
\ee
which is proportional to $\BV^\text{mes.}$ of the resolved conifold in \eqref{eq:conifoldmesonic}.

\subsection{Cone over \texorpdfstring{$Q^{1,1,1}$}{Q111}}

The four-fold obtained as the total space of the rank-2 vector bundle $\cO (-1,-1) \oplus \cO (-1,-1) $ over $\BP^1 \times \BP^1$ is defined by the charge matrix
\be
Q =
\begin{pmatrix}
1 & 1 & 0 & 0 & -1 & -1 \\
0 & 0 & 1 & 1 & -1 & -1
\end{pmatrix}
\ee
and corresponds to the resolution of the cone over $L = Q^{1,1,1}$, which itself can be written as a circle bundle over the manifold $B=\BP^1\times\BP^1\times\BP^1$ \cite{Fabbri:1999hw}.

We choose the chamber $t^1 > 0, t^2 > 0$, such that we evaluate the equivariant volume to
\be
\ba
 \VVV(t,\e) =& \oint_\text{JK} \frac{\mathd\phi_1\mathd\phi_2}{(2\pi\mathi)^2}
 \frac{\mathe^{\phi_1 t^1+\phi_2 t^2}}
 {(\e_1+\phi_1)(\e_2+\phi_1)(\e_3+\phi_2)(\e_4+\phi_2) (\e_5-\phi_1 - \phi_2) (\e_5-\phi_1-\phi_2)} \\
 =& \frac{\mathe^{-t^1\e_1-t^2 \e_3}}{(\e_2-\e_1)(\e_4-\e_3) \prod_{j=5}^6 (\e_j+ \e_1+\e_3)} +  \frac{\mathe^{-t^1\e_1-t^2 \e_4}}{(\e_2-\e_1)(\e_3-\e_4) \prod_{j=5}^6 (\e_j + \e_1 + \e_4) }\\
 &+ \frac{\mathe^{-t^1\e_2-t^2 \e_3}}{(\e_1-\e_2)(\e_4-\e_3) \prod_{j=5}^6 (\e_j+ \e_2+\e_3)} +  \frac{\mathe^{-t^1\e_2-t^2 \e_4}}{(\e_1-\e_2)(\e_3-\e_4) \prod_{j=5}^6 (\e_j+ \e_2+\e_4)} \ ,
\ea
\ee
where we took the poles for $(\phi_1, \phi_2)$ at $(-\e_1, -\e_3), (-\e_1, -\e_4), (-\e_2, -\e_3), (-\e_2, -\e_4)$. In the $\lam$-parametrization we find
\be
\ba
 \BV(\lam,\e) =&\frac{\mathe^{\lam^2 (\e_2 - \e_1) + \lam^4 (\e_4-\e_3) + \sum_{j=5}^6 \lam^j (\e_j+ \e_1+\e_3)}}{(\e_2-\e_1)(\e_4-\e_3)  \prod_{j=5}^6 (\e_j+ \e_1+\e_3)} +  \frac{\mathe^{\lam^2 (\e_2 - \e_1) + \lam^3 (\e_3-\e_4) + \sum_{j=5}^6 \lam^j (\e_j+ \e_1+\e_4)}}{(\e_2-\e_1)(\e_3-\e_4) \prod_{j=5}^6 (\e_j+ \e_1+\e_4)}\\
 &+ \frac{\mathe^{\lam^1 (\e_1 - \e_2) + \lam^4 (\e_4-\e_3) + \sum_{j=5}^6 \lam^j (\e_j+ \e_2+\e_3)}}{(\e_1-\e_2)(\e_4-\e_3) \prod_{j=5}^6 (\e_j+ \e_2+\e_3)} +  \frac{\mathe^{\lam^1 (\e_1 - \e_2) + \lam^3 (\e_3-\e_4) + \sum_{j=5}^6 \lam^j (\e_j+ \e_2+\e_4)}}{(\e_1-\e_2)(\e_3-\e_4) \prod_{j=5}^6 (\e_j+ \e_2+\e_4)} \ .
\ea
\ee

If we now impose the mesonic twist conditions, $\lam_\text{mes.}^1+\lam_\text{mes.}^2=\lam_\text{mes.}^3+\lam_\text{mes.}^4=\lam_\text{mes.}^5 + \lam_\text{mes.}^6$, we find
\be
 \BV(\lam_\text{mes.},\e) = \mathe^{\e_i\lam_\text{mes.}^i}\, \VVV(0,\e)\ .
\ee
We can parametrize the toric fan by the six vectors
\be
 v =
 \begin{pmatrix}
 1 & 1 & 1 & 1 \\
 1 & 0 & 0 & 0 \\
 1 & 1 & 0 & 1 \\
 1 & 0 & 1 & 0 \\
 1 & 1 & 1 & 0 \\
 1 & 0 & 0 & 1
 \end{pmatrix}
\ee
suggesting a possible parametrization for $\mu$ and $\nu$. We refrain from giving the explicit form of $ \BV^\text{mes.}$ and the resulting constant maps terms, which follow straightforwardly from the formulae above but are somewhat cumbersome.

\section{Fibration over \texorpdfstring{$\WPL$}{WPL}}
\label{sec:4}

In this section we consider spaces $Y$ of the form of a fibration of a toric $d$-fold $X$ over the weighted projective line $\WPL$:
\be
\label{eq:XYWPL}
\begin{array}{ccc}
 X & \to & Y \\
 && \downarrow \\
 && \WPL
\end{array}
\ee
We will be primarily interested in the case when $X$ is a CY manifold, whose equivariant volume obeys the mesonic twist condition. However, we also explore some compact fibers towards the end of this section, in order to gain some intuition for the black hole solutions discussed in section \ref{sec:6.2}.

As previously discussed, the mesonic twist in the fiber corresponds to a blow down of the non-trivial two-cycles of $X$ that results in $X$ becoming a singular space $X_\circ\cong C(L)$.
Applying the mesonic twist fiberwise, however, does not imply that the total space $Y$ also becomes a cone. In fact, if $Y$ is toric, we can generically expect that it will have an additional non-trivial two-cycle corresponding to the base of the fibration $\WPL$. Only when we impose a further mesonic twisting of the K\"ahler modulus of the spindle, we find that $Y$ becomes a cone $Y_\circ\cong C(M)$, where $M$ is the total space of the fibration of $L$ over $\WPL$, cf.\ \eqref{eq:Sigmafibrations}. We show explicitly how this works at the end of our first example, when $X = \BC^d$.

\subsection{\texorpdfstring{$\BC^d$}{Cd} fiber}
\label{subsec:4.2}
Let us first consider the geometry of $\BC^d$ fibered over the spindle.
Both the fiber and the base have been discussed separately in the basic examples in Section~\ref{sec:3}, so now we can directly deal with the total space $Y$ in \eqref{eq:XYWPL}.

Fibering $\BC^d$ over the spindle, corresponds to the direct sum of $d$ line bundles $L_i=\mathcal{O}_{\WPL}(-n_i)$ with Chern classes such that
\be
 \int_{\WPL} c_1(L_i) = -n_i \in \BZ_{\leq 0}\ .
\ee
The total space of the fibration can be described as a K\"ahler quotient $Y=\BC^{d+2}//U(1)$, associated to the charge matrix
\be
 Q=(-n_1,-n_2, \dots, -n_d,a,b)\ ,
\ee
for a choice of one K\"ahler modulus $t>0$.
The quotient $Y$ is CY if it satisfies
\be
\label{eq:CYforY}
 \sum_{i=1}^d n_i = a+b\ .
\ee
The equivariant volume of $Y$ is
\be
\begin{aligned}
 \VVV_Y(t,\e) &= \oint\frac{\mathd\phi}{2\pi\mathi}
 \frac{\mathe^{\phi t}}
 {(\e_++a\phi)(\e_-+b\phi) \prod_{i=1}^d (\e_i-n_i\phi)} \\
 &= \frac{a^{d}\, \mathe^{-\frac{\e_+ t}{a}}}
 {(a\e_--b\e_+)\, \prod_{i=1}^d(a\e_i+n_i\e_+)}
 - \frac{b^{d}\, \mathe^{-\frac{\e_- t}{b}}}
 {(a\e_--b\e_+)\, \prod_{i=1}^d(b\e_i+n_i\e_-)}\ ,
\end{aligned}
\ee
where $\e_i$ are the equivariant parameters on $\BC^d$, $\e_\pm$ those of the base $\WPL$ as in Section~\ref{sec:spindle}, and we took the poles at $\phi = - \frac{\e_+}{a}$ and $\phi = - \frac{\e_-}{b}$. We observe that the answer has a natural interpretation in terms of the equivariant volume of $\BC^d$, \eqref{eq:Cequivol},
\be
 \VVV_Y(t,\e)
 = \frac{\mathe^{-\frac{\e_+t}{a}}}{\NN}
 \VVV_{\BC^d}\left(\e_i+\frac{n_i}{a}\e_+\right)
 - \frac{\mathe^{-\frac{\e_-t}{b}}}{\NN}
 \VVV_{\BC^d} \left(\e_i+\frac{n_i}{b}\e_-\right)\ ,
\ee
where we already used the definition \eqref{eq:definingsimplifiedparameterforthespindle}, $\NN = a\e_--b\e_+$, in this case simply as short-hand since $\VVV_Y$ is still understood as a function of $\e_\pm$. In this case, the fiber has no K\"ahler moduli, but its equivariant parameters $\e_i$ get shifted by the fluxes $n_i$ multiplied by the equivariant parameters of the normal bundle to the fixed point in $\WPL$.

The equivariant volumes in the $\lam$-parametrization enjoy a similar relation, which can be seen via \eqref{eq:maineqforcalculationglamparametrization}:
\be
\label{eq:MZequivolCnoverspindle}
\begin{aligned}
 \BV_Y(\lam,\e) &= \frac{\mathe^{(\e_i + \frac{n_i}{a} \e_+) \lam^i + \frac{\NN}{a} \lam^-}}
 {\NN \prod_{i=1}^d(\e_i+\frac{n_i}{a}\e_+)}
 - \frac{\mathe^{(\e_i + \frac{n_i}{b} \e_-) \lam^i - \frac{\NN}{b} \lam^+}}
 {\NN \prod_{i=1}^d(\e_i+\frac{n_i}{b}\e_-)} \\
 & = \frac{\mathe^{\frac{\NN}{a} \lam^-}}{\NN}\,
 \BV_{\BC^d} \left(\lam^i, \e_i+\frac{n_i}{a}\e_+\right)
 - \frac{\mathe^{-\frac{\NN}{b} \lam^+}}{\NN}\,
 \BV_{\BC^d} \left(\lam^i, \e_i+\frac{n_i}{b}\e_-\right) \ .
\end{aligned}
\ee
In analogy to the preferred parametrization in \cite{Martelli:2023oqk},
which is also useful holographically, we define
\be
 \e^+_i := \e_i + \frac{n_i}{a}\, \e_+\ ,
 \qquad
 \e^-_i := \e_i + \frac{n_i}{b}\, \e_-\ ,
\ee
as well as
\be
 \lam^i_+ := \lam^i + \frac{\NN}{a\, d\, \e^+_i}\, \lam^-\ ,
 \qquad
 \lam^i_- := \lam^i - \frac{\NN}{b\, d\, \e^-_i}\, \lam^+\ ,
\ee
such that we can rewrite the equivariant volume of the total space as
\be
 \BV_Y (\lam, \e ) = \frac1{\NN}
 \left( \BV_{\BC^d} (\lam^i_+, \e^+_i) -  \BV_{\BC^d} (\lam^i_-, \e^-_i) \right)\ .
\ee
This allows us in turn to write down the equivariant constant maps terms of the total space as
\be
\label{eq:equivariantconstantmapsonYcomingfromC}
\begin{aligned}
 N^Y_{0, 0} =&\, \frac1{\NN} \left( N_{0, 0}^{\BC^d} (\lam^i_+, \e^+_i)
 -  N_{0, 0}^{\BC^d} (\lam^i_-, \e^-_i) \right)\ , \\
 N^Y_{1, 0} =&\, \frac1{\NN} \left( N_{1, 0}^{\BC^d} (\lam^i_+, \e^+_i)
 -  N_{1, 0}^{\BC^d} (\lam^i_-, \e^-_i) \right) \\
 &+\frac1{24}\, \left( (\e^+_i \lam^i_+)
 \frac{\sum_i \e^+_i}{a\prod_i \e^+_i}
 + (\e^-_i \lam^i_-)
 \frac{\sum_i \e^-_i}{b\prod_i \e^-_i}
 \right)\ , \\
 N^Y_{\mathfrak{g} > 1, 0} =&\, \frac1{\NN}
 \left( N_{\mathfrak{g}, 0} ^{\BC^d} (\lam^i_+, \e^+_i)
 - N_{\mathfrak{g}, 0}^{\BC^d} (\lam^i_-, \e^-_i) \right) \\
 &+ \frac{|B_{2\mathfrak{g}}|}{4\mathfrak{g}}\frac{|B_{2\mathfrak{g}-2}|}{(2\mathfrak{g}-2)}
 \frac{(-1)^\mathfrak{g}}{(2\mathfrak{g}-2)!}
 \left(\frac{\sum_{i < j} \e^+_i \e^+_j}{a \prod_i \e^+_i}
 + \frac{\sum_{i < j} \e^-_i \e^-_j}{b \prod_i \e^-_i} \right)\ ,
\end{aligned}
\ee
where the explicit expressions for $N_{\mathfrak{g}, 0} ^{\BC^d}$ can be found in \eqref{eq:cplanemaps0and1}-\eqref{eq:cplanemapsg}.

\subsubsection*{Mesonic twist on \texorpdfstring{$Y$}{Y}}
Let us now illustrate the mesonic twist condition on the total space $Y$, which will provide an explicit realization of the relation to $C(M)$. At the level of the $\lam$ parameters, the mesonic twist simply gives the additional relation,
\be
 a\, \lam^+_\text{mes.} + b\, \lam^-_\text{mes.} = \sum_{i = 1}^d n_i \lam^i_\text{mes.}\ ,
\ee
such that we find
\be
 \BV_Y (\lam_\text{mes.}, \e ) = \frac{\mathe^{\e_+ \lam^+_\text{mes.} +\e_-
 \lam^-_\text{mes.} + \e_i \lam^i_\text{mes.}}}{\NN}
 \left( \frac1{\prod_i \e^+_i} - \frac1{\prod_i \e^-_i} \right)\ .
\ee
The mesonic constraint on the $\e$ parameters, \eqref{eq:mesonicconstraint}, can be written as follows:
\be
\label{eq:mesconicconstraintonY}
\begin{array}{c}
 a\, \e_- = \frac{\NN\, \prod_j \e^-_j}{\prod_j \e^-_j - \prod_k \e^+_k}\ ,
 \qquad
 b\, \e_+ = \frac{\NN\, \prod_k \e^+_k}{\prod_j \e^-_j - \prod_k \e^+_k}\ , \\
 \\
 \e_i = \frac{\e^+_i\, \prod_j \e^-_j -  \e^-_i\, \prod_k \e^+_k}{\prod_j \e^-_j - \prod_k \e^+_k}\ ,
\end{array}
\ee
which gives a single non-linear relation between the $\e$ parameters. 

Let us further specialize to the simplest case $d = 1$ in order to fully flesh out the previous equations. We further choose the case of the smooth sphere for simplicity, such that
\be
    a = b= 1\ , \qquad \Rightarrow \qquad n = 2\ ,
\ee
where we suppress the index $i$ as it is no longer needed. In this case we have $X = \BC$ simply being the cone over the circle, such that $L = \mathrm{S}^1$, $X = C(L)$. The fibration of $L$ over the two-sphere with charge $n=2$ simply produces the second lens space,
\be
    M = \mathrm{S}^3/\BZ_2 = L(2; 1)\ .
\ee
In turn, the cone over $M$ is simply given by 
\be
    C(M) = \BC^2/\BZ_2\ ,
\ee
whose equivariant volume and corresponding constant maps terms follow from the discussion in Section~\ref{sec:3}. On the other hand, the mesonic constraint on $Y$, in the case when $X = \BC$ above, \eqref{eq:mesconicconstraintonY}, can be easily seen to reduce to
\be
    \e = 0\ ,
\ee
with $\e_\pm$ unconstrained. Plugging this back into the equivariant volume, we find
\be
    \BV_Y (\lam_\text{mes.}, \e) = \frac{\mathe^{\e_+ \lam^+_\text{mes.} +\e_- \lam^-_\text{mes.}}}{2\, \e_+ \e_-},
\ee
which is precisely matching with $\BV_{\BC^2/\BZ_2} (\lam, \e)$ upon the identification of $\pm$ indices with the ${1,2}$ indices of $\BC^2$. The precise match between the equivariant constant maps terms \eqref{eq:equivariantconstantmapsonYcomingfromC} and \eqref{eq:cplanemaps0and1} also follows simply from the same identification.

\subsection{General and mesonic fiber}

Let us now consider a more general fibration, where the fiber $X$ is a toric CY $d$-fold with equivariant volume $\BV_X (\lam^i, \e_i)$.
We will assume for the moment that the fibration of $X$ over the spindle also has a description as a toric symplectic quotient, and specifically we will assume that the corresponding matrix of charges for the total space $Y$ can be written schematically as the block triangular matrix
\be
 Q_Y = \left[
 \begin{array}{cc|c}
 a & b & -n_i \\
 \hline
  &  & Q_X
 \end{array}
 \right]
\ee
where $Q_X$ is the matrix of charges of the fiber and $n_i$ are (non-negative) integers describing the fibration.~\footnote{The structure group of the fibration corresponds to the torus group of isometries of the fiber, which we identify with $U(1)^n$. The charges $n_i$ then represent the first Chern numbers of the associated principal $U(1)$-bundles. Since not all of these $U(1)$'s act effectively, the topology of the total space $Y$ only depends on $d=\dim X$ linear combinations of the $n_i$, as in \eqref{eq:reducedintermsoforiginalcharges}. Another way to see this, is to observe that $Q_Y$ is only defined up to the left-action of $SL_{r+1}(\BZ)$ which can be used to shift the vector of $n_i$ by the rows of the matrix $Q_X$.} In this description, $Y$ can be explicitly identified with the quotient $\BC^{n+2}//U(1)^{r+1}$ of complex dimension $n-r+1=d+1$ ($d$ being the dimension of $X$), which is CY provided that the charges satisfy the condition
\be
\label{eq:CYconditiononchargen}
	\sum_i n_i = a + b\ .
\ee
The equivariant volume of $Y$ is then given by the contour integral
\be
\label{eq:VYfromVX}
\ba
 \BV_Y(\lam,\e)
 &= \oint\frac{\mathd\phi}{2\pi\mathi}
 \frac{\mathe^{(\e_++a\phi)\lam^++(\e_-+b\phi)\lam^-}}{(\e_++a\phi)(\e_-+b\phi)}
 \exp\left(-n_i\phi\frac{\partial}{\partial\e_i}\right)\cdot
 \BV_X(\lam,\e) \\
 &= \oint\frac{\mathd\phi}{2\pi\mathi}
 \frac{\mathe^{(\e_++a\phi)\lam^++(\e_-+b\phi)\lam^-}}{(\e_++a\phi)(\e_-+b\phi)}
 \BV_X(\lam^i,\e_i-n_i \phi) \\
 &= \frac{\mathe^{\frac{\NN}{a}\lam^-}}{\NN}\, \BV_X(\lam^i,\e_i+\frac{n_i}{a}\e_+)
  - \frac{\mathe^{-\frac{\NN}{b}\lam^+}}{\NN}\, \BV_X(\lam^i,\e_i+\frac{n_i}{b}\e_-)\ ,
\ea
\ee
with the same parameters $\lam^\pm$ and $\NN$ on the spindle.
Without further assumptions on the form of $\BV_X$ this expression has no obvious further simplification, and the constant maps terms follow directly from the generic formula.

\subsubsection*{Mesconic fiber}
We can be more specific when $X$ has blown down two cycles, i.e.\ subjected to the mesonic twist condition, \eqref{eq:mesonictwist}. Suppose further that the equivariant volume of $X$ is known and given by a function $\BV_X^\text{mes.}(\mu,\nu)$, where $\nu_\alpha$ are a set of equivariant parameters associated to the $d$ holomorphic toric isometries of $X$, i.e.\ they are \emph{effective} equivariant parameters in the fiber, while $\mu^\alpha$ are \emph{ineffective} K\"ahler parameters\footnote{Because of the mesonic twist condition, there are no \emph{effective} K\"ahler parameters in the fiber.} conjugate to the $\nu_\alpha$. We therefore know that the general expression for the equivariant volume of $X$ is given by
\be
\label{eq:mesonicequivolparametrizationCnu}
	\BV_X^\text{mes.} (\mu, \nu) = \mathe^{\nu_\alpha \mu^\alpha}\, \VVV_X (0, \e (\nu) )\ ,
\ee
where $\VVV_X$ depends on further details of the manifold $X$, a particular example (apart from the complex plane discussed previously) is given in \eqref{eq:conifoldmesonic}.
 
The equivariant volume of the total space $Y$ can be expressed in terms of that of $X$ via the following integral formula
\be
\label{eq:MZequivolmesonicoverspindle}
\ba
 \BV_Y(\lam^\pm,\e_\pm,\mu,\nu)
 &= \oint\frac{\mathd\phi}{2\pi\mathi}
 \frac{\mathe^{(\e_++a\phi)\lam^++(\e_-+b\phi)\lam^-}}{(\e_++a\phi)(\e_-+b\phi)}
 \exp\left(-p_\alpha\phi\frac{\partial}{\partial\nu_\alpha}\right)\cdot
 \BV^\text{mes.}_X(\mu,\nu) \\
 &= \oint\frac{\mathd\phi}{2\pi\mathi}
 \frac{\mathe^{(\e_++a\phi)\lam^++(\e_-+b\phi)\lam^-}}{(\e_++a\phi)(\e_-+b\phi)}
 \BV^\text{mes.}_X(\mu,\nu_\alpha-p_\alpha\phi) \\
 &= \frac{\mathe^{\frac{\NN}{a} \lam^-}}{\NN}\, \BV^\text{mes.}_X(\mu,\nu_\alpha+\frac{p_\alpha}{a}\e_+)
  - \frac{\mathe^{-\frac{\NN}{b} \lam^+}}{\NN}\,  \BV^\text{mes.}_X(\mu,\nu_\alpha+\frac{p_\alpha}{b}\e_-)\ ,
\ea
\ee
again using $\NN = a\e_- -b\e_+$, where the contour has been chosen so that it encircles the poles at $\phi=-\frac{\e_+}{a}$ and $\phi=-\frac{\e_-}{b}$ and there are no other poles coming from $\BV_X^\text{mes.}(\mu,\nu)$.
The variables $p_\alpha$ parametrize the fibration of $X$ over the spindle in terms of the Chern classes of the associated principal $U(1)$ bundles. Note that the original charges $n_i$, subject to \eqref{eq:CYconditiononchargen} implicitly depend on $p_\alpha$ via
\be
\label{eq:reducedintermsoforiginalcharges}
	p_\alpha = v^i_\alpha\, n_i\ . 
\ee
for $v_\alpha^i$ a toric fan matrix for the fiber $X$.
Roughly speaking, the charges $n_i$ describe the fibration in terms of the homogeneous coordinates of $X$, while the charges $p_\alpha$ describe the fibration in terms of its intrinsic coordinates.
Choosing the toric fan in the convention where $v_1 = (1, 1, \dots, 1)$, we thus find the CY condition on the reduced charges $p_\alpha$ as
\be
	p_1 = a + b\ ,
\ee
in agreement with \cite{Martelli:2023oqk}.~\footnote{Note that the charges in \cite{Martelli:2023oqk} are normalizaed to carry an additional inverse factor of $(a b)$ with respect to the present conventions.}
In analogy to the previous subsection, we can redefine equivariant parameters on the total space of the fibration as
\be
 \nu^+_\alpha := \nu_\alpha+\frac{p_\alpha}{a}\e_+\,,
 \qquad
 \nu^-_\alpha := \nu_\alpha+\frac{p_\alpha}{b}\e_-\ ,
\ee
where we notice that
\be
 \nu^-_\alpha - \nu^+_\alpha = \frac{p_\alpha}{ab}\, \NN\ ,
\ee
so that only $(d+1)$ of the $\nu$ parameters are linearly independent. We also use the redefinitions
\be
 \mu^\alpha_+ := \mu^\alpha + \frac{\NN}{a\, d\, \nu_{\alpha, +}}\, \lam^-\ .
 \qquad
 \mu^\alpha_- := \mu^\alpha - \frac{\NN}{b\, d\, \nu_{\alpha, -}}\, \lam^+\ ,
\ee
This allows us to rewrite the volume of the total space as
\be
 \BV_Y (\mu, \nu ) = \frac1{\NN}
 \left( \BV^\text{mes.}_X (\mu_+,\nu^+) - \BV^\text{mes.}_X (\mu_-,\nu^-) \right)\ .
\ee
In order to derive the equivariant constant maps terms, however, we need to be careful about the $\lam^i$ derivatives, similarly to the discussion of the mesonic constraint in Section~\ref{sec:2.3}.

We can alternatively look at the mesonic twist of $X$ in the original variables, $\lam$ and $\e$:
\be
\BV_X (\lam_\text{mes.}, \e) = \mathe^{\e_i \lam^i_\text{mes.}}\,\VVV^\text{mes.}_X (\e)\ ,
\ee
which from \eqref{eq:VYfromVX} would lead to
\be
 \BV_Y(\lam,\e)
 = \frac{\mathe^{\frac{\NN}{a}\lam^-}}{\NN}\, \BV_X(\lam^i_\text{mes.},\e_i+\frac{n_i}{a}\e_+)
  - \frac{\mathe^{-\frac{\NN}{b}\lam^+}}{\NN}\, \BV_X(\lam^i_\text{mes.},\e_i+\frac{n_i}{b}\e_-)\ .
\ee
As above, we can define
\be
 \e^+_i := \e_i+\frac{n_i}{a}\e_+\,, \qquad \e^-_i := \e_i+\frac{n_i}{b}\e_-\ ,
\ee
and
\be
 \lam^i_{+,\text{mes.}} := \lam^i_{\text{mes.}} + \frac{\NN}{a\, n\, \e^+_i}\, \lam^-\ ,
 \qquad
 \lam^i_{-,\text{mes.}} := \lam^i_{\text{mes.}} - \frac{\NN}{b\, n\, \e^-_i}\, \lam^+\ ,
\ee
where $n$ is number of homogenous coordinates on $X$ (not to be confused with the fluxes $n_i$).
We thus arrive at
\be
 \BV_Y (\lam^i_{\pm, \text{mes.}}, \e^\pm_i)
 = \frac1{\NN} \left( \BV_X (\lam^i_{+, \text{mes.}}, \e^+_i)
 - \BV_X (\lam^i_{-, \text{mes.}}, \e^-_i) \right)\ .
\ee
This parametrization uses the redundant set of charges, $n_i$, which are implicit functions of the charges $p_\alpha$ as explained around \eqref{eq:reducedintermsoforiginalcharges}. Using the explicit knowledge of $\BV_X (\lambda_\text{mes.}, \e)$, we can simplify the genus-zero and genus-one constant maps terms as
\be
\label{eq:spindlefibermestwistconstantmaps}
\ba
 N^Y_{0, 0} =&\, \frac1{\NN} \left( N_{0, 0} ^{X, \text{mes.}} (\lam^i_+, \e^+_i)
 -  N_{0, 0} ^{X, \text{mes.}} (\lam^i_-, \e^-_i) \right)\ , \\
 N^Y_{1, 0} =& \frac1{\NN} \left( N_{1, 0} ^{X, \text{mes.}} (\lam^i_+, \e^+_i)
 -  N_{1, 0} ^{X, \text{mes.}} (\lam^i_-, \e^-_i) \right) \\
 &+ \frac1{24\, a b} \left( b\, \sum_{i} \frac{\partial \BV_X (\lam^i,  \e^+_i)}{\partial \lam^i} \Big|_{\lam = 0}  (\e^+_i \lam_{+,\text{mes.}}^i) + a\, \sum_{i} \frac{\partial \BV_X (\lam^i,  \e^-_i)}{\partial \lam^i} \Big|_{\lam = 0}  (\e^-_i \lam_{-,\text{mes.}}^i) \right)
\ea
\ee
where the constant maps terms $N_{\mathfrak{g}, 0}^{X,\text{mes.}}$ of the fiber $X$ are given in \eqref{eq:mesonicconstmaps}, and we omitted the ``$\text{mes.}$'' label on the $\lam$ variables for brevity. 

\subsection*{Mesonic charge constraint}
Just as the $\e$ can be found explicitly (but not uniquely) in terms of $\nu$ via the mesonic constraint \eqref{eq:mesonicconstraint}, we can (but is not strictly needed if we directly use \eqref{eq:MZequivolmesonicoverspindle}) find a set of constraints on the charges $n_i$ to lift the redundancy.
As previously noticed in \cite{Hosseini:2019ddy}, one possible way to do this consists in imposing the mesonic charge constraints~\footnote{In the language of \cite{Hosseini:2019ddy}, the mesonic charge constraint can be thought of as acting with the operator $\sum_k n_k\, \partial_{\e_k}$ on the mesonic constraints \eqref{eq:mesonicconstraint}.}
\be
\label{eq:mesonicchargeconstraint}
 n_i \stackrel{!}{=} \left.\sum_k n_k\,
 \frac{\partial^2\log\BV_Y}{\partial\e_k\partial\lam^i}\right|_{\lam=\lam_{\rm mes.}} \ ,
 \qquad \text{for $i=1,\dots,n$.}
\ee
as a set of non-linear equations for the $n_i$.
The motivation for doing so is the following.
First, we notice that the derivative w.r.t.\ $\lam^i$ goes through the integral in \eqref{eq:VYfromVX} and only acts on the function $\BV_X$, and the mesonic twist on the $\lam^i$ is imposed only after taking the derivative.
We can then show that the equations in \eqref{eq:mesonicchargeconstraint} are not all linearly independent, in fact, if we contract with the matrix $v_\alpha^i$, we obtain
\be
 v_\alpha^i n_i
 = \sum_k n_k\,\frac{\partial}{\partial\e_k}
 v_\alpha^i\frac{\partial}{\partial\lam^i}\log\BV_Y
 = \sum_k n_k\,\frac{\partial}{\partial\e_k} v^i_\alpha\e_i
 = v^i_\alpha \sum_k \delta_i^k n_k
\ee
which is trivial even before the mesonic twist on $\lam$ due to \eqref{eq:vdlamBV}.
This implies that, as in the case of the constraints \eqref{eq:mesonicconstraint}, only $n-d$ of these equations are linearly independent, and they can be used to lift the redundancy from the $n_i$ parameters to the $p_\alpha$ parameters.
As mentioned, the specialization $\lam=\lam_{\rm mes.}$ is not strictly necessary for this procedure to work, however it does make the expression in the r.h.s.\ of \eqref{eq:mesonicchargeconstraint} considerably simpler.

To give a tractable example of this constraint, we can again consider the resolved conifold of Section~\ref{sec:4.2}, where we showed that the mesonic 
constraint simplifies to 
\be
 \e_1 \e_2 = \e_3 \e_4\ ,
\ee
such that the mesonic charge constraint becomes
\be
 \e_1 n_2 + \e_2 n_1 = \e_3 n_4 + \e_4 n_3\ ,
\ee
again in agreement with \cite{Hosseini:2019ddy}. A similar charge constraint can be imposed on all other examples in Section~\ref{sec:3.5} upon a further fibration over a spindle base as described here.

Finally, we can present an even more simplified expression genus-one constant maps terms,
assuming the mesonic constraints for $\e_i$, \eqref{eq:mesonicconstraint}, and $n_i$, \eqref{eq:mesonicchargeconstraint},
\be
\ba
 N^Y_{1, 0} =&\, \frac1{\NN} \left( N_{1, 0} ^{X, \text{mes.}} (\lam^i_+, \e^+_i)
 -  N_{1, 0} ^{X, \text{mes.}} (\lam^i_-, \e^-_i) \right) \\
 & +\frac1{24}\, \left( \frac{(\e^+_i \lam^i_+)}{a}\,\VVV^\text{mes.}_X (\e^+_i)
 \sum_i \e^+_i + \frac{(\e^-_i \lam^i_-)}{b}\,\VVV^\text{mes.}_X (\e^-_i)
 \sum_i \e^-_i \right)\ . 
 \ea
\ee

We refrain from spelling out the mesonic twist and constraint on the total space $Y$ in this case to avoid possible confusion with what we did so far, but it follows straightforwardly from the discussion in the previous subsection, around \eqref{eq:mesconicconstraintonY}.

\section{M2-brane backgrounds with flux}
\label{sec:5}
As explained in the introduction, our topological string examples are primarily motivated by charged brane configurations, particularly M2-brane constructions preserving at least $\cN = 2$ supersymmetry, see \cite{Gabella:2012rc} and references therein. While the relevance of cones over Sasakian manifolds extends beyond M2-branes, we focus on this specific case for clarity. It is important to emphasize that our approach from here on is guided by supersymmetry, beginning from the low-energy description. Although supergravity does not capture all details, as illustrated in Figures \ref{fig:new} and \ref{fig:1}, we rely on supersymmetry to fix the topology of the underlying manifold $X$. The immediate goal here is to clarify the claim that the $\lam$ parameters in the equivariant volume are conjugate to the fluxes $N_{\rm M2}$, proposed in terms of an extremization procedure in \cite{Couzens:2018wnk,Gauntlett:2018dpc,Hosseini:2019use,Martelli:2023oqk}. The main results of these references in practice illustrate the validity of equations \eqref{eq:mainconjecture-pert} and \eqref{eq:mainconjecture-pert-sigma} at leading order in the M2-brane charge $N$. In \cite{Cassia:2025jkr}, we elaborate in more detail on many of the aspects touched upon here and reformulate the relation as a full integral transform. The present section thus serves as an introduction to the companion paper, while also physically motivating the examples we have considered in the previous sections.

We start with the low-energy description of M-theory, the two derivative 11d supergravity action, whose bosonic fields are the metric $g_{\mu\nu}$ and a three-form field $C_3$ with a corresponding field strength $F_4 = \mathd C_3$. The bosonic action is given by
\be
 S_\mathrm{11d} = \frac{1}{2\, \kappa_\mathrm{11d}^2}
 \int \Big[ R \star 1 - \frac12\, F_4 \wedge \star F_4
 - \frac16\, C_3 \wedge F_4 \wedge F_4 \Big]\ ,
\ee
where $\kappa^2_\mathrm{11d} := 8 \pi G_\mathrm{11d} = (2 \pi l_P)^9 / 4 \pi$. The equations of motion and supersymmetry variations of this theory have been analyzed carefully in many references. The simplest and most prominent supersymmetric backgrounds correspond to the multiple spacetime filling M2 and M5-branes, but we will also be interested in cases where the branes wrap internal cycles in a supersymmetric way, see e.g.\ \cite{Gauntlett:2003di} for a pedagogical review. 

\subsection{Spacetime filling branes}
\label{sec:5.1}
\subsubsection*{11d background solution}
One type of supersymmetry-preserving (or BPS) 11d background of multiple spacetime filling M2-branes can be written as~\footnote{There are even more general backgrounds, which we have neglected in the presented form of the metric and three-form backgrounds, which correspond to further deformations of the underlying M2-brane theory, see e.g.\ \cite{Pope:2003jp,Bena:2004jw} and references thereof. Even though we do not consider the explicit 11d solutions, the topological string description still allows to identify such deformations in the dual field theory picture.}
\be
\label{eq:11d-metric}
 \mathd s^2 = H(r)^{-2/3}\, \eta_{ij} \mathd x^i \mathd x^j
 + H(r)^{1/3}\, [\mathd r^2 + r^2\, \mathd s^2 (L) ]\ ,
\ee
with 
\be
 H(r) = 1 +\alpha_2\, \frac{q}{r^6}\ ,
\ee 
and $\alpha_2 = 32 \pi^2 l_P^6$ in the standard normalization. Here the index $i = 0, 1, 2$ labels the coordinates along the spacetime directions where the branes are located. The second part of the metric corresponds to the cone over a 7d Sasaki--Einstein space $L$, which we denote as $X_\circ = C(L)$, with metric as in \eqref{eq:conicalmetricX}.
The transverse space $X_\circ$ to the M2-brane directions has a conical singularity at $r=0$, and is only smooth when $L$ is a sphere, S$^7$, with the round metric. In the general case, $X_\circ$ is a non-compact Calabi--Yau four-fold singularity. Some toric examples of resolutions of such spaces were given in section \ref{sec:3.5}. Importantly, the above background is a supersymmetric solution only in the presence of a non-vanishing three-form field, 
\be
 C_3 = H(r)^{-1}\, \mathd x^0 \wedge \mathd x^1 \wedge \mathd x^2\ .
\ee
In the near-horizon limit, $r \to 0$, one finds the smooth product space $\AdS_4\times L$ with both radii proportional to $q^{1/6}$, 
\be
 \mathd s^2_{r \to 0} = R^2_L\, \left( \frac14\, \mathd s^2 (\AdS_4)
 + \mathd s^2 (L)  \right)\ , \qquad F_4 = \frac38\, R^3_L\, \omega_{\AdS_4}\ ,
\ee
with $R_L = (\alpha_2\, q)^{1/6}$, and $\omega_{\AdS_4}$ the volume form with unit radius. We can canonically define the integer $N$ as the M2-brane electric charge (here at leading order as we consider the two derivative theory, $N_{\rm M2} = N$, see App.~\ref{app:A}), 
\be
\label{eq:4dN}
 N = \frac1{(2 \pi l_P)^6} \int_L \star F_4
 = \frac{6\, R_L^6}{(2 \pi l_P)^6}\, \VVV_L = \frac{3\, \VVV_L}{\pi^4}\,
 q \in \BZ\ ,
\ee
where $\VVV_L$ is the volume of $L$ for unit radius, such that the charge $N$ coincides with $q$ for the maximally supersymmetric case when $L$ is the round S$^7$. The three-form field plays a crucial role as its flux through the compact space balances its volume (and the curvature of the $\AdS_4$ factor) and does not allow the size of $L$ to shrink arbitrarily.

To clarify the holographic picture, schematically shown in Figure~\ref{fig:1}, the AdS/CFT correspondence applies precisely in the near-horizon limit, $\AdS_4 \times L$, and relates M-theory on this background to the gauge theory localized on the M2-branes.~\footnote{While we focus on M2-branes, the discussion can almost directly apply to D3 or M5-branes, with only minor differences due to their dimensionalities.} In the absence of a non-perturbative definition of the bulk theory—which we argue is captured by equivariant topological strings—the low-energy supergravity description is typically applied in the bulk and related to a strong coupling limit in the dual field theory. 11d supergravity typically allows for consistent truncations of finite sets of KK modes on the compact space $L$ to an effective four-dimensional gauged supergravity, see \cite{deWit:1986oxb,Gauntlett:2007ma,Gauntlett:2009zw,Cassani:2011fu,Nicolai:2011cy,Cassani:2012pj}. In this case, the effective Newton constant in the lower dimensional theory, $G_\mathrm{4d}$, depends explicitly on the internal volume, see e.g.\ \cite[Section 5]{Marino:2011nm}, 
\be
\label{eq:4dGN}
 \frac{R_L^2}{G_\mathrm{4d}} \propto \frac{N^{3/2}}{\sqrt{\VVV_L}}\ ,
\ee
with the combination on the left hand side setting the overall scale of 4d supergravity quantities (since $R_L$ is proportional to the $\AdS_4$ radius).~\footnote{In presence of higher-derivative corrections in 11d, we expext that the above equation holds with the substitution of $N$ with $N_{\rm M2}$, see App.~\ref{app:A}.}
Notice that the M2-brane charge $N$ is a natural parameter both in the field theory description,  corresponding to the gauge group rank, and in the supergravity description, via \eqref{eq:4dGN}. On the other hand, the geometry of $L$ itself only depends on the flux $N$ in its overall size and not topologically. This is the fundamental origin of the need for different ensembles between the topological string ($\lam$) and the field theory/supergravity ($N$ or $G_\mathrm{N}$) descriptions.

\subsubsection*{Off-shell geometry}
The dimensional reduction described above requires satisfying the equations of motion and supersymmetry conditions, which fix the metric on $L$ and all other background fields. In contrast, our current approach considers arbitrary metrics on a fixed topology, with topological string amplitudes allowing us to remain agnostic about the choice of metric. Specifically, the equivariant and K\"ahler parameters $\e_i$ and $\lam^i$ (or the reduced ones, $\nu_\alpha$ and $\mu^\alpha$), introduced in Section~\ref{sec:2}, probe the space of metrics on the topology of $X$, resolutions of $X_\circ$, which obey the Calabi--Yau condition for an arbitrary Sasakian metric on $L$. This enables us to describe a large moduli space of geometries, while the supergravity equations of motion select only isolated points within this space. It is therefore unsurprising that the precise relation between topological strings and supergravity involves an extremization procedure, as depicted in Figures~\ref{fig:new} and \ref{fig:1}. Various aspects of this idea, which can also be implemented via an off-shell quantity dubbed the \emph{master volume}, have been explored in a series of papers \cite{Martelli:2005tp,Butti:2005vn,Butti:2005ps,Martelli:2006yb,Amariti:2011uw,Couzens:2018wnk,Gauntlett:2018dpc,Hosseini:2019use,Hosseini:2019ddy,Gauntlett:2019roi,Gauntlett:2019pqg,Boido:2022mbe}. While we have focused on the $\lam$ parameters as a new type of ensemble, the equivariant parameters $\e$ play a similar role, typically requiring extremization to describe a true Lorentzian BPS background solution in (possibly off-shell) supergravity; see also \cite{BenettiGenolini:2019jdz,Hosseini:2019iad,Bobev:2020pjk,Panerai:2020boq,BenettiGenolini:2023ndb,Hristov:2024cgj,Cassani:2024kjn}.

Let us illustrate the above discussion explicitly to the present model, which ultimately leads to the geometric description of F-maximization, \cite{Jafferis:2010un,Jafferis:2011zi}. To allow for supersymmetric deformations without solving the supergravity equations of motion, \cite{Martelli:2005tp} proposed considering more general Sasakian manifolds that relax the Einstein condition, but are equipped with a unit norm Killing vector field, $\xi$, known as the Reeb vector \eqref{eq:Reeb}. The space $X$ corresponds to a (partial) resolution of the CY singularity described by the conical metric, $X_\circ$. In the toric case, the angles $\varphi_i \sim \varphi_i + 2 \pi$ can be chosen as part of the symplectic coordinate pairs $(H^i,\varphi_i)$ on $X$, with $H^i$ the associated moment maps. The symplectic form on $X$ is then given by
\be
 \omega = \frac{1}{2\pi}\mathd H^i \wedge \mathd \varphi_i\ ,
\ee
with an implicit summation over repeated indices. We thus see that the Reeb vector on $X$ actually parametrizes the standard torus action, such that $\e_i$ are the equivariant parameters in the sense of Section~\ref{sec:2}. In the case when we keep the 4d space maximally symmetric (i.e.\ in Euclidean signature pure $\AdS_4$ is sliced into round S$^3$ coordinates),~\footnote{In fact, we can also relax the assumption on the boundary metric without changing the internal topology by considering a squashed S$^3$ boundary that only preserves $U(1) \times U(1)$ isometry, see \cite{Martelli:2011fu}. This introduces an additional deformation parameter, usually denoted by $b$, which also generalizes the supersymmetric constraint. We elaborate further on these issues in \cite{Cassia:2025jkr}.} supersymmetry imposes a particular condition on these equivariant parameters,
\be
\label{eq:fluxsusyconstr}
	\sum_i \e_i = 2\ ,
\ee
which can also be seen as a condition on existence of Killing spinors on the $\AdS_4$ vacuum. 

As described in Section~\ref{sec:2.3}, one can also parametrize the metric on $L$ using the foliation defined by the Reeb vector and its dual one-form $\eta$,
\be
	 \mathd s^2 (L) = \eta^2 + \mathd s^2 (B)\ ,
\ee
with a K\"ahler form $\omega_B=\mathd\eta$ transverse to the Reeb foliation, see \cite{Couzens:2018wnk,Gauntlett:2018dpc,Hosseini:2019use,Hosseini:2019ddy,Gauntlett:2019roi} for details and \cite{Katmadas:2015ima} for a different approach to this construction. We can thus view the space $L$ as a circle fibration over a compact space $B$ of three complex dimensions. This gives a more direct geometric meaning of the $\lam$ parameters as responsible for the physical (non-equivariant) size of the compact space $B$, cf.\ the discussion around \eqref{eq:relationtoB}.

In what follows, we will also need to consider certain submanifolds of $L$.
Assuming that the cone over $L$ is indeed a toric CY, we can deduce that to each toric divisor $D_i$ of $C(L)$, i.e.\ the locus where the $i$-th homogeneous coordinate vanishes, one can associate a corresponding (real) codimension-2 submanifold $T_i$ in $L$, obtained as the intersection of the toric divisor and the base of the cone,
\be
\label{eq:Sasakian-divisor}
 T_i := D_i\cap L\,.
\ee
By analogy with the toric case, we will refer to the $T_i$'s as \emph{Sasakian divisors} of $L$.

Given this setup, there are many geometric relations that have been shown to hold between the Sasakian manifold $L$, the K\"ahler space $B$, and the cone $X_\circ$, see again Section~\ref{sec:2.3} and \cite{Martelli:2005tp,Butti:2005vn,Butti:2005ps,Martelli:2006yb,Amariti:2011uw,Couzens:2018wnk,
Gauntlett:2018dpc,Hosseini:2019use,Hosseini:2019ddy,Gauntlett:2019roi,Gauntlett:2019pqg,Boido:2022mbe}.
A quantity of particular interest is the Sasakian volume $\VVV_L (\e)$, which now depends on the parameters $\e_i$.~\footnote{Note that the Sasakian volume $\VVV_L (\e)$ should be thought of as an off-shell quantity and is only related to the quantity entering Eq's \eqref{eq:4dN} and \eqref{eq:4dGN} after a suitable extremization fixing the parameters $\e$ on their on-shell values.} Its relation to the equivariant volume was first shown in \cite{Martelli:2006yb} (passing through the equivariant index) and revisited in \cite{Martelli:2023oqk},~\footnote{Similarly to the definition of the equivariant volume, the Sasakian volume here is normalized differently with respect to \cite{Martelli:2006yb,Martelli:2023oqk} and the related references.}
\be
\label{eq:Sasakian-vol}
 \VVV_L (\e) = \VVV_X (t = 0, \e)\ ,
\ee
where we remark that, in our conventions, the equivariant parameters $\e_i$ are redundant, while the original definition of the Sasakian volume only depends on the effective equivariant parameters (they enter the computation of the Sasakian volume as coefficients of the Reeb vector on the basis of abelian isometries of $L$). In order to match the parametrizations of the volumes, we make use of the fact that the Sasakian volume can be obtained as the cohomological limit of the K-theoretic equivariant index, which we can naturally upgrade to a function of all equivariant parameters, including the non-effective ones. In this way, we can regard \eqref{eq:Sasakian-vol} as a definition of this upgraded function, which we obtain as the cohomological limit of the \emph{fully} equivariant index. In this way, the volume $\VVV_L (\e)$ reduces by construction to the usual definition of Sasakian volume once we lift the redundancy in the $\e$-parametrization.

We notice, moreover, that a similar definition of volume can be formulated for the Sasakian divisors $T_i$. In fact, since they arise as intersections between toric divisors in $X$ and the base $L$, one can define their volume as
\be
\label{eq:Sasakian-divisor-vol}
 \VVV_{T_i}(\e) = \VVV_{D_i} (t = 0, \e)\ ,
\ee
which is equivalent to \eqref{eq:Sasakian-vol}, once we substitute $T_i$ and $D_i$ for $L$ and $X$, respectively. These relations, \eqref{eq:Sasakian-vol}-\eqref{eq:Sasakian-divisor-vol}, already establish the link between the internal volume extremization program initiated in \cite{Martelli:2006yb} and the present approach.

A closely related object, which allows us to relate the $\lam$ parameters to the flux $N$, is the so called master volume $\cV_L$, originally defined as
\be
	\cV_L (\lam, \e) := \int_L \eta \wedge \frac{\omega_B^3}{3!}\ .
\ee
One of the main insights of \cite{Martelli:2023oqk} is the relation between the master volume and the equivariant intersection numbers, which in our approach directly result in the genus-zero constant maps contribution,
\be
	\cV_L (\lam_\text{mes.}, \e) = \frac1{(2 \pi)^2}\, N^{X, \text{mes.}}_{0, 0} (\lam_\text{mes.}, \e)\ ,
\ee
where the right hand side depends on the $\e$ and $\lam_\text{mes.}$ parameters lifted to $X$. 

The master volume plays a crucial role in holography, as it directly relates to the supergravity action and in turn to the dual holographic free energy. From the present point of view, the most important insight on the relation between $\lam$ and $N$ comes from a holographic argument proposed in \cite{Couzens:2018wnk,Gauntlett:2018dpc,Hosseini:2019use,Hosseini:2019ddy,Gauntlett:2019roi} which concerns certain supersymmetric submanifolds of $L$, which are identified with the Sasakian divisors \eqref{eq:Sasakian-divisor}.
Holography then predicts that the R-charges of the dual field theory operators associated to wrapping M5-branes on the $T_i$ are given by
\be
\label{eq:r-chargeassignmentsingeometry}
 R[T_i] := 2 \pi\, \int_{T_i} \eta \wedge \omega_B^2
 = \left.\frac{\partial \cV_L (\lam, \e)}{\partial \lam^i}
 \right|_{\lam=\lam_{\rm mes.}} = \left.\frac1{(2 \pi)^2}\,
 \frac{\partial N^X_{0, 0} (\lam, \e)}{\partial \lam^i}
 \right|_{\lam=\lam_{\rm mes.}}\ ,
\ee
where the $\e_i$ are subject to the mesonic constraint, \eqref{eq:mesonicconstraint}.
On the field theory side, in the large $N$ limit corresponding to the two derivative supergravity approximation which we took here, the R-charges are linear in $N$,
\be
	R [T_i] = \Delta_i\, N\ ,
\ee
where supersymmetry further imposes the condition
\be
	\sum_i \Delta_i = 2\ .
\ee
One can now identify the rescaled R-charges $\Delta_i$ of the M2-brane theory with the equivariant parameters $\e_i$, which we carry out explicitly in \cite{Cassia:2025jkr}.~\footnote{While we focused on an off-shell point of view here, the same setup allows also a description within on-shell 4d supergravity. In this picture the generalization to arbitrary values of $\e_i$ corresponds to switching on scalar profiles in 4d and leads to a more general solution that no longer preserves all $\AdS_4$ symmetries, see \cite{Freedman:2013oja}. This is also reflected in the fact that a generic value of the $\Delta_i$ deformations, corresponding to R-charge assignments of the dual field theory on round S$^3$, breaks the conformal symmetry.} The important point here is the remaining equation that relates the first derivative of the genus-zero constant maps to the M2-brane charge $N$,
\be
\label{eq:importantrelation}
 \frac1{(2 \pi)^2}\, \sum_i \left. \frac{\partial N^X_{0, 0} (\lam, \e)}{\partial \lam^i}
 \right|_{\lam=\lam_{\rm mes.}} = 2\, N\ .
\ee
This is precisely the Legendre transform depicted schematically in Figures~\ref{fig:new} and \ref{fig:1} relating the topological string picture to the effective supergravity and dual field theory. The extremization holds in the limit of two derivative supergravity and corresponding large $N$ field theory limit, corresponding to the first approximation of a full integral transform between the $\lam$ and $N$ ensembles. We devote our companion paper, \cite{Cassia:2025jkr}, to the extension of this relation to finite $N$ in field theory and higher derivatives in supergravity, based on the additional genus-one constant maps term we have defined here. This analysis is greatly facilitated by the host of important results on the exact M2-brane partition function on round and squashed S$^3$, see \cite{Kapustin:2009kz,Hama:2011ea,Imamura:2011wg,Fuji:2011km,Marino:2011eh,Nosaka:2015iiw,Hatsuda:2016uqa,Chester:2021gdw} and references therein.

\subsection{Wrapped branes}
\subsubsection*{11d background solution}
One can also consider wrapped brane solutions, where the two spatial dimensions of the M2-branes ($x^{1,2}$) are no longer flat but they wrap around a compact two-dimensional surface $\Sigma$. The key point is that wrapped M2-branes have a near-horizon geometry with an $\AdS_2$ factor, and the Sasaki--Einstein transverse space $L$ becomes fibered over the surface $\Sigma$, see \cite{Kim:2006qu,Gauntlett:2007ts},
\be
 \mathd s^2_{r\to 0} = R^2_{\AdS_2} \mathe^{-2\sigma/3}\,
 \left( \mathd s^2 (\AdS_2) + \mathd s^2 (M)  \right)\ ,
 \qquad F_4 = R^3_{\AdS_2}\, \omega_{\AdS_2} \wedge F_2\ ,
\ee
where $\omega_{\AdS_2}$ is the volume form of $\AdS_2$ with unit radius and $F_2$ is a closed form on $M$,
\be
\begin{array}{ccc}
 L & \to & M \\
 && \downarrow \\
 && \Sigma
\end{array}
\ee
i.e.\ the total space of the fibration of $L$ over $\Sigma$.
We leave unspecified the warp factor $\sigma$ in the metric, also a function of the internal space, as it does not affect the topology. Similarly to \eqref{eq:4dN} we impose the flux quantization condition (still only at two derivative level, $N_{{\rm M2}, i} = n_i\, N$, see App.~\ref{app:A}),
\be
	n_i\, N = \frac1{(2 \pi l_P)^6} \int_{S_i} \star F_4 \in \BZ\ ,
\ee
where $S_i \hookrightarrow M$ are submanifolds representatives of the generators of the free part of $H_7(M,\BZ)$. Moreover, the generators $S_i$ are related to the Sasakian divisors $T_i\hookrightarrow L$, introduced in \eqref{eq:Sasakian-divisor}, as they fit into the diagram
\be
\label{eq:def-Si}
\begin{array}{ccc}
 T_i & \to & S_i \\
 \downarrow && \downarrow \\
 L & \to & M \\
 && \downarrow \\
 && \Sigma
\end{array}
\ee
where the vertical arrows at the top represent fiberwise inclusion of $T_i$ in $L$ and the inclusion of $S_i$ in $M$, respectively.

From a lower-dimensional perspective, if we require $\Sigma$ to be toric, these backgrounds correspond to supersymmetric rotating black holes with spherical or spindle horizon topology, $\Sigma =\WPL$, as discussed in \cite{Cacciatori:2009iz, Halmagyi:2013sla, Hristov:2018spe, Hristov:2019mqp, Bobev:2020pjk, Ferrero:2020twa, Couzens:2021cpk, Ferrero:2021etw, Hristov:2023rel}. Similar to the discussion surrounding \eqref{eq:4dGN}, the lower-dimensional Newton constant arises from explicit integration over the internal space. In the context of black hole physics, the two-dimensional Newton constant also plays a key role, being related to the Bekenstein--Hawking entropy of the black holes, $S_\text{BH} = 1/(4 G_\mathrm{2d})$. Again, this M2-brane construction admits a dual holographic description via localization techniques, as explored in \cite{Benini:2015noa,Benini:2015eyy,Hosseini:2022vho,Bobev:2023lkx,Bobev:2024mqw, Inglese:2023wky, Colombo:2024mts} and references therein, which we revisit in \cite{Cassia:2025jkr}.

\subsubsection*{Off-shell geometry}

Let us again consider the off-shell generalization of the above construction, where we relax the conditions on the metric of $M$ that arise from the equations of motion, enforcing only supersymmetry. Conceptually, the relation between the supergravity and topological string descriptions is similar yet distinct from the unwrapped branes description. The supergravity description uses the Sasakian manifold $L$, now additionally fibered over $\WPL$. On the other hand, the topological string description makes use of the fibration of $X_\circ=C(L)$  over $\WPL$, denoted by $Y$,
\be
\begin{array}{ccc}
 X_\circ & \to & Y \\
 && \downarrow \\
 && \WPL
\end{array}
\ee
The resulting total space of the fibration, $Y$, also satisfies the CY condition. This can be seen from the choice of charges which describe the fibration, (see Section~\ref{sec:4})
\be
	\sum_i n_i = a + b\ .
 \ee
We stress that the space $Y$ can be still thought of as a particular partial resolution of the cone over $M$. This follows from the observation that the full mesonic twist on $Y$ results in a conical space, isomorphic to $C(M)$. Homotopically speaking, the two spaces fit into the push-out diagram
\be
\begin{array}{ccc}
 \WPL &\hookrightarrow& Y \\
 \downarrow && \downarrow \\
 \ast &\hookrightarrow& C(M)
\end{array}
\ee
where the horizontal arrow at the top is a section of the bundle $Y\to\WPL$ with image at each point in $\WPL$ being the tip of the cone $X_\circ$ in the fiber, and the horizontal arrow at the bottom is the inclusion of a point as the tip of the cone $C(M)$. The vertical arrow on the right can then be regarded as the blow-down of the only non-trivial two-cycle in $Y$ corresponding to the image of the global section previously mentioned. This discussion is in fact in accordance with the explicit example studied around \eqref{eq:mesconicconstraintonY}. 

The construction of $Y$ once again enables us to compute the equivariant volume and the corresponding constant maps contributions. As demonstrated in \cite{Martelli:2006yb}, and further developed using the results of \cite{Couzens:2018wnk,Gauntlett:2018dpc,Hosseini:2019use,Hosseini:2019ddy,Gauntlett:2019roi,Boido:2022mbe}, the supergravity action and the so-called entropy function are directly connected to the genus-zero constant maps term, providing a clear holographic interpretation of the dual index, as in \cite{Benini:2015noa,Inglese:2023wky}. Specifically, the analog of \eqref{eq:r-chargeassignmentsingeometry}-\eqref{eq:importantrelation} relating the geometric and the holographic perspectives is given by~\footnote{The constant of proportionality is the same in each relation.}
\be
\label{eq:importantrelationfibered1}
 \left.\frac{\partial N^{Y}_{0, 0} (\lam, \e)}{\partial \lam^i}
 \right|_{\lam^i=\lam^i_{\rm mes.}} \propto n_i\, N\ ,
\ee
and
\be
\label{eq:importantrelationfibered2}
 \left.\frac{\partial N^{Y}_{0, 0} (\lam, \e)}{\partial \lam^+}
 \right|_{\lam^i=\lam^i_{\rm mes.}} \propto -a\, N\ ,
 \qquad
 \left.\frac{\partial N^{Y}_{0, 0} (\lam, \e)}{\partial \lam^-}
 \right|_{\lam^i=\lam^i_{\rm mes.}} \propto -b\, N\ ,
\ee
with $\e_i$ and $n_i$ obeying the mesonic constraints on the fiber $X$, \eqref{eq:mesonicconstraint} and \eqref{eq:mesonicchargeconstraint}, respectively. The above equations correspond holographically to $\mathcal{I}$-extremization, \cite{Benini:2015eyy}. It would be interesting to understand the holographic role of the mesonic twist on $Y$ (on top of the one on $X$ imposed above), which has not been explored in \cite{Martelli:2023oqk}. We come back to this point in \cite{Cassia:2025jkr}.

In this section, we have provided a basic overview of M2-brane backgrounds and their conceptual connection to the equivariant topological string on toric Calabi--Yau cones. While additional subtleties and details are involved in establishing the precise relation to the M2-brane partition functions, as outlined, these are explored in greater depth in our companion paper, \cite{Cassia:2025jkr}. For the remainder of this work, we focus on the special and intriguing case of vanishing fluxes, which bypasses many of these complexities and, in a sense, completes the examples introduced in Section~\ref{sec:4}.

\section{Vanishing flux backgrounds}
\label{sec:6}
In the previous section, we discussed background solutions with coincident branes in 11d supergravity, exhibiting non-trivial flux. The special case of vanishing fluxes that we examine now involves branes that either entirely wrap cycles of the internal space or are localized at isolated points.~\footnote{We are very heuristic in the description of the brane construction here since there exist many alternative descriptions in different duality frames, see e.g.\ \cite{Maldacena:1997de} for further discussion. For our present purposes we only need to know the topology of the internal space from 11d point of view, the specific brane construction is immaterial.} From a low-energy perspective, these configurations result in asymptotically flat 4d spacetime, either empty or populated by black hole centers. In this case, the anomaly cancellation condition ensures no overall flux through the compact manifold \cite{Sethi:1996es,Witten:1996md}. The equations of motion require the internal space to be Ricci flat, with unfixed Kähler moduli. As a result, the effective Newton constant is independent of the internal space volume. 

We consider compactifications to 4d $\cN = 2$ supergravity, which in the present perspective arise from 11d supergravity on a circle times a compact CY three-fold. Since compact CYs cannot be toric, they are unfortunately not amenable to an equivariant upgrade such as the one we discussed in the present article. However, we can view the additional circle factor as the simplest example of a Sasakian manifold, in line with our earlier discussion on geometric modeling,
\be
 L = \mathrm{S}^1 \times \mathrm{CY}_3\ .
\ee 
We can now consider the cone over the circle, which is simply the complex plane (without any conical singularities),
\be
	X = C(\mathrm{S}^1) \times \mathrm{CY}_3 = \BC \times \mathrm{CY}_3\ .
\ee
Note however, that $X$ is not a cone in this case, but rather a product of a compact CY and a cone over S$^1$.
Therefore $X$ is a partially toric CY four-fold that has no conical singularity, such that we can apply the equivariant upgrade to the $\BC$ factor directly,~\footnote{Note that from the supersymmetry analysis there is no constraint on the equivariant parameter of the complex plane in this case, such that \eqref{eq:fluxsusyconstr} does not hold.} while keeping the compact three-fold unspecified and non-equivariant. Although it may not be immediately clear that the geometric modeling inherited from the flux case is useful—especially since one can use standard topological strings from a 10d perspective, ignoring the M-theory circle—we argue that the equivariant description offers important new insights. Notably, we will argue in Section~\ref{sec:6.2} that the OSV conjecture of \cite{Ooguri:2004zv} can be naturally reformulated via the equivariant topological string on $X$.

\subsection{Minkowski vacuum and the 4d prepotential}
\label{sec:6.1}
We start with the vacuum case, where the 11d spacetime is given by
\be
 M_{11} = {\rm Mink}_4 \times \mathrm{S}^1 \times \mathrm{CY}_3
 = {\rm Mink}_4 \times L\ ,
\ee
with a vanishing four-form $F_4 = 0$. One can perform a consistent reduction of the 11d action down to 4d dimensions, known as compactification. Although flat space does not lead to a well-defined on-shell action, the close connection between the effective 4d supergravity action and the topological string partition function has been widely noted and explored in numerous papers. For a recent discussion and further references, see \cite{Dedushenko:2014nya}.

We briefly highlight key aspects of 4d $\cN=2$ supergravity to clarify its connection to the equivariant topological strings framework. The 4d action, whether in superspace or component form, is governed by the \emph{prepotential}, a holomorphic function of the complex scalars $X^I$.~\footnote{In off-shell supergravity, the scalars $X^I$ are the lowest components of the vector multiplets and the vector superfield in superspace. On-shell, they overparametrize the physical scalars but are still used to define the theory via the prepotential.} Supersymmetry ensures that the prepotential fully determines the Lagrangian, including higher-derivative corrections. It must also be a homogeneous function of a specific degree, related to the number of derivatives.

Explicitly, at two derivative level, the direct dimensional reduction produces the following prepotential of degree $2$,
\be
\label{eq:2derprep}
	F^{2 \partial} (X) = \frac16\, C_{a b c}\, \frac{X^a X^b X^c}{X^0}\ ,
\ee
where $C_{a b c}$ are the triple intersection numbers of the CY three-fold, cf.\ \eqref{eq:intnum} - \eqref{eq:importantidentity}, with $a, b, c = 1, \dots, h^{1,1} (\mathrm{CY}_3)$, and we assume an implicit summation over repeated up-down indices. The additional scalar $X^0$ (we use the convention $I = \{0, a \}$) in the denominator comes precisely from the circle reduction from 5d to 4d, and gives us a strong hint about the relation to equivariant topological strings on $\BC \times \mathrm{CY}_3$.  First, we can evaluate the equivariant volume of $X$, which clearly has a product structure
\be
\label{eq:cequivolCtimesCY3}
	\VVV_{\BC \times \mathrm{CY}_3} (\e_0, t^a) = \VVV_\BC (\e_0)\, \VVV_{\mathrm{CY}_3} (t^a) = \frac{1}{6\, \e_0}\, C_{a b c}\, t^a t^b t^c\ ,
\ee
and we can consider the $\lam$-parametrization for the $\BC$ factor,
\be
\label{eq:equivolCtimesCY3}
	\BV_{\BC \times \mathrm{CY}_3} (\lam^0, \e_0, t^a) = \BV_\BC (\e_0)\, \VVV_{\mathrm{CY}_3} (t^a) = \frac{\mathe^{\e_0 \lam^0}}{6\, \e_0}\, C_{a b c}\, t^a t^b t^c\ ,
\ee
in agreement with \eqref{eq:CequivvolMZ}. We should stress that the second factor in the expression above is inherently non-equivariant, therefore we cannot replace the original K\"ahler moduli $t^a$ with some redundant moduli $\lam^i$ in the compact directions associated to the CY$_3$. Determining the triple intersection numbers of the four-fold $X$ then trivializes since, unlike the typical equivariant case, the expression is already of minimal degree $3$ due to the cubic form of the CY volume: 
\be
	N_{\mathfrak{g}, 0}^{\BC \times  \mathrm{CY}_3} (t^a, \e_0)= \frac1{\e_0}\,  N_{\mathfrak{g}, 0}^{ \mathrm{CY}_3} (t^a)\ , \qquad \forall \mathfrak{g}\ .
\ee
This simply means that there is no space for further derivatives with respect to $\lam^0$, such that it fully drops out of the equivariant constant maps term, e.g.\
\be
	N_{0, 0}^{ \BC \times \mathrm{CY}_3} (\e_0, t^a) = \frac16\,  C_{a b c}\, \frac{t^a t^b t^c}{\e_0}\ ,
\ee
and by comparing with $F^{2 \partial}$ in \eqref{eq:2derprep} we are led to relate $X^a$ to $t^a$ and $X^0$ to $\e_0$. It has long been understood that the complex scalars $X^a$ correspond to the complexified K\"ahler moduli on the $\mathrm{CY}_3$ manifold, but what becomes clear in this simple example is how the additional $U(1)$ isometry associated to the circle manifests itself geometrically as the equivariant parameter of the complex plane, $\e_0$. Note that the alternative 10d topological string description in the absence of the additional circle requires relating $X^0$ to the effective string coupling constant, $g_s$. This is in agreement with the idea that our 11d approach geometrizes $g_s$ and it ceases to be a parameter independent from the geometry. In other words, the additional equivariant parameter $\e_0$ from the additional circle in the 11d picture takes the place of $g_s$ in the 10d picture.~\footnote{We still allow from the 11d perspective that $g_s$ is a non-trivial numerical constant, which would take care of the overall normalization, but it has no relation to the geometry of the manifold.}

We move on to the second term in the derivative expansion of the prepotential, which supersymmetry constrains to be of degree $0$ in the scalars $X^I$ and linear in an additional composite scalar $\hat{A}$,~\footnote{For more details on the superconformal formalism that suitably accounts for higher-derivative corrections, see e.g.\ \cite{Lauria:2020rhc}. The scalar $\hat{A}$ in this case determines the coupling to the so-called Weyl$^2$ four-derivative invariant, \cite{Bergshoeff:1980is}.}
\be
\label{eq:4derprep}
	F^{4 \partial} (X; \hat{A}) = \frac1{24}\, \frac{b_a X^a}{X^0}\, \hat{A}\ ,
\ee
where $b_a$ are the second Chern-numbers of the corresponding four-cycles on the $\mathrm{CY}_3$ space, cf.\ \eqref{eq:2ndchernnumbers}.  It is again straightforward to determine the relation with the genus-one constant maps in this case,
\be
	N_{1, 0}^{ \BC \times \mathrm{CY}_3} (\e_0, t^a)  =   \frac1{24}\, \frac{b_a t^a}{\e_0}\ ,
\ee
in agreement with the previously stated relations, along with $\hat{A} = 1$. This is another consistency check that we have correctly identified the manifold $X$ subject to the equivariant upgrade. 

The remaining higher-derivative terms (in principle, up to an infinite number of derivatives) are consistently described in off-shell supergravity by decreasing the homogeneity degree of the prepotential in the variables $X^I$ and increasing it in the variable $\hat{A}$. However, in the direct dimensional reduction from 11d in the vanishing flux case they all vanish,
\be
	F^{(2 n) \partial} (X; \hat{A}) = 0\ , \qquad \forall n > 2\ .
\ee
On the other hand, the higher-genus constant maps terms are generically non-zero, 
\be
 N_{\mathfrak{g}>1,0}^{\BC\times \mathrm{CY}_3}
 = \frac{|B_{2\mathfrak{g}}|}{4 \mathfrak{g}}
 \frac{|B_{2\mathfrak{g}-2}|}{(2\mathfrak{g}-2)}
 \frac{(-1)^\mathfrak{g}}{(2\mathfrak{g}-2)!}\, \frac{\chi (\mathrm{CY}_3)}{\e_0}\ .
\ee
This might suggest a breakdown in the relation between the prepotential and constant map terms at higher genera. Yet, as argued in \cite{Dedushenko:2014nya}, a Schwinger-like calculation recovers the missing terms and links them precisely to higher genus constant maps. These higher terms lack explicit dependence on the Kähler moduli $t^a$ and their equivariant completion $\lam^0$, making it difficult to test this relation here. A more careful revisit of the calculation in \cite{Dedushenko:2014nya} and its connection to the equivariant upgrade of higher genus constant maps would be an interesting direction for future work.

We note a key property of the equivariant constant maps terms in this case—they are automatically independent of the additional  parameter $\lam^0$. This is a crucial self-consistency condition, as the vanishing fluxes imply that the topological string is not in a different ensemble from the supergravity description. This represents a degenerate and simplified case of the general rule outlined in Figure~\ref{fig:1}, which retrospectively supports the earlier understanding of the relationship between constant map terms and the prepotential.

\subsection{Black holes and the OSV conjecture}
\label{sec:6.2}
Let us now consider the black hole case, following the main steps of \cite{Ooguri:2004zv}, which we aim to understand through the equivariant topological string. Since this approach offers no specific insights into the multicentered case,~\footnote{Multiple asymptotically BPS black holes can exist in equilibrium at arbitrary distances if their charges are parallel, or as bound states at fixed distances in certain cases, see \cite{Denef:2000nb}.} we will focus on a single black hole center for simplicity. From a 4d perspective, the spacetime of a single center black hole interpolates between the asymptotic Minkowski solution discussed earlier and the near-horizon geometry of $\AdS_2 \times \mathrm{S}^2$, with equal radii for both spaces.~\footnote{In the vanishing flux case, respectively in ungauged supergravity, only the smooth sphere is an admissible horizon topology. We thus exclude from the start the more general spindle geometries with conical singularities.} We can then zoom in on the horizon, where the on-shell supergravity action is holographically dual to the boundary CFT index.

There are two possible ways to embed this near-horizon geometry in 11d, arising from the fact that the original internal space $L$ can be fibered over $\BP^1$ in different ways. In the general case, where all fibration charges are non-vanishing, we obtain the 11d background,
\be
 M_{11} = \AdS_2 \times M\ , \qquad \begin{array}{ccc}
 (\mathrm{S}^1 \times \mathrm{CY}_3) & \to & M \\
 && \downarrow \\
 && \BP^1
\end{array}
\ee
In this case, the holographically dual theory is one-dimensional and remains somewhat elusive. However, when the charge $p^0$, representing the first Chern class of the S$^1$ fibration, vanishes, the situation changes qualitatively. The necessary electric charge in this direction causes the circle to be fibered over the $\AdS_2$ factor, resulting in a distinct 11d background,
\be
 M_{11} = \mathrm{BTZ} \times M'\ , \qquad \begin{array}{ccc}
 \mathrm{CY}_3 & \to & M' \\
 && \downarrow \\
 && \BP^1
\end{array}
\ee
such that the dual field theory is two-dimensional. In turn, this case has been extensively studied and understood in great detail, see \cite{Maldacena:1997de} and consequent references.

Although this dichotomy seems crucial from a holographic perspective, our geometric construction treats both cases uniformly by modeling them with equivariant topological strings on the following manifold:
\be
\begin{array}{ccc}
 (\BC \times \mathrm{CY}_3) & \to & Y \\
 && \downarrow \\
 && \BP^1
\end{array}
\ee
with the special case $p^0 = 0$ a smooth limit of the general analysis for arbitrary $p^0$ (that corresponds to the charge in the $\BC$ direction).~\footnote{It is worth noting that the dichotomy is also significant in gravity and geometry, as the BTZ case allows for multiple saddles (see \cite{deBoer:2006vg}). However, this is relevant only in the non-perturbative regime and lies beyond the scope of this discussion.} 

An important remark is in order: the way the charges $p^a$ appear in the case of a compact CY is different from the way the charges $p_\alpha$ (or $n_i$) appear in the case of non-compact fiber. In the non-compact case, the $p_\alpha$ correspond to Chern classes of circle bundles over $\BP^1$ associated to the fibration of the toric CY. In this case, these bundles are in one to one correspondence with the toric isometries of the fiber and they appear as shifts of the corresponding equivariant parameters as in \eqref{eq:MZequivolmesonicoverspindle}. In the compact CY case, on the other hand, the fibration is trivial because there are no continuous isometries on the fiber, and the $p^a$ arise as the Chern classes of line bundles on $\BP^1$ associated to the reduction of the M-theory three-form along the compact two-cycles in the CY. Equivalently, we can think of the $p^a$ as the integrals of the four-form $F_4=\mathd C_3$ over the product of $\BP^1$ with some non-trivial two-cycle $\gamma^a$ in the compact fiber. For this reason, the parameters $p^a$ are in one to one correspondence with the K\"ahler moduli $t^a$.

The black hole entropy, derived from the theory governed by the two-derivative and four-derivative prepotentials \eqref{eq:2derprep} and \eqref{eq:4derprep}, was directly calculated in \cite{LopesCardoso:1998tkj, LopesCardoso:1999cv, LopesCardoso:1999fsj}. It was later reformulated as a free energy in the so called mixed ensemble of fixed magnetic charges $p^I = \{p^0, p^a \}$ and varying electric charges, more suitable for our current discussion, in \cite{Ooguri:2004zv},
\be
\label{eq:OSVlikeformula}
 I_\text{BH} (\chi, p) = \frac1{2 \pi\mathi}\,
 \left(F (X^I = \chi^I - \pi\mathi\, p^I; \hat{A} = 1)
 - F (X^I = \chi^I + \pi\mathi\, p^I; \hat{A}=1) \right)\ ,
\ee
where we have set the effective Newton constant to unity, $G_\mathrm{4d} = 1$, $\chi^I$ are the conjugate chemical potentials to the black hole electric charges, and
\be
	F(X; \hat{A}) := F^{2 \partial} (X) + F^{4 \partial} (X; \hat{A})\ ,
\ee
as given in \eqref{eq:2derprep} and \eqref{eq:4derprep}.~\footnote{In order to simplify as much as possible the relationship between the supergravity expression and the topological string predictions below, which is the focus of present discussion, we have changed the normalizations for the scalar $\hat{A}$ and $F^{4\partial}$ in comparison with \cite{Hristov:2021qsw}.}

The OSV conjecture proposes that, after recognizing the connection between the constant map terms and the prepotential, the black hole partition function is given by two instances of the topological string partition function on CY$_3$, with the K\"ahler parameters $t^a$  and the string coupling $g_s$  appropriately identified \cite{Ooguri:2004zv}, in the present conventions as~\footnote{In this subsection, we denote the topological string partition function as $F^\text{top}$ to distinguish it from the prepotential in 4d supergravity. The OSV formula here is presented with differing factors of $\mathi$ and $\pi$ compared to the original reference as we follow a different topological string normalization.}
\be
\label{eq:originalOSV}
\begin{aligned}
 \log Z_\text{BH} (\chi, p) = \frac1{2\pi\mathi}\, \Big(& F^\text{top}_{\mathrm{CY}_3}
 (g_s = \frac1{\chi^0 - \pi\mathi p^0}; t^a = \frac{\chi^a -\pi\mathi p^a}
 {\chi^0 -\pi\mathi p^0}) \\ 
 &- F^\text{top}_{\mathrm{CY}_3} (g_s = \frac1{\chi^0 + \pi\mathi p^0};
 t^a = \frac{\chi^a + \pi\mathi p^a}{\chi^0 + \pi\mathi p^0})\Big)\ ,
\end{aligned}
\ee
with $F^\text{top}$ the non-equivariant topological string free energy as given in \eqref{eq:2.2}. This relation, in its full form for the complete topological string partition function, can be viewed as a proposal for a non-perturbative definition of the black hole partition function. At leading perturbative order, however, it is an observation that the explicit supergravity calculation, \eqref{eq:OSVlikeformula}, can be expressed in terms of the constant maps terms of the topological string on the compact CY$_3$. The conjecture is partially motivated by the relations discussed in the previous subsection regarding the asymptotic Minkowski vacuum, but remains unproven rigorously. Notably, the formulation above suggests a 10d perspective, since the string coupling is a meaningful parameter.

Our equivariant approach proves useful here, as it explains why a formula like \eqref{eq:OSVlikeformula} should hold for the topological string from a 11d perspective. The following reformulation of the OSV conjecture highlights the unifying role of equivariant topological strings in connecting black hole entropy to geometric invariants.

Based on the number of different examples with $\BP^1$ fibrations we investigated in Section~\ref{sec:4}, as well as the expected match with the black hole on-shell action at two and four derivatives, we find that the equivariant volume of the total space of the fibration $Y$ is given by
\be
\label{eq:OSVequivol}
	\BV_Y (\lam, \e, t) = \frac{1}{\NN}\, (\BV_{\BC \times \mathrm{CY}_3} (\lam^0_+, \e_0^+, t^a_+) -  \BV_{\BC \times \mathrm{CY}_3} (\lam^0_-, \e_0^-, t^a_-) )\ ,
\ee
where we denoted the equivariant parameter on the two-sphere as $\NN =a \e_- - b \e_+$ (with $a = b = 1$ on the sphere), as we reserve the label $0$ for the parameters coming from the $\BC$ direction in the fiber as in the previous subsection. The equivariant volume of $\BC \times \mathrm{CY}_3$ was already evaluated there, \eqref{eq:equivolCtimesCY3}, and just as in Section~\ref{sec:4} we denoted appropriate shift of the relevant parameters on the two poles with the $\pm$ subscript. On the complex plane we explicitly derived the shifts in $\lam^0$ and $\e_0$ in Subsection~\ref{subsec:4.2}:
\be
\begin{array}{c}
 \e_0^- := \e_0 + \frac{p^0}{a}\, \e_-\ ,
 \qquad \e_0^+ := \e_0 + \frac{p^0}{b}\, \e_+\ , \\ \\
 \lam^0_- := \lam^0 - \frac{\NN}{b\, \e_0^-}\, \lam^+\ ,
 \qquad \lam^0_+ := \lam^0 + \frac{\NN}{a\, \e_0^+}\, \lam^-\ . 
\end{array}
\ee
Unfortunately, due to the fact that the compact CY three-fold is non-toric, we could not explicitly evaluate the possible shift and the respective definition of the K\"ahler parameters $t^a_\pm$, which we for now keep as unknowns. Based on the relation with the supergravity answer we are going to be able to determine these shifts below, but they remain conjectural from the point of view of the topological string.

Before evaluating the constant maps terms from the equivariant volume, let us first discuss the constraints we infer from supersymmetry and not directly from geometry. In this case supersymmetry is preserved without a topological twist, see \cite[Section~7]{BenettiGenolini:2024kyy}, which implies
\be
\label{eq:identificationabosv}
	a = -b = 1\ .
\ee
In addition, we find that the spherical equivariant parameter is actually fixed to a constant, as derived from supergravity, \cite{Hristov:2021qsw,Hristov:2022pmo},
\be
\label{eq:identificationnuosv}
	\e_+ = \e_- = \pi\mathi\ , \qquad \NN = 2 \pi\mathi\ .
\ee

Given these identifications, we can now evaluate the first two constant maps terms as
\be
 N_{0, 0}^Y (\e_{0}, t^a) = \frac{C_{a b c}}{12 \pi\mathi}\,
 \left( \frac{t_+^a t_+^b t_+^c}{\e_0 - \pi\mathi p^0}
 - \frac{t_-^a t_-^b t_-^c}{\e_0 + \pi\mathi p^0}\right)\ ,
\ee
and
\be
 N_{1, 0}^Y (\e_{0}, t^a) =  \frac{b_a}{48 \pi\mathi}\,
 \left( \frac{t^a_+}{\e_0 - \pi\mathi p^0} - \frac{t^a_-}{\e_0 + \pi\mathi p^0} \right)\ ,
\ee
which follows from \eqref{eq:spindlefibermestwistconstantmaps} noting that the additional term in $N_{1, 0}^Y$ of lower order in derivatives w.r.t.\ $t^a$ vanishes automatically.

We now observe that the following conjectural identification of the shifts,~\footnote{In the special case $p^a = 0, \forall a$, which is allowed as long as $p^0 \neq 0$, the identification $t^a_- = t^a_+ = t^a$ follows directly and is no longer conjectural.}
\be
	t^a_- = t^a + \pi\mathi\, p^a\ , \qquad t^a_+ = t^a - \pi\mathi\, p^a\ ,
\ee
is in agreement with \eqref{eq:originalOSV}, and we find the relation
\be
\label{eq:OSV-new}
	I_\text{BH} (\chi, p) = N_{0, 0}^Y (\e_{0} = \chi^0, t^a = \chi^a)  + N_{1, 0}^Y (\e_{0} = \chi^0, t^a = \chi^a) \ .
\ee
The above conjecture for the shifts $t^a_\pm$ can be viewed as a reformulation of the OSV conjecture. It seems plausible that the shifts $t^a_\pm$ here come from the complexification of the real K\"ahler parameters by the reduction of the 11d three-form field, $C_3$. On the other hand, we have rigorously derived the shift along the toric direction $\e_0^\pm$, finding perfect agreement with higher-derivative supergravity.

In addition to \eqref{eq:OSV-new}, we now propose a modified version of the OSV conjecture in its complete non-perturbative form, which arises naturally within our 11d perspective:
\be
\label{eq:finalformnewOSV}
	\log Z_\text{BH} (\chi, p) = F^\text{top}_Y (g_s = 1; \e_{0} = \chi^0, t^a = \chi^a)\ ,
\ee
as derived from \eqref{eq:OSVequivol} with the identifications \eqref{eq:identificationabosv}-\eqref{eq:identificationnuosv}. This formula remains in the spirit of the OSV conjecture and aligns perfectly with our conceptual expectations from Figure~\ref{fig:1}, due to the vanishing fluxes and consequent lack of Laplace transform. 

It is important to note that \eqref{eq:originalOSV} and \eqref{eq:finalformnewOSV} are not equivalent, aside from their agreement at the level of the genus-zero and genus-one constant maps and respective match with two- and four-derivative supergravity. Our proposal suggests a different set of subleading perturbative and non-perturbative corrections, which could potentially be tested against the dual supersymmetric index. In particular, the higher-genus constant maps predicted by our formula can be evaluated to 
\be
\begin{aligned}
	N_{\mathfrak{g}>1, 0}^Y (\e_0) &=   \frac{|B_{2\mathfrak{g}}|}{4 \mathfrak{g}} \frac{|B_{2\mathfrak{g}-2}|}{(2\mathfrak{g}-2)}
 \frac{(-1)^\mathfrak{g}}{(2\mathfrak{g}-2)!}\, \frac{\chi (\mathrm{CY}_3)}{\NN}\, \left( \frac{1}{\e_0^+} -  \frac{1}{\e_0^-} \right) \\
& = \frac{|B_{2\mathfrak{g}}|}{4 \mathfrak{g}} \frac{|B_{2\mathfrak{g}-2}|}{(2\mathfrak{g}-2)}
 \frac{(-1)^\mathfrak{g}}{(2\mathfrak{g}-2)!}\, \frac{p^0\, \chi (\mathrm{CY}_3)}{(\e_0)^2 + \pi^2 (p^0)^2}\ ,
\end{aligned}
\ee
which in the general case $p^0 \neq 0$ differs from the respective terms within the original OSV proposal, \eqref{eq:originalOSV}. We hope to be able to check these modifications against holographic calculations, as well as explore the addition of a refinement parameter, in future works.

\section{Discussion and outlook}
\label{sec:7}

We finish with a discussion that relates directly to the topological string partition function discussed here. We give a list of topics more closely related to supergravity and holography in \cite{Cassia:2025jkr}.

\subsubsection*{Refinement}
In the discussion so far we have neglected the possibility for another expansion parameter in the topological string partition function, which is allowed on toric manifolds. The additional parameter, conventionally denoted by $\mathfrak{b}$ is known as refinement.
This parameter arises naturally in the computation of the Nekrasov partition function \cite{Nekrasov:2002qd}, which generalizes the topological string partition function for toric Calabi--Yau manifolds. This partition function incorporates two independent equivariant parameters—one for each copy of the complex plane that describes the 4d/5d Omega background of Nekrasov--Okounkov \cite{Nekrasov:2003rj}. 

The general refined equivariant string partition function has the following formal expansion:~\footnote{The refinement parameter $\mathfrak{b}$ used here corresponds to the one in \cite{Alexandrov:2023wdj}, after a rescaling by $\mathi$.}
\be
    F_X (\lam, \e; g_s, \mathfrak{b}) = \sum_{\mathfrak{g} = 0}^\infty \sum_{n = 0}^\mathfrak{g} (-1)^n g_s^{2 (\mathfrak{g}- 1)}\, (\mathfrak{b} + \mathfrak{b}^{-1})^{2 n}  F_{\mathfrak{g}, n} (\lambda, \epsilon)\ ,
\ee
where in some cases the refined Gromov--Witten invariants can be explicitly calculated. 

An apparent caveat of the refinement is the lack of a clear geometric origin for the additional parameter at the level of the CY manifold $X$. This particularly affects the constant maps terms, whose refinement lacks a general definition.
We hope to come back to this point in the future, \cite{Cassia:2025jkr}.

\subsubsection*{Intersections numbers at different orders and other brane systems}

We already noted that the equivariant upgrade we pursued here can be defined for toric manifolds of arbitrary (complex) dimension. The equivariant constant maps we defined, see \eqref{eq:cconstmapssummary}, are a natural extension of the non-equivariant topological string expressions in $3$ dimensions. This is reflected in the fact that all constant maps terms we have correspond to triple (equivariant) intersections numbers, such that we have three natural classical candidates for an equivariant upgrade, cf.\ \eqref{eq:2.2}
\be
 {\rm deg} = 3 \quad \Rightarrow \quad \int_X \omega^3\ ,
 \quad \int_X\omega\cup c_2(TX)\ , \quad \int_X c_3(TX)\ ,
\ee
where the power of $\omega$ gives the homogeneity degree in the moduli $\lam^i$.
Each of these three candidates correspond to one of the terms in \eqref{eq:cconstmapssummary}. This is a choice that appears specifically relevant in the geometric description of M-theory, but this is no longer true for other brane models, see in particular \cite[(1.7)]{Colombo:2023fhu}.

Based on this, we can speculate that more general M2- and M5-brane systems are still geometrically described by the $({\rm deg} = 3)$ equivariant topological string we defined here, with the field theory partition function given schematically by
\be
 Z_L^\text{M-th.} = \int \mathd \lam\, \exp \left(
 F^{({\rm deg} = 3)}_X (\lam, \e; g_s) - \lam\, N_{\rm M2}
 - \lam^2\, N_{\rm M5} \right)\ , 
\ee
where $N_{\rm M2}$ and $N_{\rm M5}$ correspond to the exact brane charge of the corresponding type. We should however note that the relation between the transverse space $L$ and the modeled $X$ might generically become more complicated, see some examples in \cite{Martelli:2023oqk}.

Similarly, one could imagine a different equivariant topological string definition, relating to double intersections numbers, where the constant maps terms will be related to the equivariant upgrade of
\be
 \mathrm{deg} = 2 \quad \Rightarrow \quad \int_X \omega^2\ , \quad \int_X c_2(TX)\ ,
\ee
which would be used in the description of D3-branes and other type IIB objects.
Based on the scaling argument in \cite[(1.7)]{Colombo:2023fhu}, one would have to consider an integral of the form
\be\label{eq:IIB}
 Z_L^\mathrm{IIB} = \int \mathd \lam\, \exp \left( F^{({\rm deg} = 2)}_X
 (\lam, \e; g_s) - \lam\, N_\mathrm{D3} - \dots \right)\ , 
\ee
where the leading term in $F^{(2)}_X$ is proportional to $\BV^{(2)}_X (\lam, \e)$, cf.\ \eqref{eq:g0equimapn}, instead of $\BV^{(3)}_X (\lam, \e)$, thus leading to a Gaussian integral in the variable $\lam$. It then follows that the saddle-point equations for this integral match with the flux constraints considered in \cite[(1.4)]{Colombo:2023fhu}.
At present, this should be viewed merely as a reformulation of the flux constraints in terms of a saddle-point equation for an integral over $\lam$. It would nevertheless be interesting to understand whether there exists an \emph{a priori} justification for such an integral representation, as in~\eqref{eq:IIB}.

Likewise, we can consider an arbitrary degree $m$ in $\lam$,
\be
 {\rm deg} = m \quad \Rightarrow \quad \int_X \omega^m\ ,
 \quad \int_X\omega^{m-2} \cup c_2(TX)\ , \quad  \dots\ ,
 \quad  \int_X  c_m (TX)\ ,
\ee
with each higher degree allowing a larger number of different invariants made up from $m$ derivatives acting on the equivariant volume $\BV_X$ (we skipped $c_1 (TX)$ as it vanishes on CY manifolds non-equivariantly), and first term proportional to $\BV_X^{(m)}$.
Such an example of higher degree $({\rm deg} = 5)$ is a candidate for the description of massive IIA objects, see again \cite[(1.7)]{Colombo:2023fhu}, 
\be
 Z_L^\mathrm{mIIA} = \int \mathd \lam\, \exp \left( F^{({\rm deg} = 5)}_X
 (\lam, \e; g_s) - \lam^2\, N_{\rm D2} - \lam^3\, N_{\rm D4} - \dots  \right)\ . 
\ee
Note furthermore that each of the terms of a given degree might allow for additional refinement and come in a different linear combination, see the discussion of periods in \cite{Cassia:2022lfj}. We give additional holographic comments on these observations in \cite{Cassia:2025jkr}.

The ultimate goal of the present approach is to determine the \emph{fully} equivariant upgrade of the complete topological string partition function in the sense specified in Section~\ref{sec:2}, i.e.\ including dependence on all ineffective equivariant parameters $\e$ and K\"ahler moduli $\lam$. This would have to take into account not only all higher degree maps%
\footnote{The equivariant volume admits multiple inequivalent quantizations. One approach is via the 2d GLSM disk (or sphere) partition function \cite{Hori:2013ika} (or \cite{Jockers:2012dk,Gomis:2012wy,Closset:2015rna}), where the quantization parameter controls the weights of instanton contributions, i.e., higher-degree maps in GW theory. Another approach is as a quantum mechanical index, where the parameter instead sets the compactified time circle size. This corresponds to the K-theoretic uplift of the equivariant volume, also known as the Hilbert series \cite{Benvenuti:2006qr}, and is fundamentally different from the first. In general, no direct relation is expected between these quantizations. However, turning on both parameters simultaneously leads to a 3d (K-theoretic) GLSM partition function/index \cite{Jockers:2018sfl,Cassia:2022lfj}. While connections to the higher-genus GW or topological string partition functions exist in special cases, a general understanding remains elusive and likely requires additional dualities, such as holography.}, but also all higher genus contributions,
thus leading to a complete set of Gromov--Witten invariants, $N_{\mathfrak{g},\beta} (\lam, \e)$, cf.\ \eqref{eq:1.4}, for which, a rigorous definition is still lacking.
Clarifying this will allow us to directly relate the present approach with a number of already established results on the connection between the Gromov--Witten invariants and the M2-brane instantons in a series of papers, see \cite{Drukker:2011zy,Hatsuda:2013oxa,Grassi:2014zfa} and references thereof.

\subsection*{Acknowledgements}
\noindent We wish to thank Thorsten Hertl for illuminating discussions, and Dario Martelli and Alberto Zaffaroni for useful comments on the draft. We gratefully acknowledge support during the MATRIX Program “New Deformations of Quantum Field and Gravity Theories,” (Creswick AU, 22 Jan -- 2 Feb 2024), where this project was initiated.
The work of L.C.\ was supported by the ARC Discovery Grant DP210103081.
The study of K.H.\ is financed by the European Union-NextGenerationEU, through the National Recovery and Resilience Plan of the Republic of Bulgaria, project No BG-RRP-2.004-0008-C01.

\appendix

\section{Exact M2-brane charge}
\label{app:A}

Here we address the question of the exact M2-brane charge, denoted by $N_{\rm M2}$ in the introduction, first focusing on the spacetime filling case. In the classical two-derivative supergravity background, see section \ref{sec:5.1}, $N_{\rm M2}$ matches the number of branes $N$. Eight-derivative corrections to the M2-brane solution were previously discussed in \cite{Bergman:2009zh}, where two correction terms were proposed (in the absence of fractional branes). However, we find that only one of these terms should be included, while the second should be discarded.

To present our main argument, we briefly summarize the key elements of the calculation in \cite{Bergman:2009zh}, which is based on the Chern-Simons-like higher-derivative correction to the 11d supergravity action, derived in \cite{Vafa:1995fj,Duff:1995wd}.~\footnote{It is important to note that the fully supersymmetric invariant Lagrangian that includes this higher-derivative correction is still unknown. Hence, we only rely on the available action to determine the M2-brane charge, which is protected from other contributions, but cannot use it to calculate the corrections to the supergravity action.} Including this correction, the 11d supergravity bosonic action for the metric $g_{\mu\nu}$ and the three-form field $C_3$ with field strength $F_4 = {\rm d} C_3$ is given by 
\be
 S_\mathrm{11d} = \frac{2 \pi}{(2 \pi l_P)^9} \int \left[ R \star 1 -\frac12\, F_4 \wedge \star F_4 - \frac16\, C_3 \wedge F_4 \wedge F_4 + (2 \pi l_P)^6\, C_3 \wedge I_8 \right]\ , 
\ee 
where $I_8$ is the 8-dimensional anomaly polynomial, which can be expressed as a polynomial in the Pontryagin classes of the tangent bundle to the 11d spacetime,
\be
 I_8 = - \frac1{2 \cdot 4!}\, \left( \frac14\, p_1^2 - p_2 \right)\ .
\ee
The M2-brane charge in a given background can be derived from the equation of motion for $C_3$, 
\be
\label{eq:EOMC3}
{\rm d} \star F_4 + \frac12\, F_4 \wedge F_4 = (2 \pi l_P)^6\, N\, \delta^8 (x)  + (2 \pi l_P)^6\, I_8\ , 
\ee 
where the first term on the r.h.s.\ represents a localized M2-brane source with flux $N$. This flux appears naturally in the explicit supergravity background solution discussed in Sec.\ \ref{sec:5}, where the eight-dimensional space transverse to the M2-brane is identified with the CY cone $X_\circ$ over the Sasakian manifold $L$. 

At large radius, the 11d geometry looks like that of $\mathbb{R}^3\times X_\circ$, so that the tangent bundle decomposes as the direct sum of a trivial rank-3 vector bundle and the tangent of $X_\circ$ (pulled-back to 11d).
Since $X_\circ$ is a complex manifold with vanishing first Chern class (i.e.\ CY), we can further simplify the anomaly polynomial as
\be 
 I_8 = - \frac{c_4(TX_\circ)}{24}\ ,
\ee
where $c_4$ is the fourth Chern class of the tangent to $X_\circ$, and the flat directions do not give any contributions.

The key point, which holds universally for all M2-brane solutions, is that the additional terms in the $C_3$ equation of motion generate corrections to the brane charge: 
\be
\label{eq:QM2}
\begin{aligned}
 N_{\rm M2}
 :=&\, \frac1{(2 \pi l_P)^6}\, \int_L \star F_4 \\
  =&\, \frac1{(2 \pi l_P)^6}\, \int_{X_\circ} \mathd \star F_4 \\
  =&\, N - \frac{\chi(X)}{24} - \frac1{2 (2 \pi l_P)^6}\, \int_X F_4 \wedge F_4\ ,
\end{aligned}
\ee
where we used the equation of motion \eqref{eq:EOMC3} together with Stokes' theorem and identified $L$ with the boundary of $X_\circ$ at infinity in the radial direction.
In the last line, we also used that the top Chern class of a smooth manifold integrates to the Euler characteristic. However, since $X_\circ$ is not smooth, we define the integral of the top Chern class by first resolving the conical singularity and then taking the mesonic limit $t\to0$. We denote the resulting integral by $\chi(X)$ for short, and it is independent of the chosen resolution provided that the resolution contains no orbifold points. 

The last term here accounts for a contribution from discrete torsion, arising from the discrete holonomy of the $C_3$ field. This term plays an important role in theories like ABJ \cite{Aharony:2008gk}, where the ranks of the gauge groups in the quiver field theory do not match, see below.

Setting aside the discrete torsion term for the moment, we wish to highlight an important distinction in our interpretation of the correction proportional to $\chi (X)$, compared to that of \cite{Bergman:2009zh}. In \cite{Bergman:2009zh}, the authors argue that this topological term should be separated into two independent, non-vanishing contributions—a bulk and a boundary term. We find this distinction unnecessary, as the Euler characteristic is a topological invariant only when considered in its entirety, whereas the bulk and boundary terms are not separately metric independent. Since the M2-brane charge itself must be an invariant quantity, it is logical to include the full Euler characteristic on the r.h.s. of \eqref{eq:QM2}, as we have done.

This represents our only disagreement with the analysis in \cite{Bergman:2009zh}, and we will use the M2-brane charge as presented here, in the absence of discrete torsion,
\be 
N_{\rm M2} = N - \frac{\chi (X)}{24}\ . 
\ee 

The contribution from discrete torsion—previously discussed in \cite{Bergman:2009zh} in the context of the $\mathbb{C}^4/\mathbb{Z}_k$ orbifold—is instead associated with fractional brane charges \cite{Aharony:2008gk}. In such cases, the gauge group ranks at the various nodes of the quiver theory become unequal, with the imbalance captured precisely by the discrete torsion term in \eqref{eq:QM2}. For the $\mathbb{C}^4/\mathbb{Z}_k$ example, our analysis is fully consistent with the results of \cite{Bergman:2009zh}. A systematic generalization of the discrete torsion contribution to an arbitrary toric manifold $X$ remains an interesting open direction for future work.

\subsection*{Wrapped brane charges}
In the case of wrapped M2-branes on a two-dimensional surface $\Sigma$, at two-derivative level, we have
\be
\label{eq:M2chargeSi}
	n_i\, N = \frac1{(2 \pi l_P)^6} \int_{S_i} \star F_4 \in \BZ\ ,
\ee
where $S_i \hookrightarrow M$ are defined in \eqref{eq:def-Si}.
In the case that $\Sigma$ is a toric Riemann surface, there are two additional 7-dimensional submanifolds of $M$ that can support M2-brane charge; they correspond to the fibers of $M$ over the two fixed-points in the base. We denote these as $S_\pm$ and , by definition, they are isomorphic to two copies of the fiber $L$. Then we have two equations analogous to \eqref{eq:M2chargeSi},
\be
 -a N = \frac1{(2 \pi l_P)^6} \int_{S_+} \star F_4\ ,
 \hspace{30pt}
 -b N = \frac1{(2 \pi l_P)^6} \int_{S_-} \star F_4\ ,
\ee
where we identified the base of the fibration with $\WPL$. Due to the CY condition on the total space $Y$, we have that the sum of all toric divisors should vanish identically. i.e.\ $\sum_i S_i+S_++S_-=0$, and in fact one can check that, because of \eqref{eq:CYforY}, the sum of the integrals of $\star F_4$ over all divisors vanishes.

Following the previous discussion, we find a shift in the exact charges,  
\be
    N_{{\rm M2}, i} := n_i\, N + \delta_i\ , \quad N_{{\rm M2}, +} := -a\, N + \delta_+\ , \quad N_{{\rm M2}, -} := -b\, N + \delta_-\ ,
\ee
given by
\be
    \delta_i := \int_{S_i} \star f_4\ , \quad \delta_\pm := \int_{S_\pm} \star f_4\ , \qquad
    \mathd \star f_4 = I_8\ ,
\ee
again looking at the torsionless case for simplicity. 

Let us assume there exist submanifolds $y_i$ of the 11d background geometry such that $\partial y_i=S_i$, then we can write $\delta_i = \int_{y_i}I_8$. A thorough analysis of the geometry of the $y_i$'s is beyond the scope of this article, however it natural to expect that $y_i$ take the form of the total spaces of fiber bundles over $\Sigma$ with fibers isomorphic to the toric divisors $D_i$ in $X_\circ$. Whenever the space $Y$ is itself a toric CY (and $\Sigma$ is a weighted projective space $\WPL$), then it follows that $y_i$ are toric divisors in $Y$, corresponding to the vanishing of one of the homogeneous coordinates in the fiber. In addition to these divisors, $Y$ also has two divisors associated to the vanishing of the homogeneous coordinates on the base $\WPL$, and they correspond to the fibers over the two fixed points in $\WPL$; we denote these divisors as $y_\pm$ which correspond to the cones over $S_\pm$, respectively. In this case then, the shift in the brane charges can be written as
\be
 \delta_i = -\frac{1}{24}\int_{y_i} c_4(TY)\,.
\ee
Notice that the integrand is no longer the top Chern class of the space over which we integrate, which means that the result is not just the Euler characteristic of $y_i$ but something more complicated that depends on the topology of the fibration.
Similarly, for $y_\pm$ we have
\be
 \delta_\pm = -\frac{1}{24}\int_{y_\pm} c_4(TY)\,.
\ee

\section{\texorpdfstring{$\BP^k$}{Pk} and \texorpdfstring{$\BC^m\times\BP^k$}{Cm x Pk} fibrations}
\label{app:B}
In this section we shortly report on several standard examples of toric fibrations (including compact ones), where the total space of the fibration is not a CY manifold. Although they have no holographic or flux compactification applications, we include these results out of both mathematical curiosity and potential field theory interest in relation to the anomaly polynomial, see \cite{Martelli:2023oqk}.

\subsection*{Hirzebruch surfaces}

We consider the total space of the $\BP^1$ fibration over $\BP^1$ corresponding to the $k$-th Hirzebruch surface $F_k$. This is a toric quotient with charges
\be
Q =
\begin{pmatrix}
1 & 1 & 0 &  -k \\
0 & 0 & 1 & 1
\end{pmatrix}\ ,
\ee
where we regard $t^1$ (corresponding to the first row) as the volume of the base and $t^2$ (corresponding to the second row) that of the fiber.

In terms of the equivariant volume of $\BP^1$, cf.\ \eqref{eq:Cequivolspindle} for $a = b =1$,
\be
 \VVV_{\BP^1}(t, \e_1, \e_2)
 =\frac{\mathe^{-\e_1t}}{\e_2-\e_1}
 + \frac{\mathe^{-\e_2t}}{\e_1-\e_2}\ ,
\ee
we can then write the equivariant volume of the total space as
\be
 \VVV_{F_k} (t, \e)
 = \frac{\mathe^{-\e_1 t^1}}{\e_2-\e_1} \VVV_{\BP^1}(t^2, \e_3,\e_4+k\, \e_1)
 + \frac{\mathe^{-\e_2 t^1}}{\e_1-\e_2} \VVV_{\BP^1}(t^2, \e_3,\e_4+k\, \e_2)\ ,
\ee
where $\e_{1,2}$ are the parameters of the fiber and $\e_{3,4}$ are the parameters on the base. 

In the non-equivariant limit, we can compute the generating function of intersection numbers as
\be
 \VVV_{F_k}(t,\e=0) = t^1 t^2+k\, \frac{(t^2)^2}{2}\ ,
\ee
which corresponds to the compact volume of the $n$-th Hirzebruch surface.

In the $\lam$-parametrization, we also find
\be
    \BV_{F_k} (\lam, \e)= \frac{e^{\lam^2 (\e_2 -\e_1)}}{\e_2 - \e_1}\, \BV_{\BP^1} (\lam^{3,4}, \e_4 - \e_3 + k\, \e_1 ) - \frac{e^{-\lam^1 (\e_2 -\e_1)}}{\e_2 - \e_1}\, \BV_{\BP^1} (\lam^{3,4}, \e_4 - \e_3 + k\, \e_2 )\ .
\ee
We can again define
\be
    \NN := \e_2 - \e_1\ , \qquad \e_\pm := \e_4 - \e_3 + k\, \e_{1,2}\ ,
\ee
together with
\be
    \lam^{3,4}_+ = \lam^{3,4} \mp \frac{\NN}{\nu_+}\, \lam^2\ , \qquad \lam^{3,4}_- = \lam^{3,4} \pm \frac{\NN}{\nu_-}\, \lam^1\ ,
\ee
where the sign $\pm$ is correlated with the index $3, 4$ respectively, such that we arrive at 
\be
\label{eq:MZequivolHirzebruch}
    \BV_{F_k} (\mu, \nu) = \frac{1}{\NN}\, \left( \BV_{\BP^1} (\lam_+^{3,4}, \e_+) - \BV_{\BP^1} (\lam_-^{3,4}, \e_-) \right)\ .
\ee
Notice the close similarity to the general expressions for the CY cases discussed in the main text, \eqref{eq:MZequivolCnoverspindle} and \eqref{eq:MZequivolmesonicoverspindle}. It is however important to note that the underlying variables are not the same and depend on the details of the fiber.

\subsection*{\texorpdfstring{$\BP^m$}{P^m} over \texorpdfstring{$\WPL$}{WPL}}
We can easily generalize to a fibration of $\BP^1$ on the spindle with
\be
Q =
\begin{pmatrix}
a & b & 0 &  -k \\
0 & 0 & 1 & 1
\end{pmatrix}\ ,
\ee
again having the base on the first row. Then the volume becomes
\be
 \VVV_{F^{(a,b)}_k}
 = \frac{\mathe^{-\frac{\e_1 t^1}{a}}}{a\e_2-b\e_1}
 \mathrm{vol}_{\BP^1}\left(t^2, \e_3,\e_4+\frac{k}{a}\e_1\right)
 + \frac{\mathe^{-\frac{\e_2 t^1}{b}}}{b\e_1-a\e_2}
 \mathrm{vol}_{\BP^1}\left(t^2, \e_3,\e_4+\frac{k}{b}\e_2\right)\ ,
\ee
and the resulting $\lam$-parametrization follows analogously as above.

Moreover, we can also generalize the dimension of the fiber and consider $\BP^m$ over the spindle, with a charge matrix given by
\be
Q =
\begin{pmatrix}
a & b & 0 & \dots & 0 &  -k \\
0 & 0 & 1 & \dots & 1 & 1
\end{pmatrix}\ ,
\ee
such that the equivariant volume becomes
\be
\begin{aligned}
 \VVV(t, \e) &= \oint\frac{\mathd\phi_1}{2\pi\mathi}
 \frac{\mathe^{\phi_1t^1}}{(\e_a+a\phi_1)(\e_b+b\phi_1)}
 \left[\oint\frac{\mathd\phi_2}{2\pi\mathi} \frac{\mathe^{\phi_2t^2}}
 {\prod_{i=1}^{m+1}(\e_i-k\, \phi_1+\delta_{i,m+1}\phi_2)}\right] \\
 &= \frac{\mathe^{-\frac{\e_at^1}{a}}}{a\e_b-b\e_a}
 \VVV_{\BP^m}\left(t^2, \e_i+\frac{k}{a}\e_a\delta_{i,m+1}\right)
 + \frac{\mathe^{-\frac{\e_bt^1}{b}}}{b\e_a-a\e_b}
 \VVV_{\BP^m}\left(t^2, \e_i+\frac{k}{b}\e_b\delta_{i,m+1}\right)\ .
\end{aligned}
\ee
Again, it is straightforward to rewrite this expression in the $\lam$-parametrization to generalize \eqref{eq:MZequivolHirzebruch}.

\subsection*{\texorpdfstring{$\BC\times\BP^1$}{CxP1} over \texorpdfstring{$\BP^1$}{P1}}
Let us consider the fiber $\BC \times \BP^1$ over a $\BP^1$ base, giving us an interesting example of a product space fiber with a compact and a non-compact factor. This is the closest example to the situation discussed in the main text in section \ref{sec:6.2} that features a non-toric compact factor in the fiber. 

We can choose the matrix of charges
\be
Q =
\begin{pmatrix}
1 & 1 & 0 &  -n & -p \\
0 & 0 & 1 & 1 & 0
\end{pmatrix}\ ,
\ee
where we regard $t^1$ (corresponding to the first row) as the volume of the base and $t^2$ (corresponding to the second row) that of the $\BP^1$ factor in the fiber. We also assume that $n,p\in\BZ_{\geq0}$, and note that the fifth direction above corresponds to the complex plane factor $\BC$, and we use the notation $\e_5 = \e$ to distinguish it.

 The equivariant volume in the chamber $t^1,t^2>0$ is given by the integral
\be
\begin{aligned}
 \VVV(t, \e) &= \oint\frac{\mathd\phi_1}{2\pi\mathi}\oint\frac{\mathd\phi_2}{2\pi\mathi}
 \frac{\mathe^{\phi_1t^1+\phi_2t^2}}{(\e_1+\phi_1)(\e_2+\phi_2)(\e_3+\phi_2)
 (\e_4-n\phi_1+\phi_2)(\e-p\phi_1)} \\
 &= \frac{\mathe^{-\e_1t^1}}{\e_2-\e_1} \VVV_{\BP^1}(t^2,\e_3,\e_4+n\e_1)
 \VVV_{\BC}(\e+p\e_1) \\
 &\hspace{30pt}+ \frac{\mathe^{-\e_2t^1}}{\e_1-\e_2} \VVV_{\BP^1}(t^2,\e_3,\e_4+n\e_2)
 \VVV_{\BC}(\e+p\e_2)
\end{aligned}
\ee
where we took the poles corresponding to solutions of:
\be
\begin{aligned}
 & (\e_1+\phi_1)(\e_3+\phi_2)=0 \\
 & (\e_2+\phi_1)(\e_3+\phi_2)=0 \\
 & (\e_1+\phi_1)(\e_4-n\phi_1+\phi_2)=0 \\
 & (\e_2+\phi_1)(\e_4-n\phi_1+\phi_2)=0
\end{aligned}
\ee
as suggested by the JK prescription in this chamber. 

It is particularly interesting to observe what happens in the non-equivariant limits in this case. Setting to zero the equivariant parameters $\e_3,\e_4$ associated to the compact directions on the fiber, we find
\be
 \VVV(t, \e, \e_{1,2}, \e_{3,4} = 0) = \frac{1}{\e_2 - \e_1}\, \left( \frac{\mathe^{-t^1 \e_1} (1-\mathe^{-n\, t^2 \e_1})}{n\, \e_1 (\e + p\, \e_1)} - \frac{\mathe^{-t^1 \e_2} (1-\mathe^{-n\, t^2 \e_2})}{n\, \e_2 (\e + p\, \e_2)}\right)\ ,
\ee
such that the equivariant volume has an inverse cubic dependence on the $\e$ parameters. We can also take the further non-equivariant limit of the compact base as well, finding
\be
 \lim_{\e_{1,2,3,4} \to 0} \VVV(t, \e)
 = \frac{p\, t^2}{\e^2} + \frac{t^2 (2 t^1 + n t^2)}{2 \e}\ .
\ee

In the $\lam$-parametrization, we instead find
\be
\begin{aligned}
 \BV =& \frac{\mathe^{\lam^2 (\e_2 -\e_1)}}{\e_2 - \e_1}\, \BV_{\BP^1} (\lam^{3,4}, \e_4 - \e_3 + n\, \e_1 )\, \BV_{\BC} (\lam,\e + p\, \e_1) \\
 &- \frac{\mathe^{-\lam^1 (\e_2 -\e_1)}}{\e_2 - \e_1}\, \BV_{\BP^1} (\lam^{3,4}, \e_4 - \e_3 + n\, \e_2 )\, \BV_{\BC} (\lam,\e + p\, \e_2)\ ,
\end{aligned}
\ee
It is again straightforward to find the parametrization that brings the equivariant volume in the form analogous to \eqref{eq:MZequivolHirzebruch}.

The generalization of these results to a fibration of $\BC^m \times \BP^k$ over $\WPL$ follows the same pattern by adjusting the corresponding number of variables.

\bibliographystyle{JHEP}
\bibliography{refs.bib}

\end{document}